\documentclass[%
 reprint,
 onecolumn,
 superscriptaddress,
 amsmath,
 amssymb,
 aps,
]{revtex4-2}

\usepackage{graphicx}
\usepackage{dcolumn}
\usepackage{bm}
\usepackage{amsfonts, amssymb, amsmath, amsthm, mathrsfs, braket, xcolor, longtable, float}
\usepackage[bottom]{footmisc}
\usepackage[section]{placeins}
\usepackage{algcompatible}
\usepackage{natbib}

\newcommand{\R}{\mathbb{R}}
\newcommand{\C}{\mathbb{C}}
\newcommand{\N}{\mathbb{N}}

\newcommand{\real}{\operatorname{Re}}
\newcommand{\imag}{\operatorname{Im}}

\newcommand{\sumtwo}{\operatorname*{\sum\sum}}

\newcommand{\modu}{\operatorname{mod}}

\newcommand{\norm}[1]{\left\lVert#1\right\rVert}
\usepackage{tikz}
\usepackage[linesnumbered,ruled,vlined]{algorithm2e}

\SetCommentSty{mycommfont}

\SetKwInput{KwInput}{Input}
\SetKwInput{KwOutput}{Output}
\usetikzlibrary{quantikz}
\theoremstyle{definition}
\newtheorem{definition}{Definition}[section]

\newtheorem{theorem}{Theorem}[section]
\allowdisplaybreaks[4]

\usepackage{longtable}

\newtheorem{claim}[theorem]{Claim}
\newtheorem{proposition}[theorem]{Proposition}
\newtheorem{corollary}[theorem]{Corollary}

\def\bra#1{\mathinner{\langle{#1}|}}
\def\ket#1{\mathinner{|{#1}\rangle}}
\newcommand{\bket}[2]{\langle #1|#2\rangle}

\newtheorem{lemma}{Lemma}[section]
\theoremstyle{definition}

\renewcommand{\part}[2]{\frac{\partial #1}{\partial #2}}

\newcommand{\all}[2]{\begin{align}\label{#2} #1\end{align}}
\newcommand{\al}[1]{\begin{align} #1\end{align}}

\newcommand{\thmref}[1]{\hyperref[#1]{{Theorem~\ref*{#1}}}}
\newcommand{\lemref}[1]{\hyperref[#1]{{Lemma ~\ref*{#1}}}}
\newcommand{\remref}[1]{\hyperref[#1]{{Remark~\ref*{#1}}}}
\newcommand{\corref}[1]{\hyperref[#1]{{Corollary~\ref*{#1}}}}
\newcommand{\eqnref}[1]{\hyperref[#1]{{Equation~(\ref*{#1})}}}
\newcommand{\claimref}[1]{\hyperref[#1]{{Claim~\ref*{#1}}}}
\newcommand{\remarkref}[1]{\hyperref[#1]{{Remark~\ref*{#1}}}}
\newcommand{\propref}[1]{\hyperref[#1]{{Proposition~\ref*{#1}}}}
\newcommand{\factref}[1]{\hyperref[#1]{{Fact~\ref*{#1}}}}
\newcommand{\defref}[1]{\hyperref[#1]{{Definition~\ref*{#1}}}}
\newcommand{\exampleref}[1]{\hyperref[#1]{{Example~\ref*{#1}}}}
\newcommand{\hypref}[1]{\hyperref[#1]{{Hypothesis~\ref*{#1}}}}
\newcommand{\secref}[1]{\hyperref[#1]{{Section~\ref*{#1}}}}
\newcommand{\chapref}[1]{\hyperref[#1]{{Chapter~\ref*{#1}}}}
\newcommand{\apref}[1]{\hyperref[#1]{{Appendix~\ref*{#1}}}}



\begin{document}

\title[Quantum loading of probability distributions]{Efficient quantum loading of probability distributions through Feynman propagators}

\author{Elie Alhajjar}
\email{\texttt{ealhajjar@rand.org}}
\affiliation{%
RAND Corporation, Engineering and Applied Sciences, Arlington, VA, USA}

\author{Jesse Geneson}
\email{\texttt{jgeneson@rand.org}}
\affiliation{%
RAND Corporation, Engineering and Applied Sciences, Pittsburgh, PA, USA}

\author{Anupam Prakash}
\email{\texttt{anupam.prakash@qcware.com}}
\affiliation{%
QCWare, Palo Alto, CA, USA}

\author{Nicolas Robles}
\email{\texttt{nrobles@rand.org}}
\affiliation{%
RAND Corporation, Engineering and Applied Sciences, Arlington, VA, USA}

\date{\today}

\begin{abstract}
We present quantum algorithms for the loading of probability distributions using Hamiltonian simulation for one dimensional Hamiltonians of the form ${\hat H}= \Delta + V(x) \mathbb{I}$. We consider the potentials $V(x)$ for which the Feynman propagator is known to have an analytically closed form and utilize these Hamiltonians to load probability distributions including the normal, Laplace and Maxwell-Boltzmann into quantum states. We also propose a variational method for probability distribution loading based on constructing a coarse approximation to the distribution in the form of a `ladder state' and then projecting onto the ground state of a Hamiltonian chosen to have the desired probability distribution as ground state. These methods extend the suite of techniques available for the loading of probability distributions, and are more efficient than general purpose data loading methods used in quantum machine learning. 
\end{abstract}

\pacs{Valid PACS appear here} 
\maketitle

\tableofcontents

\section{Introduction} \label{sec:intro}
 \label{sec:intro01} 
The path integral formulation of quantum mechanics is closely related to the theory of stochastic processes through the Feynman Kac formula. This formula offers a method of expressing the solutions of a parabolic second order differential equation (heat and diffusion equations being the prototypical examples) as a conditional expectation over paths of stochastic processes induced by Brownian motion. The Schr\"{o}dinger equation can be viewed as a diffusion equation in an imaginary time variable. Indeed, the Wick rotated Schr\"{o}dinger equation corresponds to a classical diffusion problem. For instance, in the recent work \cite{feynmankac}, the variational quantum imaginary time evolution (VarQITE) algorithm was developed to solve the Wick-rotated Schr\"{o}dinger equation numerically, see also \cite{accelerated, xanaduOption, quantumPDEs, protocols, pricingmultiasset} for additional recent and relevant literature. 

In this manuscript, we explore the application of path integral-based methods to quantum computing for a simpler problem, that of creating quantum states representing specific probability distributions. A large collection of problems in quantum computing require the preparation of quantum states that represent probability distributions. The loading of distributions is particularly useful as a building block for quantum Monte Carlo methods that can be used to estimate expectations and other quantities of interest for data drawn from the said distribution. Quantum Monte Carlo methods use variations of Grover's algorithm such as Quantum Amplitude Estimation (QAE) and can achieve a quadratic speed-up over classical Monte Carlo methods over a wide range of estimation problems \cite{BHMT00, M15}. A concrete example of a probability distribution loading problem is the loading of the normal or log-normal distributions in the amplitudes $\alpha_i$ of a quantum state starting from the initial state $\ket{0}^n$, this has attracted considerable interest due to its potential applications to quantitative finance. 

In other words, there are many quantum computing instances that require an efficient algorithm $\mathcal{A}$ such that
\[
\mathcal{A} : \ket{0}^n \to \sum_{i=0}^{2^n-1} \mathfrak{a}_i \ket{i}, 
\]
where $|\mathfrak{a}_i|^2 \sim \mathcal{N}(0,1)$ or $|\mathfrak{a}_i|^2 \sim \mathcal{L}(0,1)$. Here $\mathcal{N}(\mu,\sigma^2)$ is the normal distribution and $\mathcal{L}(\mu,\sigma^2)$ is the lognormal distribution. The necessity of loading the normal distribution is well-understood and its usefulness does not require further motivation. However, encodings where $|\mathfrak{a}_i| \sim \mathcal{L}(0,1)$ are also useful for other applications outside finance such as incubation period of diseases \cite{lognormal1} and neuron densities in the cerebral cortex \cite{lognormal2}. We shall discuss the possibility of loading the log-normal distribution in Section \ref{sec:inverselist}.

The quantum loading of probability distributions is different from the data loading problem commonly considered in quantum machine learning (QML), where the problem is to generate quantum states representing arbitrary vectors $x \in \R^{n}$. 
The Grover-Rudolph method and its variants provide a solution with gate complexity $O(n \log n)$ and with $O(\log n)$ qubits in general and can be more efficient only in cases where the  probability distribution that can be integrated efficiently \cite{grover2002creating}. Recent literature on probability loading on quantum devices can be found in \cite{wellsfargo1, wellsfargo2, capgemini, cprv23, dasguptaPaine}.

In contrast to the data loading problem, the probability distribution loading problem aims to load a specific vector that corresponds to the density function for a given classical probability distribution. Given the highly structured nature of the problem and the absence of $O(n)$ free parameters, the gate complexity can be expected to be polynomial rather than exponential in the number of qubits used. 
The ideal solution will have gate complexity $O(\text{polylog}(n))$, that is polynomial in the representation size. 
A well-studied problem in the literature is that of loading the \textit{normal} distributions, where efficient circuits with gate complexity $O(\text{polylog}(n))$ are known as we discuss in more detail below. 

The normal distribution loading problem also provides the prototype for the relation between path integrals and the distribution loading problem. It is well-known that the standard normal distribution is the ground state for the quantum harmonic oscillator with Hamltonian $\hat H= \Delta + x^{2}\mathbb{I}$, with $\Delta = \frac{ \partial}{ \partial x^{2}}$ being the Laplacian for the one dimensional path (in general we shall drop the identity operator $\mathbb{I}$ for brevity). The quantum harmonic oscillator is also a basic example of a system where the path integral can be computed explicitly to yield a closed form expression for the Feynman propagator, that is, the transition amplitudes $\bra{x} e^{-\hat Ht} \ket{x}$ have a closed form expression known as Mehler's formula. Hamiltonian simulation for $\hat H$ therefore provides a method for loading the normal distribution, the initial state is taken to be a suitable coarse approximation to the normal distribution that can be prepared easily and the loading algorithm performs time evolution for the projection onto the ground state. 

In this work, we illustrate these ideas by going beyond the traditional loading of the Gaussian distribution, and we explore other continuous distributions that have real as well as positive support such as the Maxwell and Laplace distributions. Instances of such efficient loading have not been considered in the literature and we provide some novel answers to these questions. 

The first approach is to start with a suitable initial state and perform time evolution for  one dimensional Hamiltonians of form $\hat H= \Delta + V(x)$. The potential $V(x)$ is chosen so that the Feynman propagator is known to have a closed form analytic formula. The Hamiltonians of this form are rather special and have been collected from the literature on exactly path integrable one dimensional systems \cite{handbookPI}. The choice of the initial state is not canonical, however the closed form expressions for the propagators allows us to evaluate the resulting final distribution in closed form for many cases. 

The potentials considered include the harmonic oscillator, radial, linear, semi-linear, Dirac $\delta$, and Coulomb potentials. The probability distributions we are able to generate using the closed from Feynman propagators include the Normal, Maxwell-Boltzmann, and Laplace distributions. Moreover, other distributions which are made of convolutions of normal and $\chi$ distributions can also be generated. 
Due to the many links between the theory of special functions and probability distributions it is likely that more probability distributions could be loaded efficiently using similar ideas and by varying the initial states, this paper can be viewed as a first step in this direction. 

The Hamiltonians with closed form propagators that we investigate in this paper and the potentials and probability distributions corresponding to them are summarized in the table below. One salient feature of this table is that the resulting distributions are not always known for the chosen initial states, which ties up with the mathematical difficulty of elucidating the exact nature of the resulting distribution. However, for most of the cases where the distribution is either unknown or unnamed in the standard literature, we are able to provide a closed analytic formula for the distribution, sometimes in the form of an infinite series or as an expression involving special functions. 

\begin{center}
\begin{table}[h]
\begin{tabular}{|c|c|c|c|c|c|c|}
\hline
\textbf{Potential} & \textbf{Range} & \textbf{Resulting distribution} \\
\hline
$V(x)=0$ & $x \in \R$ & Weighted sum of $\chi$'s \\
\hline
$V(x)=0$ & $x >0$ & Maxwell-Boltzman \\
\hline
$V(x)=0$ & $-b < x < b$ & Weighted sum of wrapped $\mathcal{N}$ \\
\hline
$V(x)=\frac{1}{2}m\omega^2x^2$ & $x \in \R$ & Normal $\mathcal{N}$ \\
\hline
$V(x)=\frac{\lambda^2-\frac{1}{4}}{2mx^2}$ & $x >0$ & Unknown \\
\hline
$V(x)=\frac{\lambda^2-\frac{1}{4}}{2mx^2}$ & $x >0$ & Unknown \\
\hline
$V(x) = \{kx,\infty\}$ & $x >0$ & Unknown \\
\hline
$V(x)=kx$ & $x\in \R$ & Unknown \\
\hline
$V(x)=\{0, \infty\}$ & $0< x < a$ & Unknown \\
\hline
$V(x)=\frac{1}{\sin^2 x}$ & $x \in [0,\frac{\pi}{2}]$ & Unknown \\
\hline
$V(x)=\frac{e}{x}$ & $x \in \R^+$ & Unknown \\
\hline
$V(x)=-a\delta(x)$ & $x \in \R$ & Laplace $\mathfrak{L}$ \\
\hline
\end{tabular}
\caption{Potentials with closed form Feynman propagators and probability distributions associated with them. The full table with the resulting amplitudes as well as the initial state details can be found in Table \ref{table:fulltable}.}
\end{table}
\label{table:partialtable}
\end{center}

A second method that we propose for probability distribution loading is the construction of Hamiltonians having the given probability distribution as a ground state. We provide such a Hamiltonian for the log-normal distribution as an example in section III. This construction suggests a variational method for the probability distribution loading, that is, start with a coarse approximation to the probability density function and then project onto the ground state of the Hamiltonian. In general, there is no reason to expect that the Hamiltonians corresponding to a general probability distribution like the log-normal will provide substantial quantum speedups over classical sampling, however this still offers a reasonably efficient way of loading probability distributions with broader applicability using the well-known Trotter based or variational methods for Hamiltonian simulation. 

Coarse approximations of states representing probability distributions are easy to prepare. We shall provide constructions of `ladder' states that discretize the probability distribution into a certain number of bins and can be used as the initial state for both the Hamiltonian evolution and projection to the ground state methods. Experimental results on Hamiltonian evolution starting with such ladder states and other suitable initial states are provided in Section \ref{sec:list}. 

Besides the application to the probability distribution loading problem, there is also a computatational motivation for studying the quantum systems 
with closed form Feynman propagators. The quantum harmonic oscillator, which is the prototypical example of Hamiltonians in this family, is known to be fast-forwardable, that is, in order to simulate $e^{i \hat H t}$ for time $t=T$, the gate complexity has poly-logarithmic dependence on $T$, see \cite{atia2017fast}. The multivariate harmonic oscillator is also known to be fast-forwardable as it corresponds to a quadratic bosonic Hamiltonian \cite{gu2021fast}. It is of considerable interest to find other families of fast-forwardable Hamiltonians \cite{su2021fast}. As the underlying quantum algorithm for probability distribution loading relies on Hamiltonian simulation starting with a suitable initial state, fast-forwardability of the Hamiltonian is the source of the quantum speedups for the distribution loading problem. 

Determining the potentials that lead to closed form propagators is an open problem in physics \cite{susyQM} and determining the class of fast-forwardable Hamiltonians is an open problem in the theory of quantum computation \cite{su2021fast}. Both of these questions have received considerable attention in the literature. This paper raises the question of the overlap between these classes, namely if the Hamiltonians known to have closed form Feynman propagators are also fast forwardable. The answer is known to be affirmative for the quantum harmonic oscillator and the linear potential, it however remains unknown for the other cases. 

\subsection{Previous literature} 

One of the first applications of path integral based methods to quantum computing was the numerical solution of the Wick-rotated Schr\"{o}dinger equation with the VarQITE algorithm \cite{feynmankac}. This resulting partial differential equation is known as the Feynman-Kac formula. The specific details led to an anisotropic second order partial differential equation 
\begin{align}
    \frac{\partial u}{\partial t} = \mathcal{G} u, \quad t>0
\end{align}
where the differential operator $\mathcal{G}$ is given by
\begin{align}
    \mathcal{G} = \rho \frac{\partial^2}{\partial x \partial y} + \frac{1}{2} \frac{\partial^2}{\partial x^2} + \frac{1}{2} \frac{\partial^2}{\partial y^2} = \mathcal{G}^\dagger.
\end{align}
The operator $\mathcal{G}$ can be viewed as a Hamiltonian for a two dimensional system and the Feynman-Kac formula gives a closed form expression for the solution $u=u(x,y,t)$. The solution $u(x, y, t)$ is in fact the conditional expectation 
\begin{align}
    u(x,y,t) = \mathbb{E}[\psi(X_t,Y_t) \;|\; X_t = x, Y_t = y],
\end{align}
where $X_t$ and $Y_t$ are stochastic processes given by
\begin{align}
    dX_t = dW_t^1, \quad dY_t = dW_t^2,
\end{align}
and where $\rho$ is the correlation between the Brownian motions $X_t$ and $Y_t$. The initial condition $\psi$ is given by
\begin{align}
    u(x,y,0) = \delta(x-x_0) \delta(y-y_0),
\end{align}
where $\delta$ is the Dirac delta function. The solution $u$ is given by
\begin{align}
u(x,y,t) = \frac{1}{2\pi t \sqrt{1-\rho^2}} \exp \bigg[-\frac{1}{t(1-\rho^2)}\bigg(\rho(x_0-x)(y-y_0) + \frac{1}{2}(x-x_0)^2 + \frac{1}{2}(y-y_0)^2 \bigg)\bigg].
\end{align}
The function $u(x,y,t)$ is $\ell_1$ normalized and solving this PDE with VarQITE necessitated the introduction of a proxy norm to keep the solution $\ell_2$ normalized. This caused the VarQITE solutions to overestimate the norm of $u$ thereby somewhat sacrificing accuracy and eroding the efficiency of the algorithm. The VarQITE algorithm is an example of a setting where Hamiltonian evolution for a system with closed form Feynman propagators is explicitly used in a quantum algorithm. 

The probability distribution loading problem that we consider here has many applications, one of the most promising is for quantum algorithms in mathematical finance. Unlike classical Monte Carlo simulation, quantum algorithms can utilize a superposition of all possible paths of a stochastic process and amplitude estimation is then used to estimate path dependent functionals of the process with potentially quadratic speedups over classical Monte Carlo methods. A prototypical example from quantitative finance arises in the context of pricing options and financial derivatives whose underlying model follows a geometric Brownian motion. 
Recent work \cite{chakrabarti2021threshold} on the complexity of option pricing for this model using quantum Monte Carlo methods uses a digital encoding where the trajectory of a stochastic process is stored in quantum registers. Generating a single step of the trajectory is equivalent to loading a normal distribution. The loading of the normal distribution is then performed by emulating an efficient classical sampler as a quantum circuit. The advantage of the approach is that it allows for Monte Carlo estimation of general functions of the trajectories of the stochastic process. However, with this approach, there is no quantum speedup over for the distribution loading problem and the resources needed for loading the multivariate Gaussian distributions are a large contribution to the overall complexity of the quantum Monte Carlo method. 

Thus, procedures for generating probability distributions more efficiently using simple quantum circuits rather than emulating classical samplers could lead to substantial savings for several application domains. As argued in \cite{feynmankac}, one could in principle realize this objective for some probability distributions efficiently, but a systematic treatment would demand careful analysis. 

The normal distribution offers an interesting test case for quantum methods for loading probability distributions. 
Other recent work \cite{rattew2021efficient} uses a quantum version of the classical Galton board to generate normal distributions. The gate complexity for the resulting iterative algorithm is poly-logarithmic in the resolution and the circuit depth is polynomial in the number of qubits. A quantum algorithm with gate complexity of $O(n + \beta^{1/4} (\beta + 1/\epsilon)) $ for creating a quantum representation of the normal distribution $e^{-\beta x^{2}/4}$ was recently reported in \cite{mcardle2022quantum}. The method uses the quantum singular value transformation and is applicable more generally to functions with sparse approximations in the Fourier domain. 
Explicit (as opposed to the asymptotic results discussed above) efficient quantum circuits for loading the normal distribution on a small number of qubits are given in the recent references \cite{wellsfargo1, wellsfargo2}. A new method for loading normal distributions using tensor networks along with experimental results on a 20 qubit machine are discussed in \cite{iaconis2023quantum}.

\subsection{Underlying quantum mechanical formalism} \label{sec:intro02}
We now propose to study the problem of loading a desired -- ideally arbitrary -- probability distribution function through the use of the two most common quantum mechanical tools, namely the Schr\"{o}dinger equation and the Feynman path integral. The fact that quantum mechanics is, under the accepted orthodox interpretations, a \textit{probabilistic} theory of nature has not been leveraged sufficiently in the literature of probability loading in the context of quantum computing. We shall now describe the rationale of the idea. In one dimension, the time-dependent Schr\"{o}dinger equation (TDSE) reads
\begin{align} \label{eq:TDSE1D}
    i\hbar \frac{\partial \Psi(x,t)}{\partial t} = \hat H \Psi(x,t) = - \frac{\hbar^2}{2m} \frac{\partial^2 \Psi(x,t)}{\partial x^2} + V \Psi(x,t),
\end{align}
where the potential $V$ is, for the most part, a function of space only, i.e., $V=V(x)$. The necessary Cauchy data to solve this partial differential equation is typically the properly normalized \textit{initial state} $t=t_0=0$, i.e. $\Psi(x,0)$ with $\int_{\R} |\Psi(x,0)|^2dx = 1$. The wave function $\Psi(x,t)$ is usually a complex-valued function. The solution $\Psi(x,t)$ to \eqref{eq:TDSE1D} is deterministic, in the sense that if we know $\Psi$ at $t=t_0$, then we know $\Psi$ for all $t > t_0$. In other words, at a given time $t$, the physical state of the quantum mechanical system is \textit{completely} described by the wave function $\Psi(x,t)$. On the other hand, the Born interpretation of the quantity $|\Psi|^2$ is that of the probability of an event. This event corresponds to finding the particle associated to $\Psi$ at a given time $t$ and at given spatial interval $x \in [a,b]$. In a standard quantum mechanics course, one proves very early on that $|\Psi(x,t)|^2$ integrates to $1$ over $\R$ for a given $t \ge 0$ and, naturally, $|\Psi(x,t)|^2 \ge 0$, thereby making $|\Psi(x,t)|^2$ a probability distribution function (PDF) associated to the potential $V(x)$, provided that the Hamiltonian $\hat H$ is Hermitian, i.e. $\hat H = \hat H^\dagger$. Probability conservation as $t$ evolves follows by imposing mild vanishing boundary conditions on $\Psi(x,t)$ as $x \to \pm \infty$, see e.g. \cite[$\mathsection$ 1.4]{griffiths}. This should already hint at the underlying idea we shall exploit: different potentials can lead to different PDFs and if we could easily load an initial state $\Psi(x,0)$ and efficiently solve \eqref{eq:TDSE1D} with a quantum algorithm, then we shall have successfully loaded the probability distribution $|\Psi(x,t)|^2$ in our quantum computer with a quantifiable and advantageous speedup.

It is also worth remarking that upon analytic continuation $t \to -it$ \eqref{eq:TDSE1D} becomes the heat or diffusion equation
\begin{align} \label{eq:heat1D}
    \frac{\partial u(x,t)}{\partial t} = \alpha \frac{\partial^2 u(x,t)}{\partial x^2}  + g(x) u(x,t).
\end{align}
One usually solves this with the boundary condition $u(x,0)=u_0(x)$. This is at the heart of the Feynman-Kac formula and the intimate connection between second order parabolic differential equations (PDEs) and stochastic differential equations (SDEs), see for instance \cite{feynmankac} in the context of quantum computing algorithms. For example if we take $g(x)=0$ and $u_0(x) = \delta(x-x_0)$, the Dirac delta function, then the solution $u(x,t)$ corresponds to the so-called heat kernel
\begin{align} \label{eq:heatkernel}
u(x,t) = \frac{1}{\sqrt{4 \pi \alpha t}} \exp\bigg(-\frac{(x-x_0)^2}{4\alpha t}\bigg) >0
\end{align}
which integrates to $1$ for all $t>0$. In fact $u(x,t) \sim \mathcal{N}(x_0, 2\alpha t)$. However, the drawback is that $u(x,t)$ corresponds to a PDF, rather than having $|u(x,t)|^2$ correspond to a PDF. In other words, $u(x,t)$ for $t>0$ is $\ell_1$-normalizable, rather than $\ell_2$-normalizable, and hence Euclidean or imaginary-time evolution is not as susceptible to being quantum native as real-time evolution. We shall elaborate on this in Section \ref{sec:intro04}.

If we know the state $\Psi(x,t_0)$ at $t_0$, then the problem of quantum mechanics is to find the state of the system at an arbitrarily final time $t>t_0$. The general solution to \eqref{eq:TDSE1D} is 
\begin{align} \label{eq:UPsi}
    \Psi(x,t) = U(t,t_0) \Psi(x,0)
\end{align}
where $U$ denotes the unitary time evolution operator satisfying
\begin{align} \label{eq:PDEU}
    i\hbar \frac{\partial}{\partial t} U(t,t_0) = \hat H U(t,t_0) \quad \textnormal{with} \quad U(t_0,t_0)=\mathbf{1}.
\end{align}
If the potential (and hence the Hamiltonian) is time-independent, then the solution to \eqref{eq:PDEU} is given by the deceivingly simple-looking formula
\begin{align} \label{eq:TrotterUfunction}
    U(t,t_0) = \exp\bigg(-\frac{i}{\hbar} \hat H (t-t_0)\bigg).
\end{align}
One can easily check the composition law $U(t,t_0)=U(t,t_m)U(t_m,t_0)$ for arbitrary times $t_0, t_m$ and $t$. In the context of quantum algorithms, it is much more convenient to work in the position space. Typically, one considers the position operator $\hat x$ and its eigenvalues $x$ which satisfy $\hat x \ket{x} = x \ket{x}$. The spectrum is continuous as $x \in \R$. The completeness relation $\int_{\R} dx \ket{x}\bra{x} = \mathbf{1}$ allows us to write
\begin{align} 
    \ket{\Psi(t)} = \int_{\R} dx \ket{x}\braket{x | \Psi(t)} = \int_{\R} dx \Psi(x,t) \ket{x} \quad \textnormal{where} \quad \Psi(x,t) = \braket{x | \Psi(t)}  = \braket{\Psi(t)|x}^*.
\end{align}
Equipped with these tools, one can derive the time-evolution equation for \eqref{eq:TDSE1D} as
\begin{align} \label{eq:PsiK}
    \Psi(x,t) = \int_{\R} K(x,y;t,t_0) \Psi(y,t_0)dy 
\end{align}
where the kernel $K$ is given by
\begin{align}  \label{eq:TrotterUbraket}
    K(x,y;t,t_0) = \braket{x|U(t,t_0)|y} \Theta(t-t_0) = \bra{x} \exp\bigg(-\frac{i{\hat H}(t-t_0)}{\hbar}\bigg) \ket{y}\Theta(t-t_0).
\end{align}
The function $\Theta(t)$ is the Heaviside step function, which is $1$ if $t \ge 0$ and $0$ if $t<0$. Sometimes $K$ is called the propagator or the kernel. Often times one sees the notation $T=t - t_0$ in the literature and since we shall always take $t_0=0$, one can think of $T$ as $T=t$. 

Using efficient or adequate quantum algorithms to simulate \eqref{eq:TrotterUfunction} such as Trotterization or variational quantum real-time evolution (VarQRTE) will allow us to solve the Schr\"{o}dinger equation on a quantum device, provided we are first able to load a suitable $\Psi(x,0)$. 

Because $\Theta'(T) = \delta(T)$, then we have that $K$ satisfies the inhomogeneous Schr\"{o}dinger equation 
\begin{align}
    \bigg(i \hbar \frac{\partial}{\partial t} - \hat H \bigg)K(x,y;t;t_0) = i \hbar \delta(x-y)\delta(t-t_0),
\end{align}
with initial condition
\begin{align}
    \mathop {\lim }\limits_{t-t_0 \to 0^+} K(x,y;t,t_0) = \delta (x-y).
\end{align}
Through the path integral approach, one sees that $K$ is the probability amplitude for a particle to go from $(y,t_0)$ to $(x,t)$. This amplitude is the sum of contributions from \textit{all} paths that the particle can take and it can be shown that
\begin{align} \label{eq:FPI}
    K(x'',x';t'',t') &= \mathop {\lim }\limits_{N \to \infty} \bigg(\frac{m}{2 \pi i \varepsilon N}\bigg)^{N/2} \prod_{k=1}^{N-1}\int_{\R} dx_k \exp \bigg[ \frac{i}{\hbar} \sum_{j=0}^{N-1} \bigg(\frac{m}{2\varepsilon}(x_{j+1}-x_j)^2 - \varepsilon V(x_j)\bigg)\bigg]\nonumber \\
    &=: \int_{x(t')=x'}^{x(t'')=x''} \mathcal{D}x(t)) \exp\bigg[\frac{i}{\hbar} \int_{t'}^{t''}  \bigg(\frac{1}{2}m{\dot x}^2 - V(x) \bigg) dt\bigg].
\end{align}
Here $\varepsilon$ denotes the lattice constant $\varepsilon := (t-t_0)/N > 0$ and it plays the role of $dt$, see e.g. \cite{handbookPI, kleinertPI}. 

The classical quantity $S$ is known as the \textit{action} and it is given by the integral over time of the classical \textit{Lagrangian} function $\mathcal{L}=\frac{1}{2}m{\dot x}^2-V(x)$, in other words
\begin{align}
    S = S[x(t)] = \int_{t'}^{t''} \bigg(\frac{1}{2}m{\dot x}^2 - V(x) \bigg) dt  = \int_{t'}^{t''} \mathcal{L}(x(t), {\dot x}(t))dt .
\end{align}
Feynman's formulation of quantum mechanics in terms of path integrals is usually called the `Lagrangian formulation of quantum mechanics', as opposed to the standard Hamiltonian approach given by the Schr\"{o}dinger equation \eqref{eq:TDSE1D}. 

Alternatively, there might be instances when computing $K$ through a path integral approach is not the optimal route. If the potential $V$ is not a function of time $t$ and $\hat H$ has a discrete spectrum, then the well-known technique of separation of variables allows us to write the wave function as the infinite series
\begin{align} \label{eq:PsixtPsin}
\Psi(x,t) = \sum_{n=1}^\infty c_n \Psi_n(x) e^{-iE_nt/\hbar}
\end{align}
where the eigenfunctions, or steady states, $\Psi_n$ and the eigenvalues $E_n$ are obtained by solving the time-independent Schr\"{o}dinger equation (TISE)
\begin{align} \label{eq:TISE}
-\frac{\hbar^2}{2m} \frac{d^2 \Psi_n(x)}{dx^2} + V(x) \Psi_n(x) = E_n \Psi_n(x),
\end{align}
and the constants $c_n$ are given by the general linear combination of solutions
\begin{align} \label{eq:Psix0cn}
\Psi(x,0) = \sum_{n=1}^\infty c_n \Psi_n(x).
\end{align}
The interpretation of $c_n$ is that $|c_n|^2$ is the probability that measurement of the energy would return the value $E_n$. The connection between $K(x,y;t,t_0)$ and $\Psi_n$ and $E_n$ is given by
\begin{align} \label{eq:KPsinEn}
K(x,y;t,t_0) = \sum_{n=0}^\infty \Psi_n(x) \Psi_n^*(y) e^{-iE_n (t-t_0) / \hbar}.
\end{align}
Therefore, knowing $\Psi_n$ and $E_n$ is, in principle, nearly as useful as knowing the path integral $K$, specially if we can perform the sum over $n$ in \eqref{eq:KPsinEn}. This summation usually requires knowledge of polynomials associated to the eigenfunctions $\Psi_n$ as well as Mehler or Hardy-Hille type of formulas, see Sections \ref{sec:intro03} and \ref{sec:list04}.

Let us summarize these ideas as follows.
\begin{enumerate}
    \item[(1)] If we are given a valid initial state $\Psi(y,0)$,
    \item[(2)] and if we know, or are able derive using \eqref{eq:TrotterUbraket} or \eqref{eq:FPI}, the propagator $K(x,y;t,t_0)$, 
\end{enumerate}
then we will have full knowledge of $\Psi(x,t)$ for $t>0$ as per \eqref{eq:PsiK}. Moreover, $|\Psi(x,t)|^2 = f(x)$ corresponds to a PDF. More specifically, $|\Psi(x,t)|^2$ corresponds to the PDF of the propagator $K$ given its initial ground state $\Psi(x,t_0)$. We see from \eqref{eq:TrotterUbraket} or \eqref{eq:FPI} that what characterises one propagator $K$ from another is their associated potential $V$. Therefore, different potentials $V$ can lead to different propagators $K$ and hence to different PDFs. In other words, for a given potential $V(x)$ and initial $\Psi(x,0)$ we have a continuous random variable $X$ such that
\begin{align}
    \Psi(x,t) &= \int_{\R} K(x,y,t,0) \Psi(y,0) dy \nonumber \\
   |\Psi(x,t)|^2 & \sim X \quad \textnormal{with} \quad f_X(x) \ge 0 \quad \textnormal{and} \quad \int_{\R} f_X(x) dx = 1,
\end{align}
effectively making $f_X(x)$ the continuous PDF of $X$. On the other hand, there are several obstacles to overcome. Continuing from our list above, these are as follows.
\begin{enumerate}
    \item[(3)] One needs to know ahead of time which potential $V(x)$ and which initial state $\Psi(x,0)$ are going to lead to a desired, or targeted, $f(x)$. This amounts to solving a differential equation or, equivalently, computing a path integral. However, this is a purely analytical problem and there exists ample literature on how to proceed.
    \item[(4)] The initial state $\Psi(x,0)$ must not be too difficult or too taxing in the loading process while still retaining desirable properties that lead to a suitable PDF $f(x)$.
    \item[(5)] The quantum algorithm employed for evolving the Schr\"{o}dinger equation must be efficient or at least efficient enough not to add overhead to the overall industry-application mission. For instance, if we wish to perform Monte Carlo simulations on the PDF $f(x)$, we usually expect a quadratic speedup over classical methods by the use of QAE. If the loading of $f(x)$ is faster than quadratic, then the loading algorithm will be successful. Both Trotterization and VarQRTE satisfy this requirement.
\end{enumerate}
We shall illustrate this strategy with a specific example.
\subsection{Example: the normal distribution and the harmonic oscillator} \label{sec:intro03}
Our aim is to load the normal distribution. Suppose we consider the harmonic oscillator, i.e. $V(x)=\frac{1}{2}m\omega^2 x^2$ with $\omega >0$ and $x \in \R$. It is a routine exercise to show that
\begin{align} \label{eq:HermitePsin}
\Psi_n(x) = \bigg(\frac{m \omega}{2^{2n}\pi \hbar (n!)^2}\bigg)^{1/4} H_n \bigg(\sqrt{\frac{m \omega}{\hbar}}x\bigg) \exp\bigg(-\frac{m\omega}{2\hbar}x^2\bigg) \quad \textnormal{and} \quad E_n = \hbar \omega (n + \tfrac{1}{2}),
\end{align}
where $H_n$ are the Hermite polynomials
\begin{align} \label{eq:HermitePolDef}
    H_n(x) = (-1)^n e^{x^2} \frac{d^n}{dx^n}e^{-x^2}.
\end{align}
By the use of the Mehler formula \cite{handbookPI, gradryz}
\begin{align} \label{eq:Mehler}
e^{-(x^2+y^2)/2} \sum_{n=0}^\infty \frac{1}{n!} \bigg(\frac{z}{2}\bigg)^n H_n(x) H_n(y) = \frac{1}{(1-z^2)^{1/2}} \exp\bigg[ \frac{4xyz - (x^2+y^2)(1+z^2)}{2(1-z^2)} \bigg]
\end{align}
one ends up showing that the propagator $K$ of the harmonic oscillator is given by
\begin{align} \label{eq:KHO}
K(x'',x';t'',t') &= \int_{x(t')=x'}^{x(t'')=x''} \mathcal{D}x(t) \exp \bigg[\frac{im}{2\hbar} \int_{t'}^{t''} (\dot x^2 - \omega^2 x^2) dt\bigg] \nonumber \\
&= \bigg(\frac{m \omega}{2 \pi i \hbar \sin(\omega T)}\bigg)^{1/2} \exp \bigg[ \frac{i m\omega}{2 \hbar \sin (\omega T)} \{ (x''^2+x'^2) \cos \omega T - 2x''x' \} \bigg] \quad \textnormal{with} \quad T=t''-t'.
\end{align}
As we have noted previously, the crucial point is that if the initial state $\Psi(x,0)$ is normalized and satisfies vanishing boundary conditions at $x \to \pm \infty$, then $|\Psi(x,t)|^2$ will be a probability distribution at all times $t$. Therefore knowing \eqref{eq:KHO} and specifying a \textit{valid} initial state $\Psi(x,0)$ will return a probability distribution which we can calculate by the use of \eqref{eq:PsiK}. For instance, if we were to take a particle in the initial state
\begin{align} \label{eq:initialHO}
        \Psi(x,0) = \bigg(\frac{m\omega}{\pi \hbar}\bigg)^{1/4} \exp\bigg(-\frac{m\omega}{2\hbar}(x-x_0)^2\bigg), 
    \end{align}
for $x_0 \in \R$, then we integrate $K(x,y;t,0)\Psi(y,0)$ to find
  \begin{align} \label{eq:evolvedHO}
    \Psi(x,t) &= \int_{\R} \bigg(\frac{m \omega}{2 \pi i \hbar \sin(\omega t)}\bigg)^{1/2} \exp \bigg[ \frac{i m\omega}{2 \hbar \sin (\omega t)} \{ (x^2+y^2) \cos \omega t - 2xy \} \bigg]  \bigg(\frac{m\omega}{2 \pi \hbar}\bigg)^{1/4} \exp\bigg(-\frac{m\omega}{4\hbar}(y-x_0)^2\bigg) dy \nonumber \\
    &= \bigg(\frac{m\omega}{\pi \hbar}\bigg)^{1/4} \exp \bigg[-\frac{m \omega}{2 \hbar}(x-x_0 \cos \omega t)^2 - i\bigg(\frac{1}{2} \omega t + \frac{m \omega}{\hbar} x_0 x \sin \omega t - \frac{m \omega}{4 \hbar} x_0^2 \sin 2 \omega t \bigg)\bigg] .
\end{align}
Consequently the resulting probability distribution is 
    \begin{align}
        |\Psi(x,t)|^2 =  \frac{1}{\sqrt{\pi}}\bigg(\frac{m \omega}{\hbar}\bigg)^{1/2} \exp \bigg[-\frac{m \omega}{\hbar}(x-x_0 \cos \omega t)^2 \bigg] \sim \mathcal{N}\bigg(x_0 \cos \omega t, \sqrt{\frac{\hbar }{2m\omega}}\bigg).
    \end{align}
In other words, \textit{we have produced a normal distribution by considering a harmonic oscillator potential}. This is a well-known result but it has a caveat that we will explore in a moment. Before we do so, we shall review the success clauses from Section \eqref{sec:intro02}. 
\begin{itemize}
    \item Clearly (1) is satisfied as it was our choice and $\int_{\R} |\Psi(y,0)|^2 dy =1$. 
    \item The propagator $K$ was produced by advanced, but standard, quantum \textit{mechanics} (as opposed to quantum \textit{computing}) techniques. This takes care of item (2).
    \item The integration of \eqref{eq:PsiK} roughly corresponds to calculating a Fourier transform. The resulting distribution was the normal distribution as we originally intended, thereby covering item (3). One way to get a feeling for the resulting distributions is to examine $K(x,y;-it,0)$
\begin{align} 
K(x,y;-it,0) &= \bigg(\frac{m \omega}{2 \pi i \hbar \sin(\omega (-it))}\bigg)^{1/2} \exp \bigg[ \frac{i m\omega}{2  \hbar \sin (-i\omega t)} \{ (x^2+y^2) \cos (-i\omega t) - 2xy \} \bigg] \nonumber \\
&= \frac{1}{\sqrt{2 \pi}} \frac{m\omega \operatorname{csch}(w t)}{\hbar} \exp\bigg[-\frac{m\omega \operatorname{csch}(w t)}{2 \hbar} ((x^2+y^2)\cosh(\omega t)-2xy)\bigg], \nonumber
\end{align} 
which, given the negative sign and the presence of $(x^2+y^2)$ in the exponential, has a structure resembling that of the PDF of the normal distribution. Note that
\begin{align}
    \iint_{\R^2} K(x,y;-it,0)dxdy = 2 \sqrt{\pi \operatorname{csch} \frac{t}{2}}. \nonumber
\end{align}
    \item Knowing that Trotterization or VarQRTE will be efficient in the global setup of the algorithm, then we have successfully loaded the normal distribution. This takes care of item (5).  If we want, for example, to load the standard normal $Z \sim \mathcal{N}(0,1)$, then we take $x_0=0$, set $\hbar = m = 1$, as is customary, and choose $\omega = \frac{1}{2}$.
\end{itemize}
The caveat is item (4) from the success clauses. Indeed, \eqref{eq:initialHO} \textit{already} corresponds to Gaussian amplitude and preparing such a system would be tantamount to being able to prepare \textit{any} state whose resulting $|\Psi(x,t)|^2$ is already a normal distribution. Therefore, one might argue that if we are able to prepare $\ket{\Psi(0)}$ as given by \eqref{eq:initialHO}, then we are able to prepare $\ket{\Psi(t)}$ as given by \eqref{eq:evolvedHO}. This can be alleviated by relaxing the condition that $\ket{\Psi(0)}$ be prepared in a state that only \textit{approximates} \eqref{eq:initialHO}. The better the approximation, the more taxing the preparation will be. However, one can still achieve very good results with very easily preparable states. We can illustrate this point now as follows.

Our choice is $\omega=\frac{1}{2}$ and $x_0 = 0$ and our goal is to approximate 
\[
f_{\mathcal{N},\mu,\sigma}(x) = f(\mu,\sigma,x) = \frac{1}{\sigma\sqrt{2\pi}} \exp\bigg[-\bigg(\frac{x-\mu}{\sigma}\bigg)^2\bigg]
\]
for $\mu=0$ and $\sigma=1$ with an easy initial state. Suppose we start with the rather unpromising initial state
\begin{align} \label{lad1} 
\Psi(x,0) 
    = 
    \begin{cases}
    \frac{1}{\sqrt{2}}, \quad &\mbox{for $-1 \le x \le 1$},  \\
    0, \quad &\mbox{otherwise}. 
    \end{cases}
\end{align}
This is valid as $\int_{\R}|\Psi(x,0)|^2 dx=1$ and $\Psi(x,0) = 0$ as $x \to \pm \infty$ but it does not resemble \eqref{eq:initialHO} at all. This corresponds to a plateau function which is hardly representative of the standard Gaussian distribution as seen in Figure \ref{fig:coarseGaussian1a}. We shall call such initial states \textit{ladders}.
    \begin{figure}[h!] 
   	\includegraphics[scale=0.467]{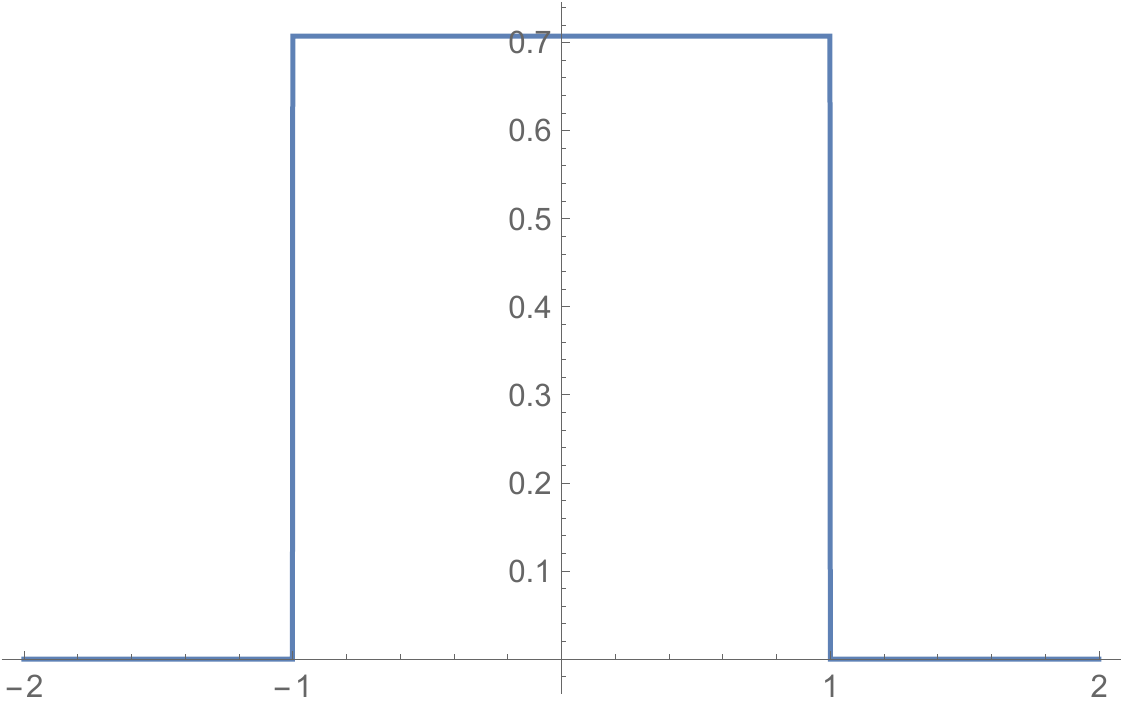} 
        \includegraphics[scale=0.467]{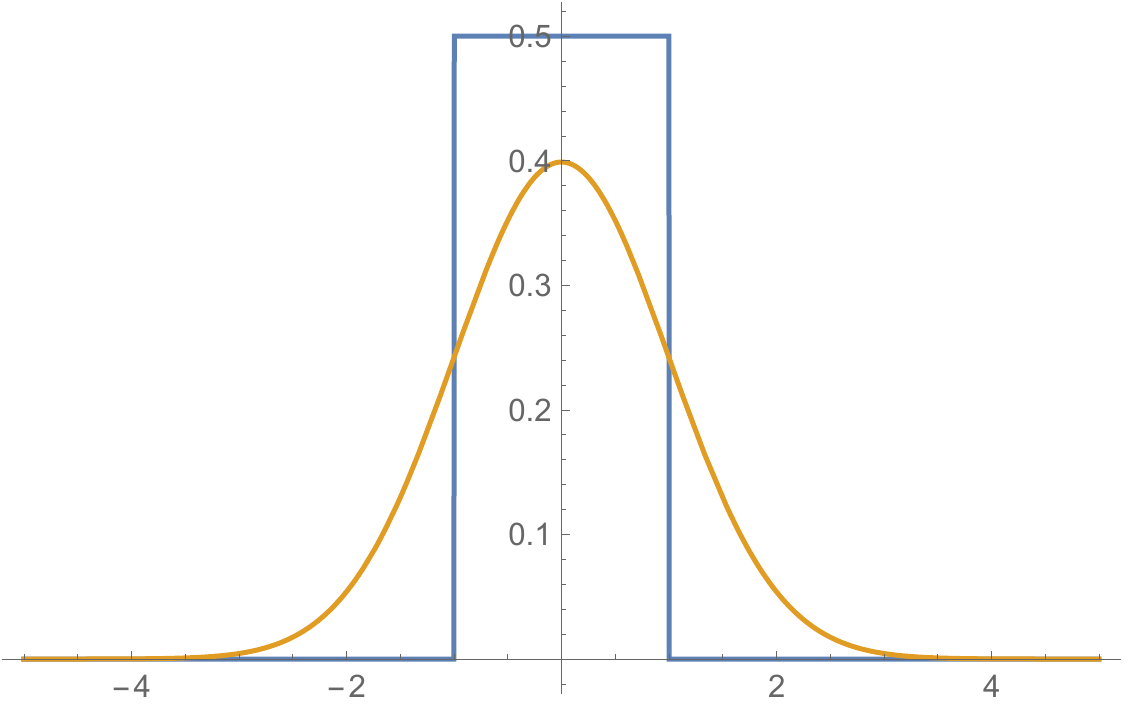} 
\caption{\underline{Left}: plot of $\Psi(x,0)$ as given \eqref{eq:exampleinitialGaussian}. \underline{Right}: plot of $|\Psi(x,0)|^2$ in blue and $f(0,1,x)$ in orange.}
    \label{fig:coarseGaussian1a}
       	\includegraphics[scale=0.467]{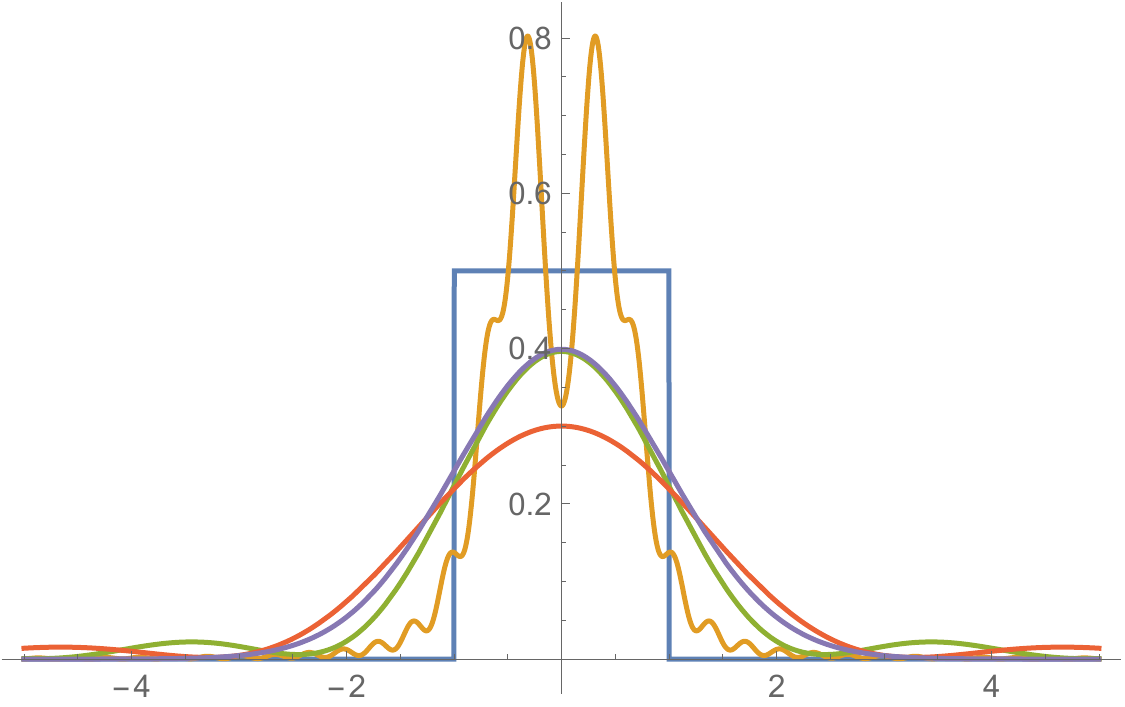} 
        \includegraphics[scale=0.467]{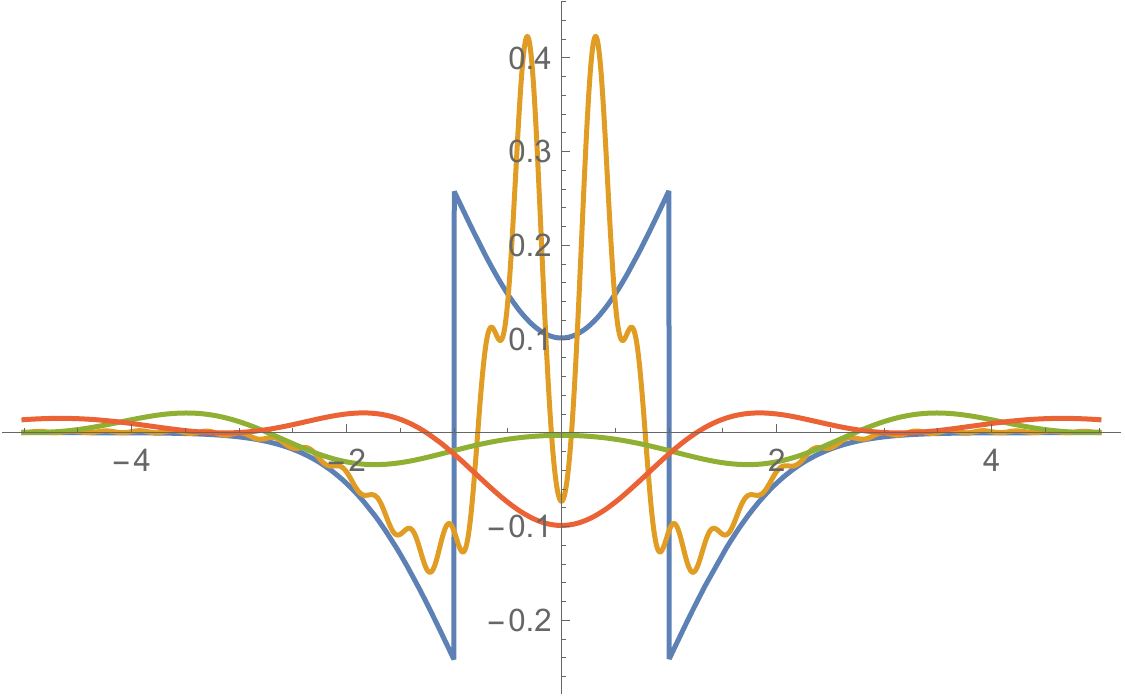} 
\caption{\underline{Left}: plot of $|\Psi(x,t)|^2$ for \textcolor{blue}{$t=0$ in blue}, \textcolor{orange}{$t=0.1$ in orange}, \textcolor{green}{$t=0.8$ in green} and \textcolor{purple}{$t=1.1$ in purple}, plotted against the \textcolor{red}{$f(0,1,x)$ in red}. \\
\underline{Right}: plot of $|\Psi(x,0)|^2 - f(0,1,x)$ in \textcolor{blue}{blue}, $|\Psi(x,0.1)|^2 - f(0,1,x)$ in \textcolor{orange}{orange}, $|\Psi(x,0.8)|^2 - f(0,1,x)$ in \textcolor{green}{green}, and $|\Psi(x,1.1)|^2 - f(0,1,x)$ in \textcolor{red}{red}.}
    \label{fig:coarseGaussian1b}
       	\includegraphics[scale=0.467]{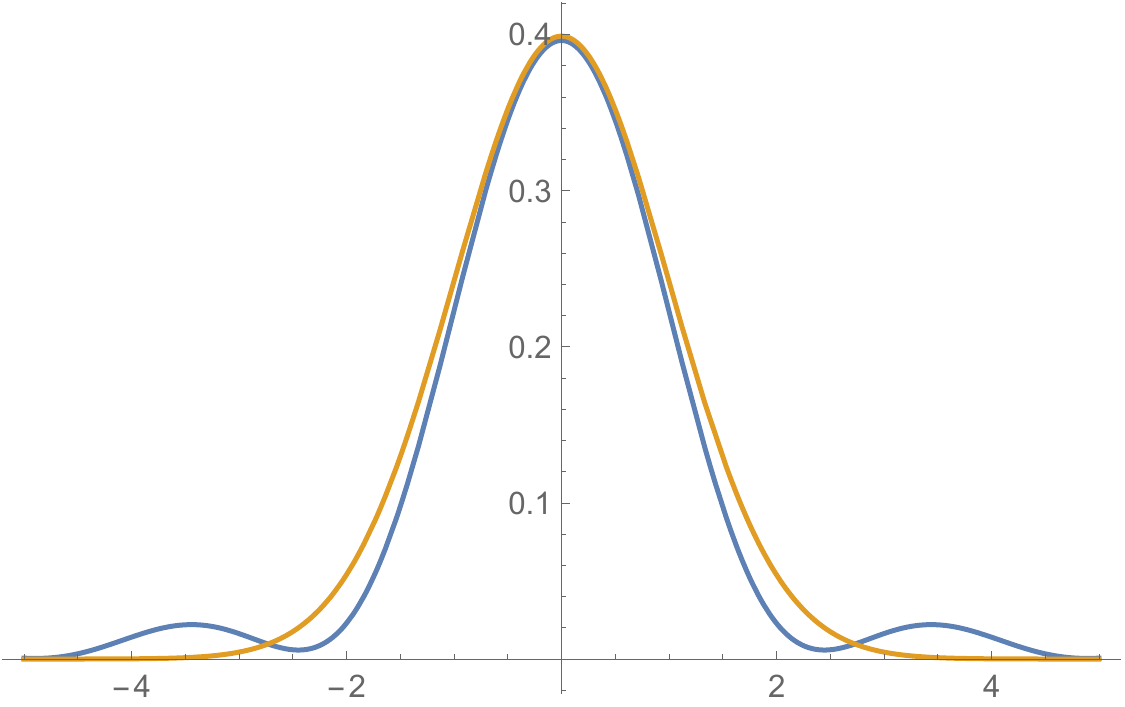} 
        \includegraphics[scale=0.467]{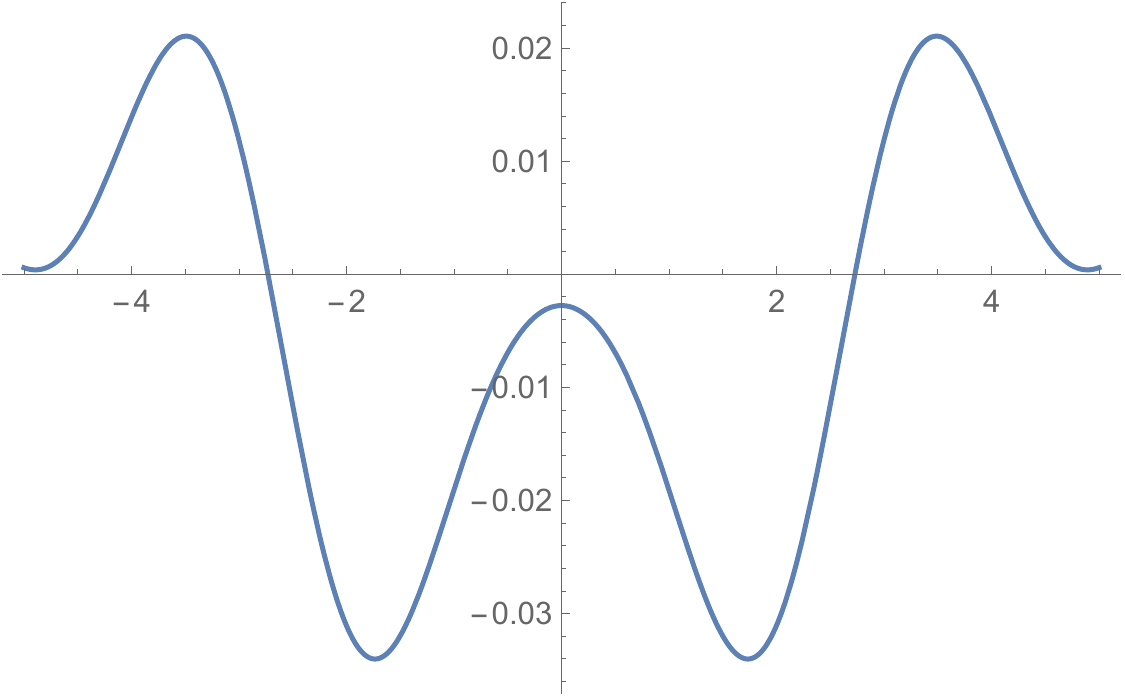} 
\caption{\underline{Left}: plot of $|\Psi(x,0.8)|^2$ in blue and $f(0,1,x)$ in orange. \underline{Right}: plot of $|\Psi(x,0.8)|^2 - f(0,1,x)$ for $-5 \le x \le 5$.}
    \label{fig:coarseGaussian1c}
    \end{figure}

Moreover, when we evolve with respect to $t$ we obtain
\begin{align}
    \Psi(x,t) &= \int_{-\infty}^\infty K(x,y;t,0) \Psi(y,0) dy = \frac{1}{\sqrt{2}}\int_{-1}^1 K(x,y,t,0) dy \nonumber \\
    &= \frac{i e^{-\frac{1}{4} i x^2 \tan
   \left(\frac{t}{2}\right)}}{2 \sqrt{2 \cos
   \left(\frac{t}{2}\right)}} \bigg[\operatorname{erfi}\bigg(\frac{\left(\frac{1}{2}+\frac{i}{2}\right)
   \left(x-\cos \left(\frac{t}{2}\right)\right)}{\sqrt{\sin
   \left(t\right)}}\bigg)-\operatorname{erfi}\bigg(\frac{\left(\frac{1}{2}+\frac{i}{2}\right) \left(x+\cos
   \left(\frac{t}{2}\right)\right)}{\sqrt{\sin
   \left(t\right)}}\bigg) \bigg],
\end{align}
where $\operatorname{erfi}$ is the imaginary error function $\operatorname{erf}(iz)/i$ and $\operatorname{erf}$ is defined by
\begin{align}
    \operatorname{erf}(z) = \frac{2}{\sqrt{\pi}} \int_0^z e^{-t^2}dt. 
\end{align}
The probability evolutions are given by
\begin{align}
    |\Psi(x,t)|^2 = \frac{1}{8} \sec \left(\frac{t}{2}\right) \bigg|
   \text{erfi}\bigg(\frac{\left(\frac{1}{2}+\frac{i}{2}\right)
   \left(x-\cos \left(\frac{t}{2}\right)\right)}{\sqrt{\sin
   \left(t\right)}}\bigg)-\text{erfi}\bigg(\frac{\left(\frac{1}{2}+\frac{i}{2}\right) \left(x+\cos
   \left(\frac{t}{2}\right)\right)}{\sqrt{\sin
   \left(t\right)}}\bigg) \bigg| ^2
\end{align}
and they are plotted for different values of $t$ in Figure \ref{fig:coarseGaussian1b}.

These evolutions become more acceptable at $t$ evolves. For instance, if we we stop at $t=0.8$, then we empirically find the plots shown in Figure \ref{fig:coarseGaussian1c}.

This approximation seems to be \textit{adequate} except for the sinusoidal \textit{tails} and in fact only produces an error of at most $\lesssim 4\%$.

We can improve the situation by considering the initial ground state defined by the more precise piece-wise function
\begin{align} \label{eq:exampleinitialGaussian}
    \Psi(x,0) = 
    \begin{cases}
    0, \quad &\mbox{if $-\infty < x \le -2$}, \\
    \frac{1}{4}, \quad &\mbox{if $-2 < x \le -1$}, \\
    \frac{\sqrt{7}}{4}, \quad &\mbox{if $-1 < x \le 1$}, \\
     \frac{1}{4}, \quad &\mbox{if $1 < x \le 2$}, \\
     0, \quad &\mbox{if $2 < x < \infty$}.
    \end{cases}
\end{align}
Note that we still have $\int_{\R}|\Psi(x,0)|^2 dx = \int_{\R}\Psi(x,0)^2 dx = 1$ as well as $\Psi(x,0) = 0$ as $x \to \pm \infty$ but it still bears little resemblance to \eqref{eq:initialHO}. This has a somewhat more Gaussian symmetry as we can appreciate in Figure \ref{fig:coarseGaussian2a} but it can hardly be labeled normal (closer to the binomial distribution, rather).
    \begin{figure}[h!] 
   	\includegraphics[scale=0.467]{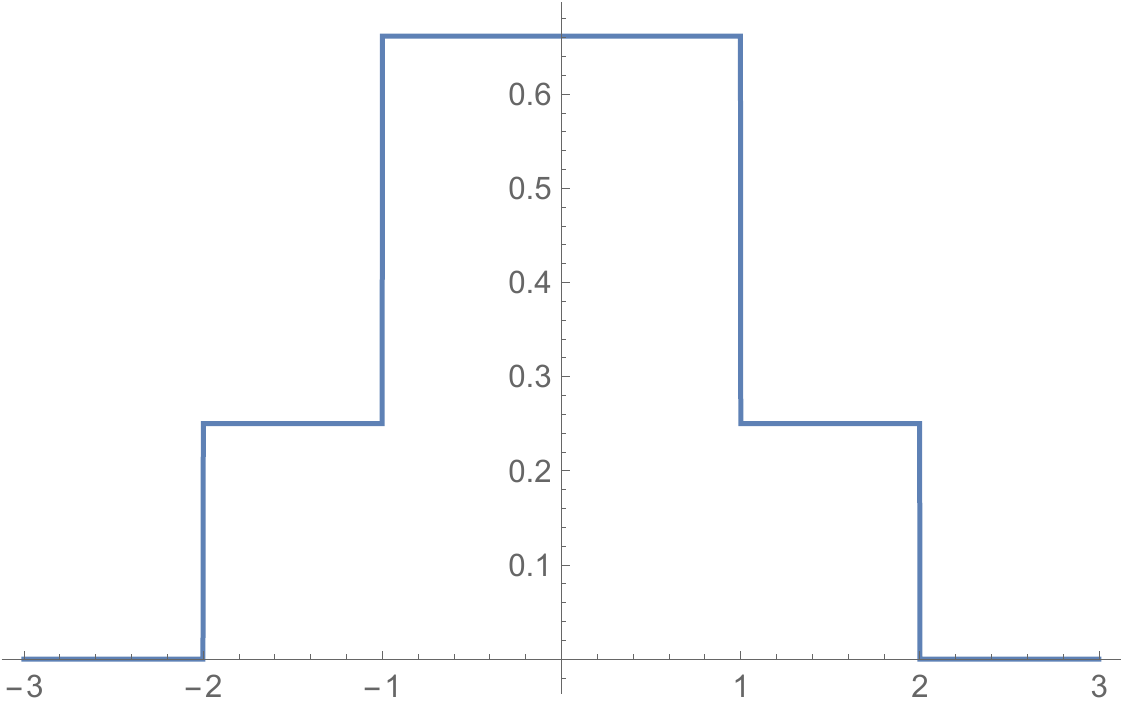} 
        \includegraphics[scale=0.467]{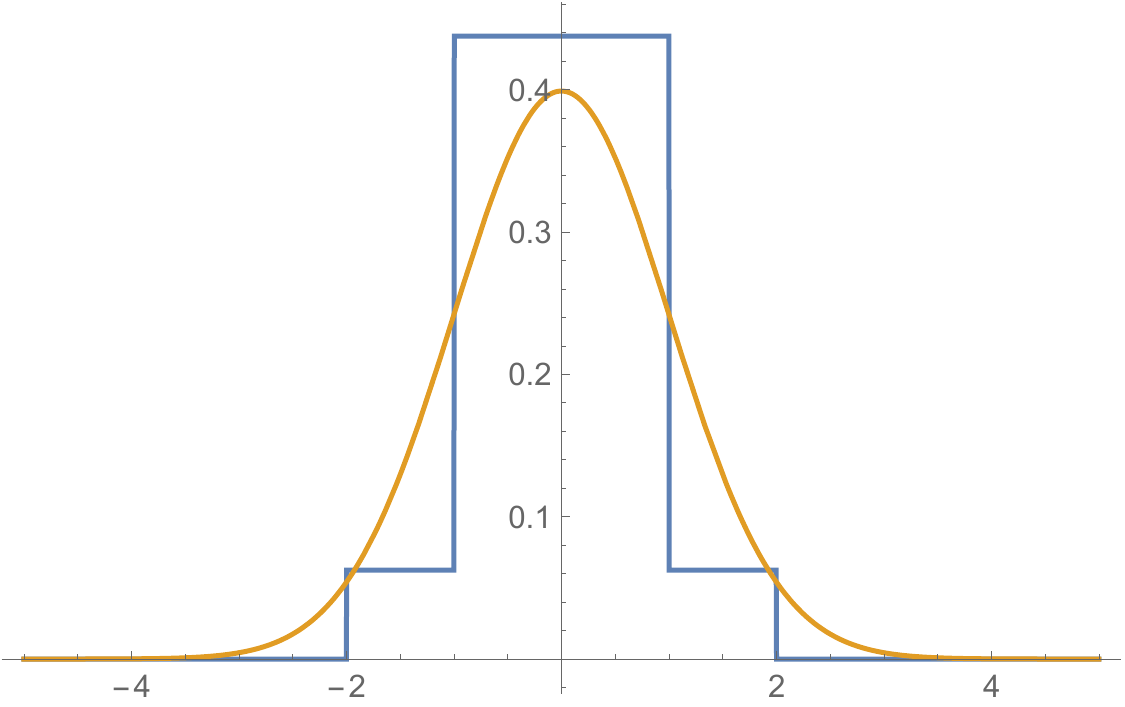} 
\caption{\underline{Left}: plot of $\Psi(x,0)$ as given \eqref{eq:exampleinitialGaussian}. \underline{Right}: plot of $|\Psi(x,0)|^2$ in blue and $f(0,1,x)$ in orange.}
    \label{fig:coarseGaussian2a}
        \includegraphics[scale=0.467]{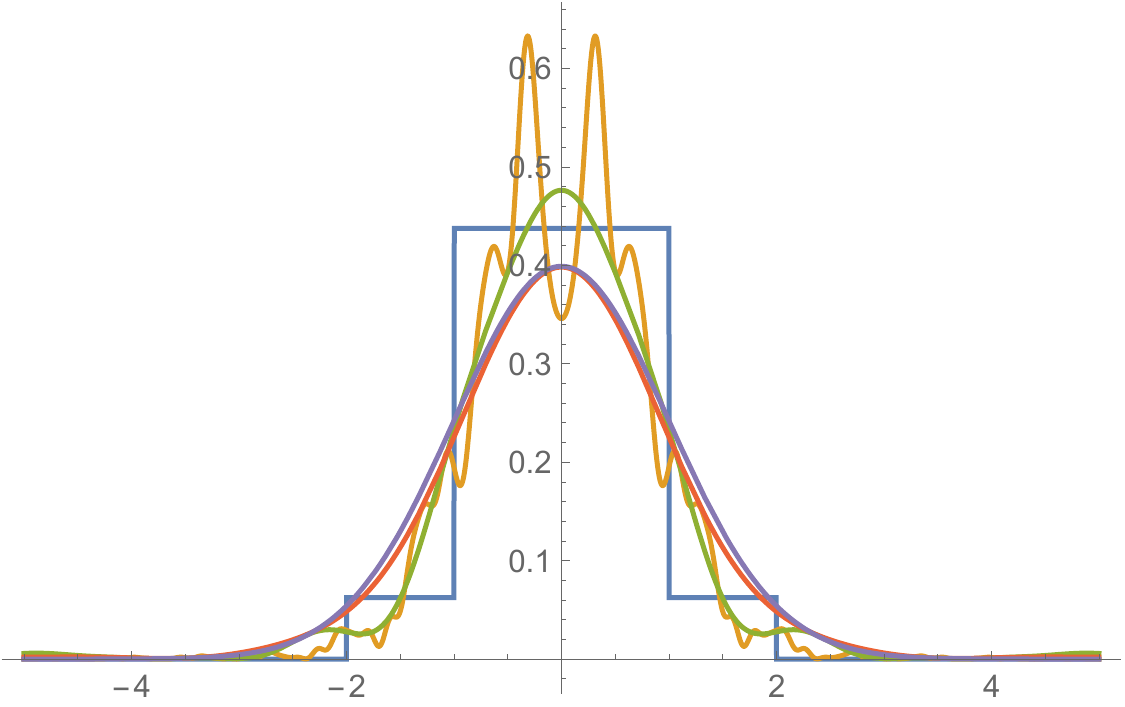} 
    \includegraphics[scale=0.467]{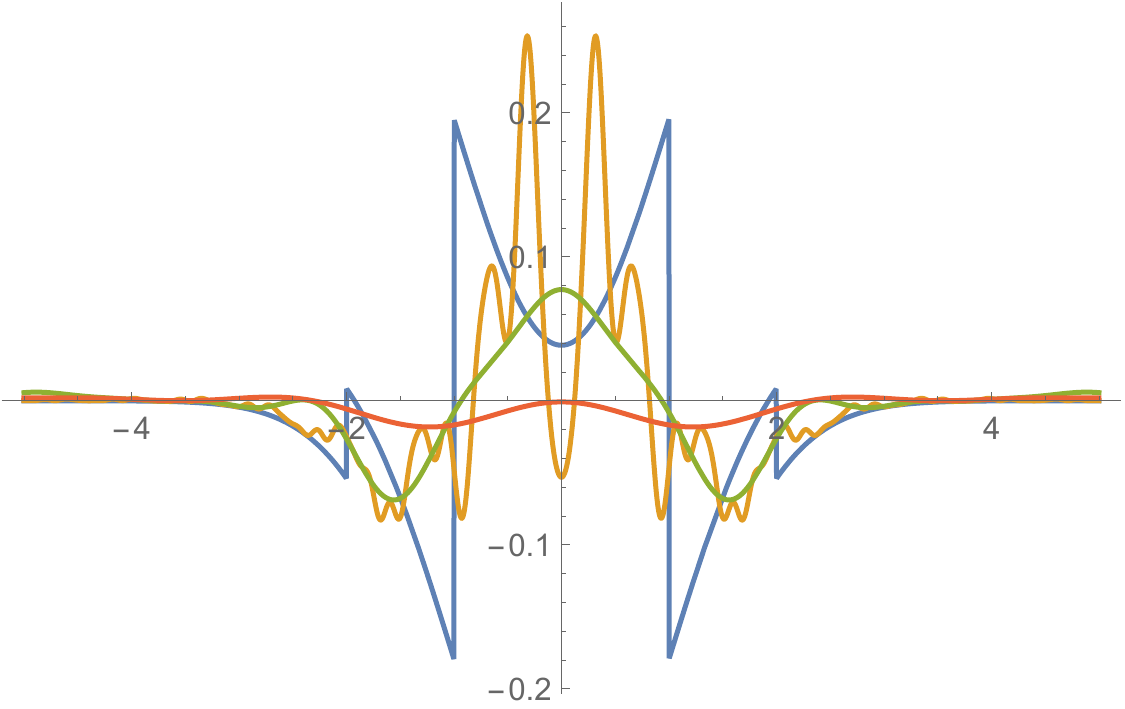} 
\caption{\underline{Left}: plot of $|\Psi(x,t)|^2$ for \textcolor{blue}{$t=0$ in blue}, \textcolor{orange}{$t=0.1$ in orange}, \textcolor{green}{$t=0.6$ in green} and \textcolor{purple}{$t=1.25$ in purple}, plotted against the \textcolor{red}{$f(0,1,x)$ in red}. \\
\underline{Right}: plot of $|\Psi(x,0)|^2 - f(0,1,x)$ in \textcolor{blue}{blue}, $|\Psi(x,0.1)|^2 - f(0,1,x)$ in \textcolor{orange}{orange}, $|\Psi(x,0.6)|^2 - f(0,1,x)$ in \textcolor{green}{green}, and $|\Psi(x,1.25)|^2 - f(0,1,x)$ in \textcolor{red}{red}.}
    \label{fig:coarseGaussian2b}
       	\includegraphics[scale=0.467]{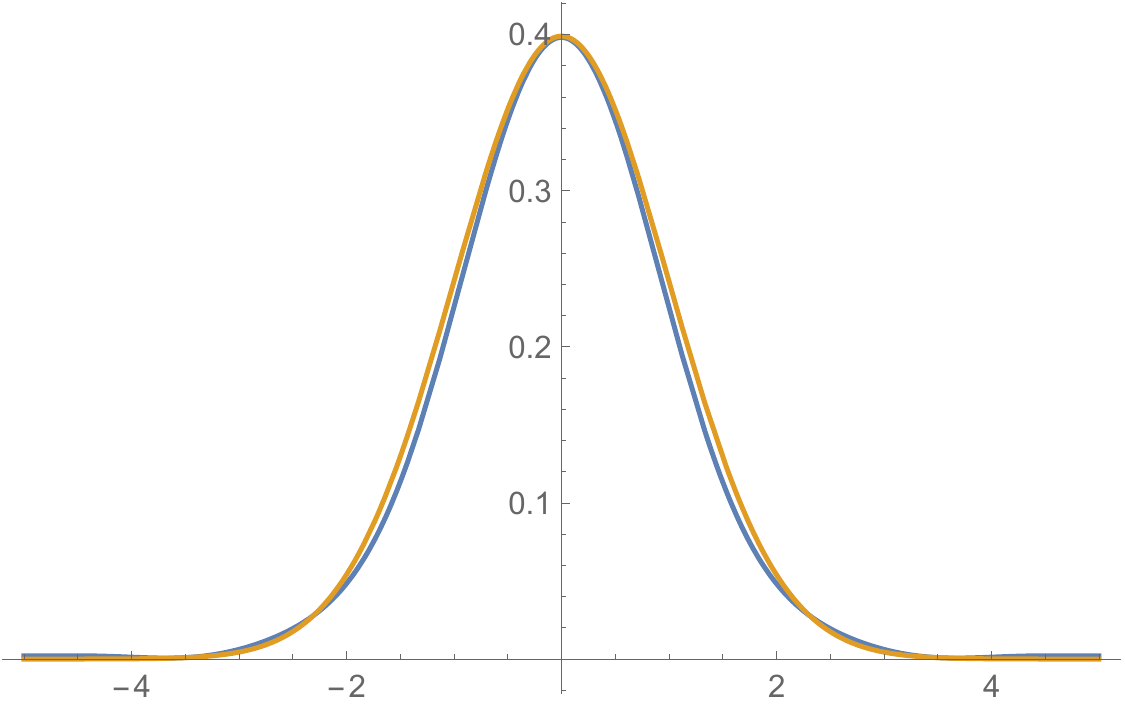} 
        \includegraphics[scale=0.467]{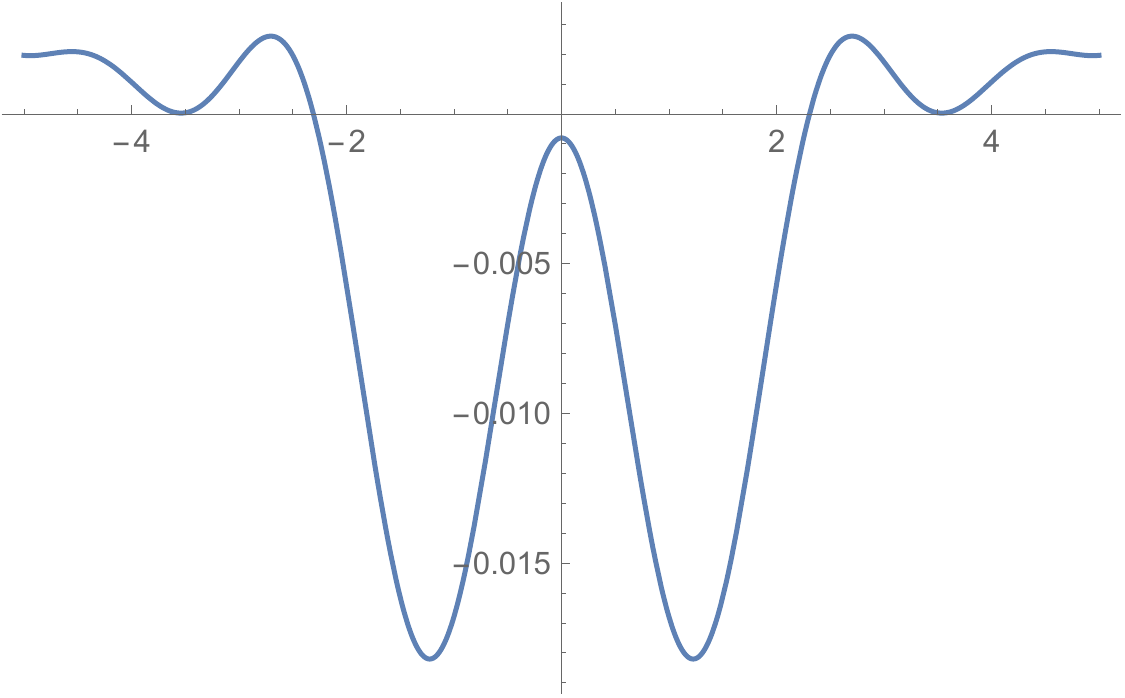} 
\caption{\underline{Left}: plot of $|\Psi(x,1.25)|^2$ in blue and $f(0,1,x)$ in orange. \underline{Right}: plot of $|\Psi(x,1.25)|^2 - f(0,1,x)$ for $-5 \le x \le 5$.}
    \label{fig:coarseGaussian2c}
    \end{figure}

On the other hand, it should be noted that it still is very easy and efficient to manufacture this initial state. 

As $t$ evolves away from $0$ we find the following evolution
\begin{align}
    \Psi(x,t) &= \int_{-\infty}^\infty K(x,y;t,0) \Psi(y,0) dy \nonumber \\
    &= \frac{1}{4}\int_{-2}^{-1} K(x,y,t,0) dy + \frac{\sqrt{7}}{4}\int_{-1}^{1} K(x,y,t,0) dy + \frac{1}{4}\int_{1}^{2} K(x,y,t,0) dy \nonumber \\
    &= \frac{\sqrt[4]{-1} \sqrt{h} \sqrt{\sin (2 t w)} \sec (t w)
   \sqrt{-\frac{i m w \csc (t w)}{h}} e^{-\frac{i m w x^2 \tan (t
   w)}{2 h}}}{8 \sqrt{2} \sqrt{m} \sqrt{w}} \nonumber \\
   & \quad \times \bigg[ \text{erf}\bigg(\frac{(-1)^{3/4} \sqrt{m} \sqrt{w} (x-2 \cos (t
   w))}{\sqrt{h} \sqrt{\sin (2 t
   w)}}\bigg)+(\sqrt{7}-1)
   \text{erf}\bigg(\frac{(-1)^{3/4} \sqrt{m} \sqrt{w} (x-\cos (t
   w))}{\sqrt{h} \sqrt{\sin (2 t
   w)}}\bigg) \nonumber \\
   & \quad \quad -(\sqrt{7}-1)
   \text{erf}\bigg(\frac{(-1)^{3/4} \sqrt{m} \sqrt{w} (\cos (t
   w)+x)}{\sqrt{h} \sqrt{\sin (2 t
   w)}}\bigg)-\text{erf}\bigg(\frac{(-1)^{3/4} \sqrt{m} \sqrt{w}
   (2 \cos (t w)+x)}{\sqrt{h} \sqrt{\sin (2 t w)}}\bigg)\bigg]. \nonumber
\end{align}

The resulting PDF is given by
\begin{align}
    |\Psi(x,t)|^2
    &= \frac{1}{128} \bigg|\sin (2 t w) \csc (t w) \sec ^2(t w) \nonumber \\
    &\quad \times \bigg\{\text{erf}\bigg(\frac{(-1)^{3/4} m w
   (x-2 \cos (t w))}{\sqrt{h m w \sin (2 t w)}}\bigg) -\text{erf}\bigg(\frac{(-1)^{3/4} m
   w (2 \cos (t w)+x)}{\sqrt{h m w \sin (2 t w)}}\bigg) \nonumber \\
   &\quad \quad +(\sqrt{7}-1)
   \text{erf}\bigg(\frac{(-1)^{3/4} m w (x-\cos (t w))}{\sqrt{h m w \sin (2 t
   w)}}\bigg) -(\sqrt{7}-1) \text{erf}\bigg(\frac{(-1)^{3/4} m w (\cos
   (t w)+x)}{\sqrt{h m w \sin (2 t w)}}\bigg) \bigg\}^2 \bigg|. \nonumber 
\end{align}

We can now plot this for different values of $t$ in Figure \ref{fig:coarseGaussian2b}.

It is empirically clear that when $t=1.25$, the resulting PDF of $|\Psi(x,t)|^2$ is nearly $f(0,1,x)$ as shown in Figure \ref{fig:coarseGaussian2c}. 

We may therefore conclude that a very easy state to prepare such as the one given by \eqref{eq:exampleinitialGaussian} will, after some evolution steps, resemble very closely the probability distribution function of the standard Gaussian. In fact, the sinusoidal tails are much more subdued and the largest error is now only $\lesssim 2\%$. Just how good this approximation actually is will be discussed in future research.

Our objective here is to replicate this for other probability distributions (e.g. $\chi$, $\chi^2$, Gamma, Laplace, etc...) with different potentials (free particles, radial harmonic oscillators, Dirac $\delta$-potentials, hydrogen atoms etc...)

In order to simulate the quantum Harmonic oscillator Hamiltonian $H=\frac{p^{2}}{2m} +  \frac{m\omega^{2}x^{2}}{2}$ for the case $m=1/2, \omega=2$, we use a disentangling formula that expresses the Hamiltonian as a product of three terms, each diagonal in either the standard or the Fourier basis \cite{curtright2021lie}:
\all{ 
\exp( -it (p^{2} + x^{2}) ) = \exp ( - i x^{2} \tan(t)/2 ) \exp( - i p^{2} \sin(2t) /2 ) \exp ( - i x^{2} \tan(t)/2 )
} {qho}
The disentangling formula directly yields an extremely efficient quantum circuit for simulating the Hamiltonian $H$. The circuit is a composition of phase  $\exp ( - i x^{2} \tan(t)/2 )$  applied in the standard basis, followed by a phase $\exp( - i p^{2} \sin(2t) /2 )$ applied in the Fourier basis followed by a phase  $\exp ( - i x^{2} \tan(t)/2 )$ applied in the standard basis. The circuit requires two quantum Fourier transforms, and has gate complexity independent of $t$. 

\subsection{Real vs imaginary evolution for the zero potential} \label{sec:intro04}
As discussed in Section \ref{sec:intro02} we shall now illustrate the difference between evolving a particle with $V(x)=0$, i.e. a free particle in real and imaginary time and compare both situations.
\begin{center}
    \begin{longtable}
    {| p{3cm} || p{7cm} | p{7cm}|}
    \hline
    \begin{center}~\\\textbf{Feature}\end{center} & \begin{center}\textbf{Real time}\end{center} & \begin{center}\textbf{Imaginary time}\end{center} \\
    \hline 
    \begin{center}Name\end{center} & \begin{center} Schr\"{o}dinger equation for the free particle \end{center}  & \begin{center} Heat equation\end{center} \\
    \hline
    \begin{center}PDE\end{center} & \begin{align}i\hslash \frac{\partial \psi(x,t)}{\partial t} = - \frac{\hslash^2}{2m}\frac{\partial^2 \psi(x,t)}{\partial x^2} \nonumber \end{align} & \begin{align}\frac{\partial u(x,t)}{\partial t} = k \frac{\partial^2 u(x,t)}{\partial x^2} \nonumber \end{align}\\
    \hline
    \begin{center}Initial parameters\end{center} & \begin{align}\varepsilon, \hbar, m, t > 0 \textnormal{ and } k_0 \in \R, \nonumber \end{align} \begin{center} usually $\hbar = m = 1$ and $k_0 = 0$ \end{center} & \begin{align} k, t>0 \textnormal{ and } x_0\in \R, \nonumber \end{align} \begin{center} usually $k = \frac{1}{2}$ from It\^{o} arguments \end{center} \\
    \hline
    \begin{center}Dirichlet boundary conditions\end{center} & \begin{align} \psi(x,0) = \frac{1}{(2 \pi \varepsilon^2)^{1/4}} \exp \bigg(-\frac{x^2}{4\varepsilon^2} + i k_0 x\bigg), \nonumber \end{align} \begin{center}minimum uncertainty wave packet\end{center}& \begin{align} u(x,0) = \delta(x-x_0), \nonumber \end{align} \begin{center}Dirac delta function\end{center}\\
    \hline
    \begin{center}Limiting case of initial condition\end{center} &\begin{align} \textnormal{As $\varepsilon \to 0$ we have } \psi(x,0) \to \begin{cases} +\infty & x = 0, \nonumber \\ 0 & x \ne 0. \nonumber \end{cases}\end{align} \begin{align} \textnormal{Also } \int_{\R} \psi(x,0)dx = 2^{3/4} \pi^{1/4} e^{-k_0^2 \varepsilon^2} \sqrt{\varepsilon} \ne 1. \nonumber \end{align} \begin{center} It does not hold that $\psi(x,0) \to \delta(x)$.\end{center} & \begin{center} N/A \end{center}\\
    \hline
    \begin{center}Loading difficulty\end{center} & \begin{center}Hard, depending on $\hbar, m, k_0$ and specially $\varepsilon$. It becomes for easier for $\varepsilon \ll 1.$\end{center} & \begin{center}Easy\end{center}\\
    \hline
    \begin{center}Loading accuracy\end{center} & \begin{center}Approximate, depending on layers of ladder states.\end{center} & \begin{center}Exact\end{center}\\
    \hline
    \begin{center}Evolution accuracy on quantum device\end{center} & \begin{center}Exact\end{center} & \begin{center}Approximate, depending on approximation to proxy norm $Z$.\end{center}\\
    \hline
    \begin{center}Computational speedup\end{center} & \begin{center}Exponential\end{center} & \begin{center}Exponential\end{center}\\
    \hline
    \begin{center}Solution readout\end{center} & \begin{center}Quadratic\end{center} & \begin{center}Quadratic\end{center}\\
    \hline
    \begin{center}Initial normalization\end{center} & \begin{align} \int_{\R} |\psi(x,0)|^2 dx = 1 \nonumber \end{align} & \begin{align} \int_{\R} |u(x,0)| dx = \int_{\R} u(x,0) dx = 1 \nonumber \end{align} \\
    \hline
    \begin{center}Norm\end{center} & \begin{align} \ell_2 \nonumber \end{align} & \begin{align} \ell_1 \nonumber \end{align} \\
    \hline
    \begin{center}Quantum encoding\end{center} & \begin{align} \ket{\psi} = \sum_{i=0}^{2^n-1} \sqrt{p_i} \ket{i} \nonumber \end{align} & \begin{align} \ket{{\tilde \psi}} = \frac{1}{Z}\sum_{i=0}^{2^n-1} p_i \ket{i} \nonumber \end{align} \\
    \hline
    \begin{center}Proxy norm\end{center} & \begin{center}None -- quantum native\end{center} & \begin{align} Z = \bigg(\sum_{i=0}^{2^n-1}p_i^2(t)\bigg)^{1/2} \nonumber \end{align} \\
    \hline
    \begin{center}Solution to PDE with corresponding initial condition\end{center} & \begin{align} \psi(x,t) &= \frac{1}{(2 \pi \varepsilon^2)^{1/4}} \bigg(1+ \frac{i \hslash t}{2 m \varepsilon^2}\bigg)^{-1/2} \nonumber \\ & \times \exp \bigg[\frac{1}{1+\tfrac{i \hslash}{2m \varepsilon^2}t} \bigg(-\frac{x^2}{4\varepsilon^2} + ik_0 x - \frac{i\hslash k_0^2 t}{2m} \bigg) \bigg] \nonumber \end{align} & \begin{align} u(x,t) = \frac{1}{\sqrt{4 \pi k t}} \exp\bigg(-\frac{(x-x_0)^2}{4kt}\bigg) \nonumber \end{align} \\
    \hline
    \begin{center}Nature of solution\end{center} & \begin{center} Quantum mechanical wave function \nonumber \end{center} & \begin{center}Heat kernel\end{center}  \\
    \hline
    \begin{center}Range of solution\end{center} & \begin{align} \psi(x,t) \in \C \nonumber \end{align} & \begin{align} u(x,t) \in \R^+ \nonumber \end{align} \\
    \hline
    \begin{center}Evolution normalization\end{center} & \begin{align} \int_{\R} |\psi(x,t)|^2 dx = 1 \quad \forall t>0 \nonumber \end{align} & \begin{align} \int_{\R} |u(x,t)| dx =\int_{\R} u(x,t) dx = 1  \quad \forall t>0 \nonumber \end{align} \\
    \hline
    \begin{center}Gaussian distribution\end{center} & \begin{align} |\psi(x,t)|^2 &= \frac{1}{(2 \pi \varepsilon^2)^{1/2}} \bigg(1+\bigg(\frac{\hbar t}{2 m \varepsilon^2}\bigg)^2\bigg)^{-1/2} \nonumber \\
    &\times \exp \bigg[-\frac{2}{1+(\tfrac{\hbar t}{2m \varepsilon^2})^2}\bigg(\frac{x}{2\varepsilon}-\frac{\hbar k_0 t}{2m \varepsilon}\bigg)^2\bigg] \nonumber \\ &\sim \mathcal{N}(\mu, \sigma^2) \nonumber \nonumber \end{align} & \begin{align} u(x,t) &= \frac{1}{\sqrt{2kt} \sqrt{2 \pi}} \exp\bigg[-\bigg(\frac{x-x_0}{2\sqrt{2kt}}\bigg)^2\bigg] \nonumber \\ &\sim \mathcal{N}(\mu, \sigma^2 ) \nonumber \end{align}  \\
    \hline
    \begin{center}Distribution parameters\end{center} & \begin{align} \mu = \frac{\hbar k_0 t}{m} \textnormal{ and } \sigma^2 = \varepsilon^2 +\bigg(\frac{\hbar t}{2 m \varepsilon}\bigg)^2 \nonumber \end{align} & \begin{align} \mu &= x_0 \textnormal{ and } \sigma^2 = 2kt \nonumber \end{align}  \\
    \hline
    \begin{center}Distribution encoding \end{center} & \begin{center} braket $|\psi(x,t)|^2=\braket{\psi(x,t)|\mathbb{I}|\psi(x,t)}$ \end{center} & \begin{center} ket $u(x,t) = \frac{1}{Z}\ket{\tilde \psi(x,t)}$ \end{center} \\
    \hline
    \begin{center}Limiting case as $\varepsilon \to 0$ or as $t \to 0$\end{center} & \begin{align} |\psi(x,0)|^2 &=  \frac{1}{|\varepsilon|\sqrt{2\pi}} \exp \bigg(-\frac{x^2}{2\varepsilon^2}\bigg) \to \delta(x) \textnormal{ as } \varepsilon \to 0, \nonumber \end{align} \begin{center}heat kernel\end{center} & \begin{align} u(x,t) \to \delta(x-x_0)  \textnormal{ as } t \to 0, \nonumber \end{align} \begin{center} Dirac delta function\end{center} \\
    \hline
    \begin{center}Ansatz set up\end{center} & \begin{align} i \frac{\partial}{\partial t}\ket{\phi({\boldsymbol \theta}(t))} = {\hat H} \ket{\phi({\boldsymbol \theta}(t))} \nonumber\end{align} & \begin{align} \frac{\partial}{\partial t}\ket{\tilde \phi({\boldsymbol \theta}(t))} = \hat H \ket{\tilde \phi({\boldsymbol \theta}(t))} \nonumber \end{align}  \\
    \hline
    \begin{center}Method\end{center} & \begin{center} Variational Quantum Real Time Evolution \end{center} & \begin{center} Variational Quantum Imaginary Time Evolution \end{center}  \\
    \hline
    \begin{center}Variational principle\end{center} & \begin{center} McLachlan \end{center} & \begin{center} McLachlan \end{center}  \\
    \hline
    \begin{center}Equations of motion for gate angles $\boldsymbol{\theta}$\end{center} & \begin{align} \sum_k \real(A_{i,k}) \dot \theta_k = \imag(C_i) \nonumber \end{align} & \begin{align} \sum_k \real(A_{i,k}) \dot \theta_k = \real(C_i) \nonumber \end{align}  \\
    \hline
    \end{longtable}
\end{center}

The rest of the paper is organized as follows. Section \ref{sec:list} will describe the different types of potentials frequently encountered in the literature and for which we can find the associated propagators or for which we can solve the Schr\"{o}dinger equation and investigate the probability distributions that they lead to. We shall take the reverse approach in Section \ref{sec:inverselist} and discuss the possibilities of generating probability distributions from given potentials. In Section \ref{sec:amplitudeencoding} we will formalize the empirical findings from this section and discuss the complexity and realization of algorithms for ladder states. Moreover, in Section \ref{sec:evolution} we will discuss Hamiltonian simulations algorithms that could be applied to these ladder states in order to obtain efficient approximations to probability distribution functions. Finally in Section \ref{sec:conclusion} we contextualize our findings with existing literature, notably that of fast forwarding. Appendix \ref{sec:appendix1} contains new mathematical results in the theory of Bessel integrals and illustrations on how to apply them for a certain type of path integral and will help elucidate future research.

\section{Potentials and their associated probability distributions} \label{sec:list}
\subsection{The unrestricted free particle and the weighted sums of $\chi$ distributions} \label{sec:list01}
Recall that $T=t-t_0=t''-t'$. The potential is simply $V(x)=0$ which yields Hamiltonian and Lagrangian 
\begin{align}
        \hat H = -\frac{\hbar^2}{2m}\frac{\partial^2}{\partial x^2} \quad \textnormal{and} \quad \mathcal{L} = \frac{m}{2}{\dot x}^2.
    \end{align}
Unlike the harmonic oscillator, in this case the spectrum is continuous, not discrete. For the unrestricted particle $x \in \R$ we have the following properties. The propagator is \cite[$\mathsection$6.2.1.2]{handbookPI}
    \begin{align} \label{eq:Kfree}
        K(x'',x';t'',t') &= \int_{x(t')=x'}^{x(t'')=x''} \mathcal{D}x(t) \exp \bigg[\frac{i}{h} \int_{t'}^{t''} \frac{m}{2}{\dot x}^2 dt \bigg] = \sqrt{\frac{m}{2 \pi i \hbar T}} \exp\bigg[\frac{im}{2 \hbar T}(x''-x')^2\bigg].
    \end{align}
Note how letting $\omega \to 0^+$ in \eqref{eq:KHO} leads to \eqref{eq:Kfree}. Our choice for the initial wave function will be
    \begin{align} \label{eq:Psi0freechoice}
        \Psi_{a,b}(y,0) = \Psi(y,0) = \bigg(\frac{2^{b+1/2}a^{b-1/2}}{(1+(-1)^{2b})\Gamma(b-\frac{1}{2})}\bigg)^{1/2} \exp(-ay^2)y^{b-1}.
    \end{align}
We can easily verify that $|\Psi_{a,b}(y,0)|^2$ integrates and vanishes as $y \to \pm 0$ provided $a>0$ and $b> \frac{1}{2}$. We shall also assume that $b \in \N$. Clearly if $a=b=1$, then this will reduce to the normal case. For $t>0$ the wave function is
\begin{align} \label{eq:freewave}
        \Psi_{a,b}(x,t) &= \int_{\R} K(x,y;t,0)\Psi_{a,b}(y,0)dy \nonumber \\
        &= \frac{2^{b-\frac{7}{4}} a^{\frac{b}{2}-\frac{1}{4}} (h t)^{\frac{b-1}{2}}
   \sqrt{-\frac{i m}{h t}} e^{\frac{i m x^2}{2 h t}} (2 a h t-i
   m)^{\frac{1}{2} (-b-1)}}{\sqrt{\pi } \sqrt{\left((-1)^{2 b}+1\right)
   \Gamma \left(b-\frac{1}{2}\right)}} \nonumber \\
        &\quad \times \bigg\{-\sqrt{2} \left((-1)^b-1\right) \Gamma \left(\frac{b}{2}\right) \sqrt{h t
   (2 a h t-i m)} \, _1F_1\left(\frac{b}{2};\frac{1}{2};\frac{m^2 x^2}{2 i
   h m t-4 a h^2 t^2}\right) \nonumber \\
        & \quad \quad \quad \quad -2 i \left((-1)^b+1\right) m x \Gamma \left(\frac{b+1}{2}\right) \,
   _1F_1\left(\frac{b+1}{2};\frac{3}{2};\frac{m^2 x^2}{2 i h m t-4 a h^2
   t^2}\right) \bigg\},
    \end{align}
where $_1F_1$ is the confluent hypergeometric function \cite{gradryz, watson}. For simplicity we set $\hbar = m =1$. Plugging $b=1$ shows that the case $|\Psi_{a,1}(x,1)|^2$ in \eqref{eq:freewave} is a normal distribution centered at $\mu=0$ as
\begin{align}
    |\Psi_{a,1}(x,t)|^2 = \sqrt{\frac{2a}{\pi+4 \pi a^2 t^2}} \exp\bigg(-\frac{2ax^2}{1+4a^2t^2}\bigg) \sim \mathcal{N}\bigg(0, \frac{1}{2\sqrt{a}}\sqrt{1+4a^2t^2} \bigg).
\end{align}
On the left of Figure \ref{fig:list01} we plot $|\Psi_{a,1}(x,t)|^2$ for $a=1$ at $t=1,2,3$. The free particle provides an \textit{alternative} method to loading the normal distribution.

\begin{figure}[h] 
   	\includegraphics[scale=0.467]{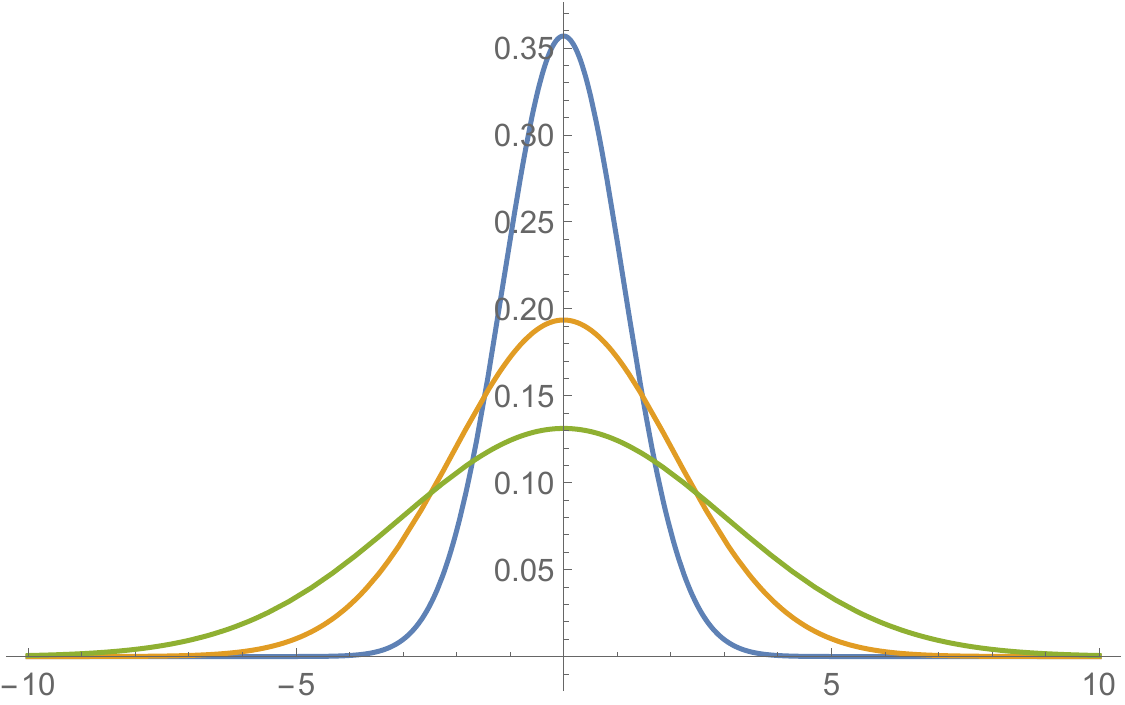} 
        \includegraphics[scale=0.467]{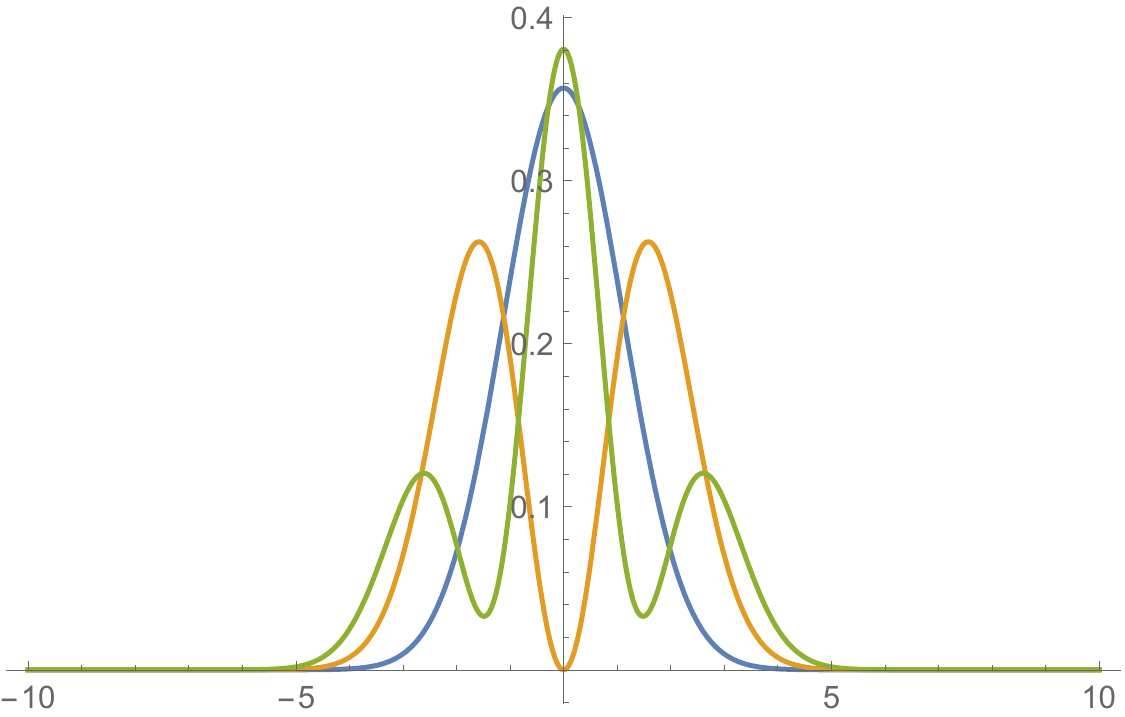} 
\caption{\underline{Left}: plot of $|\Psi_{1,1}(x,1)|^2$ in blue, $|\Psi_{1,1}(x,2)|^2$ in orange and $|\Psi_{1,1}(x,3)|^2$ in green. \underline{Right}: plot of $|\Psi_{1,1}(x,1)|^2$ in blue, $|\Psi_{1,2}(x,1)|^2$ in orange and $|\Psi_{1,3}(x,1)|^2$ in green.}
    \label{fig:list01}
    \end{figure}
    
Setting $a=1$ and now letting $b=2$ yields
\begin{align}
    |\Psi_{1,2}(x,t)|^2 &= \sqrt{\frac{2}{\pi}}\frac{4}{(1+4t^2)^{3/2}}\exp\bigg(-\frac{2x^2}{1+4t^2}\bigg)x^2 = \frac{(2x)^2}{1+4t^2} |\Psi_{1,1}(x,t)|^2, \nonumber \\
    |\Psi_{1,2}(x,1)|^2 &= \frac{1}{\sqrt{5}}f_{\chi,3}\bigg(-\sqrt{\frac{4}{5}}x\bigg) + \frac{1}{\sqrt{5}}f_{\chi,3}\bigg(\sqrt{\frac{4}{5}}x\bigg),
\end{align}
by taking $t=1$ for simplicity and where $f_{\chi,k}(x)$ is the PDF of the $\chi$ distribution with $k$ degrees of freedom given by
\begin{align} \label{eq:chidstribution}
    f_{\chi,k}(x) = \begin{cases}
        \frac{2^{1-\frac{k}{2}}}{\Gamma(\frac{k}{2})} e^{-x^2/2} x^{k-1}, \quad &\mbox{if $x>0$}, \nonumber \\
        0, \quad &\mbox{otherwise}.
    \end{cases}
\end{align}
Keeping $a=t=1$ but now increasing the power from $b=2$ to $b=3$ we have that
\begin{align}
    |\Psi_{1,3}(x,1)|^2 &= \frac{16}{75} \sqrt{\frac{2}{5 \pi}} \exp\bigg(-\frac{2}{5}x^2\bigg) (5-4x^2+x^4) \nonumber \\
    &= \frac{1}{\sqrt{5}} f_{\chi,5}\bigg(-\sqrt{\frac{4}{5}}x\bigg) + \frac{1}{\sqrt{5}} f_{\chi,5}\bigg(\sqrt{\frac{4}{5}}x\bigg)  - \frac{16}{15\sqrt{5}} f_{\chi,3}\bigg(-\sqrt{\frac{4}{5}}x\bigg) - \frac{16}{15\sqrt{5}} f_{\chi,3}\bigg(\sqrt{\frac{4}{5}}x\bigg) \nonumber \\
    & \quad + \frac{32}{15\sqrt{5}} f_{\mathcal{N},0,1} \bigg(\sqrt{\frac{4}{5}} x\bigg).
\end{align}
On the right of Figure \ref{fig:list01} we plot $|\Psi_{1,b}(x,1)|^2$ for $b=1,2,3$ at $t=1$.

In general, one can show that if $b \in \N^+$ then the PDF of $|\Psi_{a,b}(x,t)|^2$ will be a normalized sum of adjusted $\chi$ and $\mathcal{N}$ distributions. The situations where $b$ is not an integer need to be handled separately and will be the subject of future research. Note that $b$ being an integer is consistent with its association as the \textit{number} of degrees of freedom of the $\chi$ distribution.
\subsection{The free particle restricted to positive values and the Maxwell-Boltzman distribution} \label{sec:list02}
We shall now consider the same potential as in Section \ref{sec:list01} but with the restriction $x \in \R^+$. The Feynman propagator is given by \cite[$\mathsection$6.3.1.1]{handbookPI} 
    \begin{align} \label{eq:Kfreepositive}
        K(x'',x';t'',t') &= \int_{x(t')=x'}^{x(t'')=x''} \mathcal{D}_{x>0}x(t) \exp \bigg[\frac{i}{h} \int_{t'}^{t''} \frac{m}{2}{\dot x}^2 dt \bigg] = \frac{m\sqrt{x''x'}}{i \hbar T} \exp\bigg[\frac{im}{2 \hbar T}(x''^2+x'^2)\bigg] I_{1/2}\bigg(\frac{mx''x'}{i\hbar T}\bigg),
    \end{align}
where $I_z(x)$ is the modified Bessel function of the first kind \cite{gradryz, watson}. At $z=\frac{1}{2}$ it reduces to $I_{1/2}(x) = \sqrt{\frac{2}{\pi x}} \sinh x$. The choice of the initial state for the restricted free particle will be
    \begin{align} \label{eq:Psi0freepositivechoice}
        \Psi_{a,b}(y,0) = \bigg(\frac{2^{b+1/2}a^{b-1/2}}{\Gamma(b-\frac{1}{2})}\bigg)^{1/2} \exp(-ay^2)y^{b-1}.
    \end{align}
This is different from \eqref{eq:Psi0freechoice} as $y>0$. Again our restriction is $a>0$ and $b > \frac{1}{2}$ with $b$ integer. For $t>0$ we have that the wave function obeys the equation
\begin{align}
        \Psi_{a,b}(x,t) &= \int_{\R^+} K(x,y;t,0) \Psi_{a,b}(y,0) dy \nonumber \\
        &= \frac{2^{b+\frac{1}{4}} x a^{\frac{b}{2}-\frac{1}{4}} \Gamma
   \left(\frac{b+1}{2}\right) \left(-\frac{i m}{h t}\right)^{\frac{3}{2}}
   e^{\frac{i m x^2}{2 h t}} \left(2 a-\frac{i m}{h t}\right)^{\frac{1}{2}
   (-b-1)}}{(\pi  \Gamma \left(b-\frac{1}{2}\right))^{\frac{1}{2}}} ~_1F_1\bigg(\frac{b+1}{2};\frac{3}{2};\frac{m^2 x^2}{-4 a h^2 t^2+2 i h mt}\bigg).
    \end{align}
As customary, we now may simplify matters by setting $m=\hbar=1$. Let us start by considering $a=b=1$ so that
\begin{align}
    |\Psi_{1,1}(x,t)|^2 = \frac{2 \sqrt{2} e^{-\frac{2 x^2}{4 t^2+1}}}{\sqrt{4 \pi  t^2+\pi }} \bigg| \text{erfi}\bigg(\frac{e^{-\frac{1}{2} i\tan ^{-1}\left(\frac{1}{2 t}\right)} x}{\sqrt{2} \sqrt[4]{4 t^4+t^2}}\bigg) \bigg|^2.
\end{align}
Plots of $|\Psi_{1,1}(x,t)|^2$ for $t=1,2,3$ are provided on the left of Figure \ref{fig:list02a}. Increasing from $b=1$ to $b=2$ shows a very interesting behavior. In this case
\begin{align}
f_{\chi,3}(x) = \sqrt {\frac{2}{\pi }} e^{ - x^2/2}x^2 \quad \textnormal{and} \quad |\Psi_{1,2}(x,t){|^2} = \sqrt {\frac{2}{\pi }} \frac{8}{(1 + 4t^2)^{3/2}}e^{ - \tfrac{2x^2}{1 + 4t^2}}x^2.
\end{align}
This corresponds to the case $k = 3$ in $f_{\chi,k}(x)$ when we further take $t= \frac{\sqrt{3}}{2}$ then
$
|\Psi_{1,2}(x,\tfrac{\sqrt 3}{2}){|^2} = f_{\chi,3}(x)
$.
This is known as the Maxwell-Boltzman distribution and it is plotted in orange on the right of Figure \ref{fig:list02a}. 
Lastly, the full expression for $|\Psi_{1,3}(x,t)|^2$ is too complicated for general $t$, therefore focusing on $t=1$ we obtain
\begin{align}
    |\Psi_{1,3}(x,1)|^2 &= \frac{32 e^{\left(-\frac{2}{5}-\frac{i}{10}\right) x^2}}{75\sqrt{5} \pi } \bigg[ e^{\frac{i x^2}{10}} \left(x^2-(2+i)\right)
   \text{erf}\left(\frac{x}{\sqrt{-4-2
   i}}\right)+\sqrt{-\frac{4+2 i}{\pi }} e^{\frac{x^2}{5}}
   x\bigg] \nonumber \\
   & \quad \times \bigg[\sqrt{2 \pi } \left(x^2-(2-i)\right)
   \text{erf}\bigg(\sqrt{-\frac{1}{5}-\frac{i}{10}}
   x\bigg)+2 \sqrt{-2+i}
   e^{\left(\frac{1}{5}+\frac{i}{10}\right) x^2} x \bigg].
\end{align}
The plot of $|\Psi_{1,1}(x,1)|^2$ is on the right of Figure \ref{fig:list02a} in green. Identifying the PDFs of $|\Psi_{1,1}(x,t)|^2$ and $|\Psi_{1,3}(x,t)|^2$ as well as other elements of this family will be the subject of future research. Figures \ref{fig:list02a} and \ref{fig:list02b} contain several plots for different values of the parameters of $|\Psi_{a,b}(x,t)|^2$.

\begin{figure}[h] 
   	\includegraphics[scale=0.467]{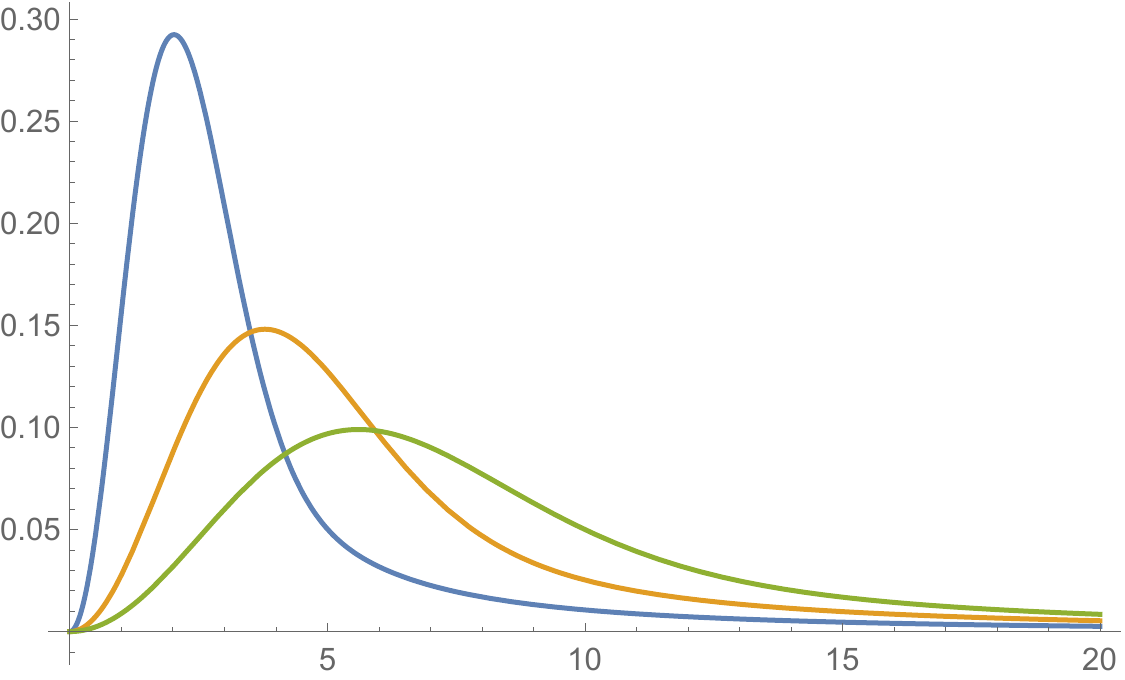} 
        \includegraphics[scale=0.467]{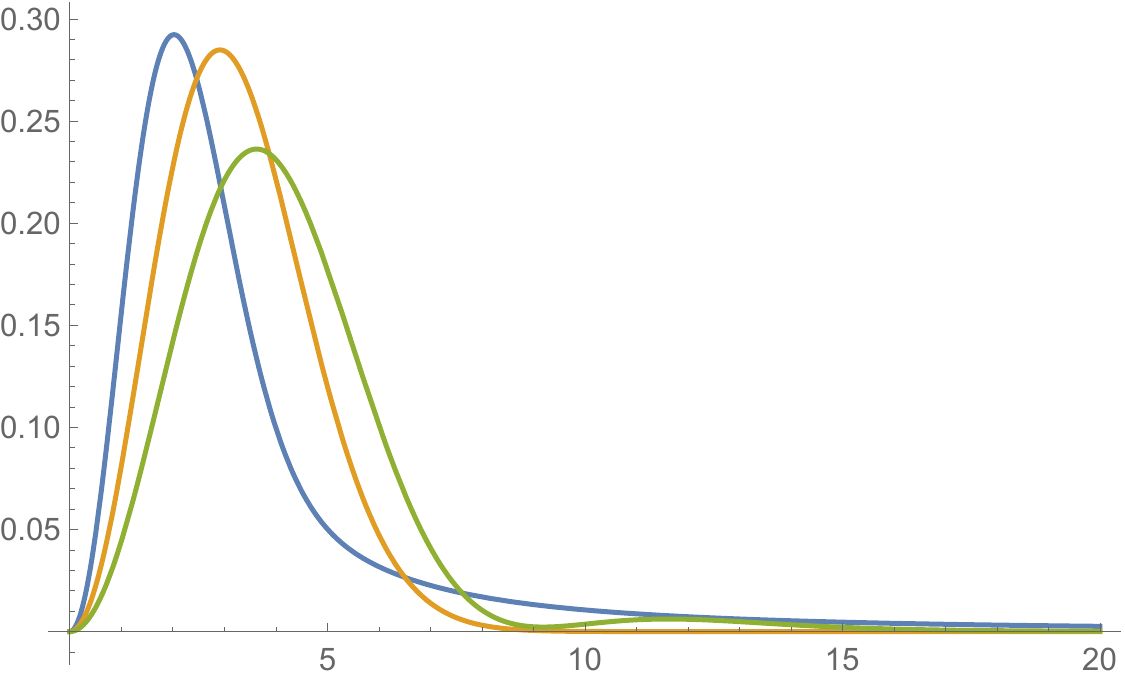} 
\caption{\underline{Left}: plot of $|\Psi_{1,1}(x,1)|^2$ in blue, $|\Psi_{1,1}(x,2)|^2$ in orange and $|\Psi_{1,1}(x,3)|^2$ in green. \underline{Right}: plot of $|\Psi_{1,1}(x,1)|^2$ in blue, $|\Psi_{1,2}(x,2)|^2$ in orange and $|\Psi_{1,3}(x,3)|^2$ in green.}
    \label{fig:list02a}
\end{figure}

\begin{figure}[h] 
   	\includegraphics[scale=0.467]{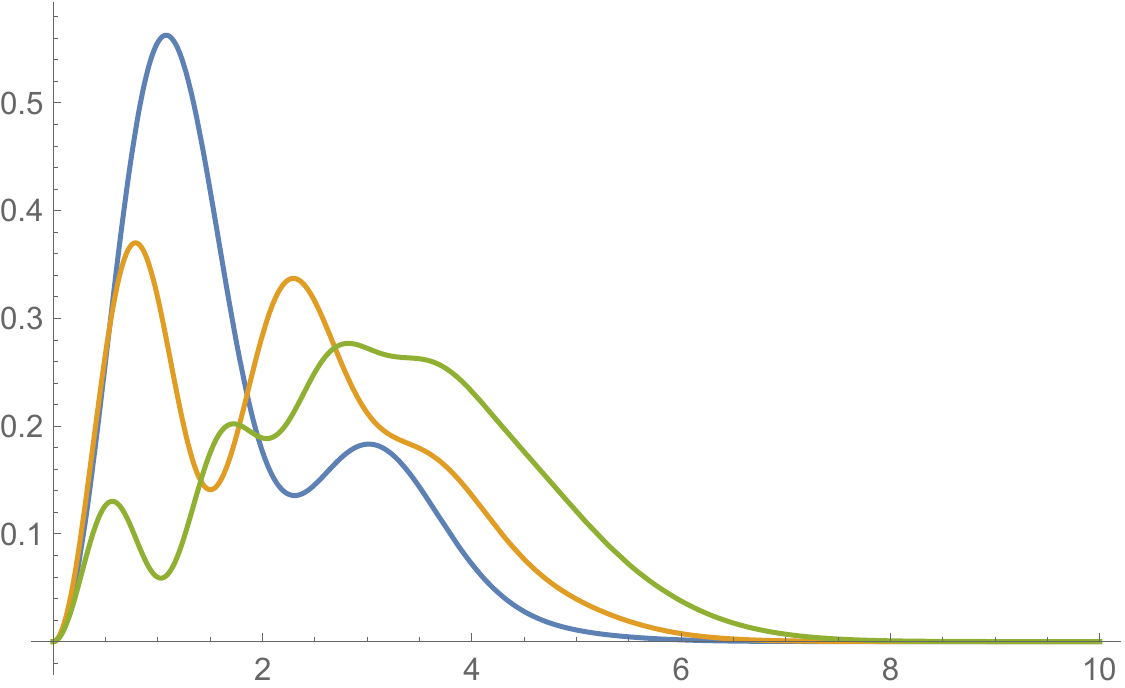} 
        \includegraphics[scale=0.467]{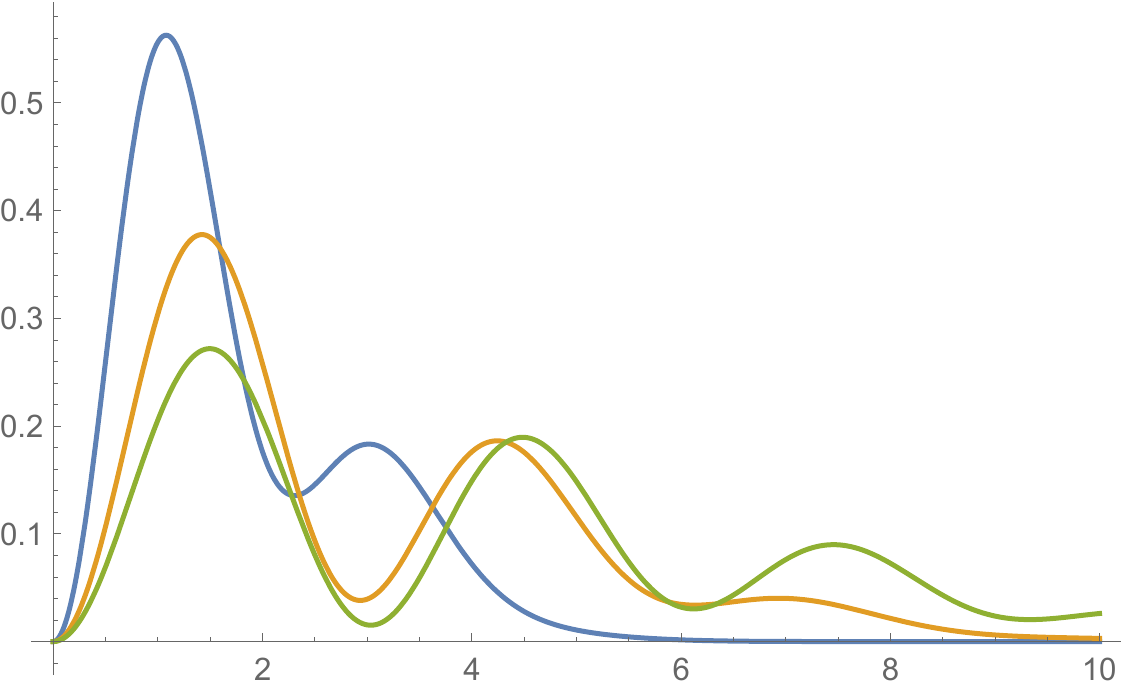} 
\caption{\underline{Left}: plot of $|\Psi_{1,5}(x,1)|^2$ in blue, $|\Psi_{1,10}(x,1)|^2$ in orange and $|\Psi_{1,20}(x,1)|^2$ in green. \underline{Right}: plot of $|\Psi_{1,5}(x,1)|^2$ in blue, $|\Psi_{1,10}(x,2)|^2$ in orange and $|\Psi_{1,20}(x,3)|^2$ in green.}
    \label{fig:list02b}
\end{figure}

\subsection{The free particle confined to an interval and twisted sums of semi-wrapped normal distributions} \label{sec:list03}
For the free particle in the box $-b <  x < b$ the Feynman propagator  is given by (\cite[$\mathsection$6.3.2.1]{handbookPI})
    \begin{align} \label{eq:Kfreebox}
        K(x'',x';t'',t') &= \int_{x(t')=x'}^{x(t'')=x''} \mathcal{D}_{|x|<b}x(t) \exp \bigg[\frac{i}{h} \int_{t'}^{t''} \frac{m}{2}{\dot x}^2 dt \bigg] \nonumber \\
        &= \frac{1}{b} \sum_{n=1}^\infty \sin \bigg(\frac{\pi n}{2b}(x''+b)\bigg)\sin \bigg(\frac{\pi n}{2b}(x'+b)\bigg) \exp\bigg(-i\hbar T \frac{\pi^2 n^2}{8mb^2} \bigg).
    \end{align}
The choice for the initial state is
    \begin{align}
        \Psi_b(y,0) = \bigg(\frac{2}{\pi}\bigg)^{1/4}(\operatorname{erf}(\sqrt{2}b))^{-1/2} \exp(-y^2),
    \end{align}
    where $\operatorname{erf}$ is the error function. We proceed to integrate term by term. First we need the identity
    \begin{align} \label{eq:auxiliaryfreebox}
        \int_{ - b}^b \sin \left( {\frac{{\pi n}}{{2b}}(y + b)} \right) e^{-y^2}dy &= \frac{{\sqrt \pi  }}{2} e^{-\frac{\pi^2 n^2}{16b^2}} \sin \left( {\frac{{\pi n}}{2}} \right) \left[ {\operatorname{erf} \left( {b - \frac{{i\pi n}}{{4b}}} \right) + \operatorname{erf} \left( {b + \frac{{i\pi n}}{{4b}}} \right)} \right].
    \end{align}
    With the aid of \eqref{eq:auxiliaryfreebox} we may now then write
    \begin{align}
  \Psi_b (x,t) &= \int_{ - b}^b {K(x,y;t,0)\Psi_b(y,0)dy}  \nonumber \\
   &= \frac{1}{b}{\left( {\frac{2}{\pi }} \right)^{1/4}}{(\operatorname{erf} (\sqrt 2 b))^{ - 1/2}}\sum\limits_{n = 1}^\infty  {\sin \left( {\frac{{\pi n}}{{2b}}(x + b)} \right)} \exp \left( { - it\frac{{{\pi ^2}{n^2}}}{{8{b^2}}}} \right) \int_{ - b}^b {\sin \left( {\frac{{\pi n}}{{2b}}(y + b)} \right)\exp ( - {y^2})dy}  \nonumber \\
   &= \frac{{\sqrt \pi  }}{{2b}}{\left( {\frac{2}{\pi }} \right)^{1/4}}{(\operatorname{erf} (\sqrt 2 b))^{ - 1/2}} \nonumber \\
   &\quad \times \sum\limits_{n = 1}^\infty  {\sin \left( {\frac{{\pi n}}{{2b}}(x + b)} \right)} \exp \left( { - it\frac{{{\pi ^2}{n^2}}}{{8{b^2}}}} \right)\exp \left( { - \frac{{{\pi ^2}{n^2}}}{{16{b^2}}}} \right)\sin \left( {\frac{{\pi n}}{2}} \right) \left[ {\operatorname{erf} \left( {b - \frac{{i\pi n}}{{4b}}} \right) + \operatorname{erf} \left( {b + \frac{{i\pi n}}{{4b}}} \right)} \right] .
\end{align}
We now plot several examples of the distribution $|\Psi_b(x,t)|^2$ in Figure \ref{fig:list03}.

    \begin{figure}[h] 
   	\includegraphics[scale=0.467]{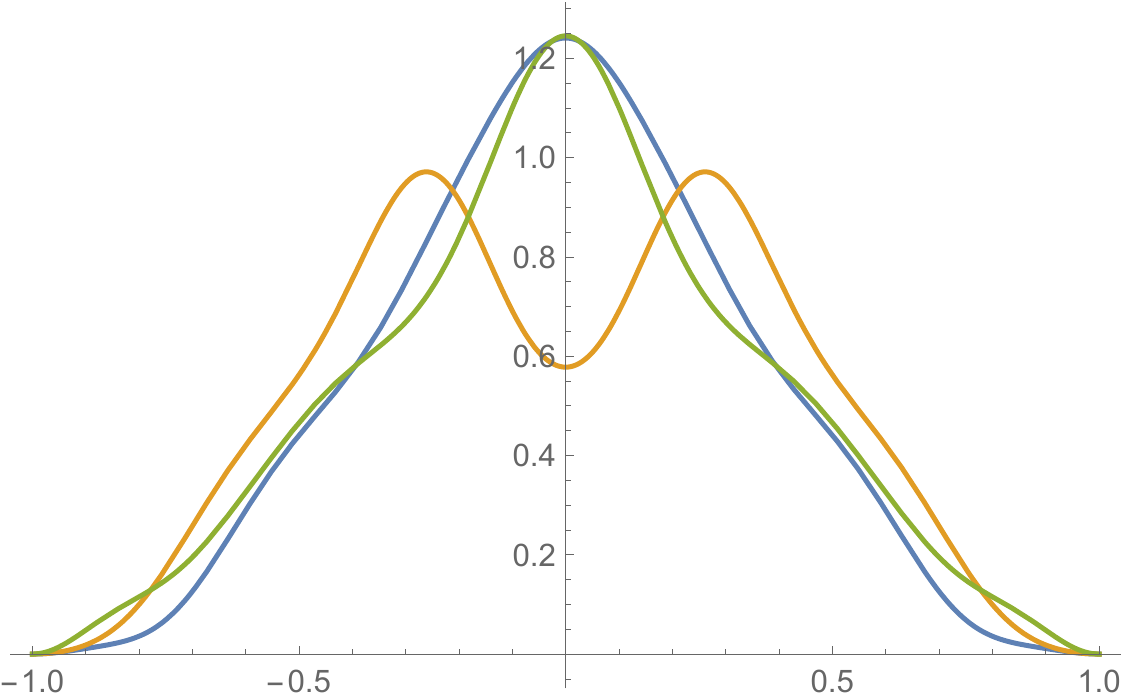} 
        \includegraphics[scale=0.467]{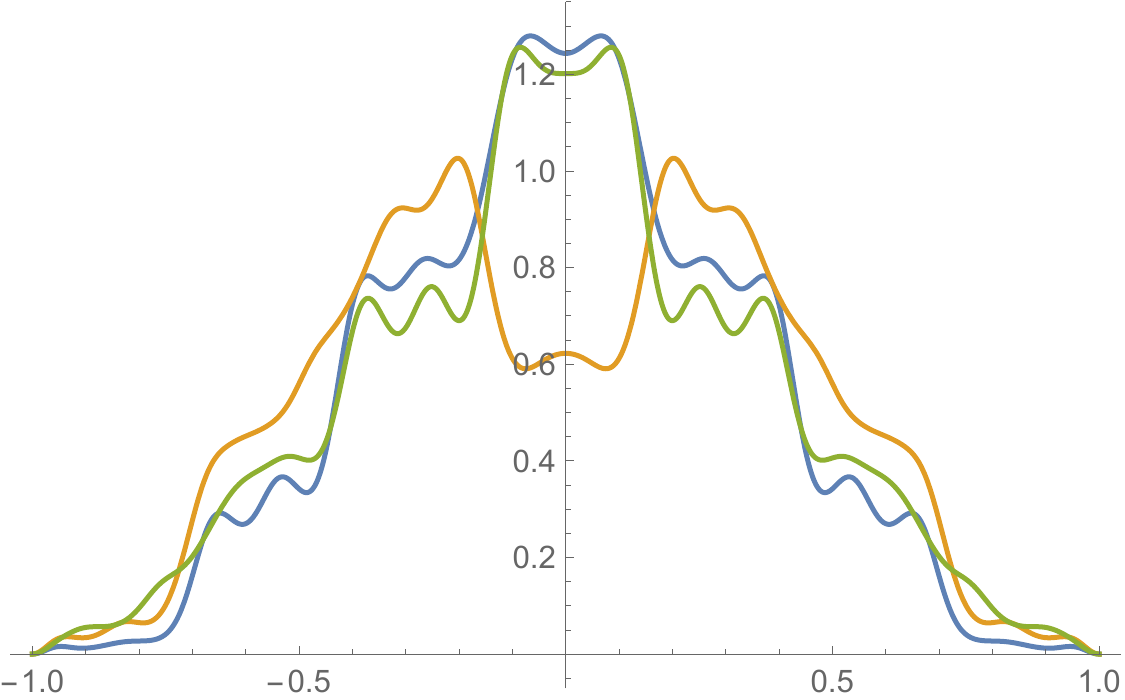} 
\caption{Plot of $|\Psi(x,t)|^2$ with $t=1,2,3$ in blue, orange and green with $10$ terms (left) and $50$ terms (right) for $b=1$.}
    \label{fig:list03}
    \end{figure}
\subsection{The radial potential} \label{sec:list04}
The potential of the radial harmonic oscillator in one dimension is defined by
    \begin{align} \label{eq:VRHO}
        V_{\operatorname{RHO}}(x)= \frac{m}{2}\omega^2 x^2 + \hbar^2 \frac{\lambda^2-\frac{1}{4}}{2mx^2}
    \end{align}
for $x>0$ and for $\real(\lambda)>-1$ but for $|\lambda|<\frac{1}{2}$ both signs must be taken into account. The associated Hamiltonian and Lagrangian are therefore
\begin{align}
        \hat H = -\frac{\hbar^2}{2m}\frac{\partial^2}{\partial x^2} + \frac{m}{2}\omega^2 x^2 + \hbar^2 \frac{\lambda^2-\frac{1}{4}}{2mx^2} \quad \textnormal{and} \quad  \mathcal{L} = \frac{m}{2}{\dot x}^2 - \frac{m}{2}\omega^2 x^2 - \hbar^2 \frac{\lambda^2-\frac{1}{4}}{2mx^2}.
\end{align}
The Feynman propagator \cite[$\mathsection$6.4.1]{handbookPI} is given by
    \begin{align} \label{eq:RadialPropagator}
    K(x'',x';t'',t') &= \int_{x(t') = x'}^{x(t'') = x''} \mathcal{D}x(t) \exp\bigg[\frac{i}{\hbar} \int_{t'}^{t''}  \bigg(\frac{m}{2}{\dot x}^2 - \frac{m}{2}\omega^2 x^2 - \hbar^2 \frac{\lambda^2-\frac{1}{4}}{2mx^2}\bigg) dt\bigg] \nonumber \\
    &= \frac{2m\omega}{\hbar}\sqrt{x''x'} \sum_{n=0}^\infty \frac{n!}{\Gamma(n+\lambda+1)}\bigg(\frac{m\omega}{\hbar}x''x'\bigg)^\lambda \exp\bigg(\frac{-m\omega}{2\hbar}(x''^2+x'^2)\bigg)  L_n^{(\lambda)}\bigg(\frac{m\omega}{\hbar}x''^2\bigg)L_n^{(\lambda)}\bigg(\frac{m\omega}{\hbar}x'^2\bigg) \nonumber \\
    &= \frac{m\omega\sqrt{x''x'}}{i\hbar \sin(\omega T)} \exp\bigg[-\frac{m\omega}{2i\hbar}(x''^2+x'^2)\cot(\omega T)\bigg] I_{\lambda}\bigg( \frac{m\omega x''x'}{i\hbar \sin(\omega T)}\bigg),
\end{align}
where $I_z(x)$ is the modified Bessel function of the first kind and $L_n^{\alpha}(x)$ are the Laguerre polynomials. The last equality is known as the Hardy-Hille formula \cite{gradryz,handbookPI,watson,higher1,higher2}. Our first choice for initial state $\Psi_{a,b}(y,0)$ will be \eqref{eq:Psi0freepositivechoice} with the restrictions $a>0$ and $b > \frac{1}{2}$. In this case the evolution for $t>0$ is given by
    \begin{align}
        \Psi_{a,b}(x,t,\lambda) &= \int_{\R^+} K(x,y;t,0)\Psi_{a,b}(y,0)dy \nonumber \\
        &= -\frac{i m \omega  a^{\frac{b}{2}-\frac{1}{4}} 2^{b-\frac{\lambda
   }{2}-\frac{1}{2}} x^{\lambda +\frac{1}{2}} \csc (t \omega )}{h \sqrt{\Gamma \left(b-\frac{1}{2}\right)} \Gamma (\lambda +1)}
    \left(2
   a-\frac{i m \omega  \cot (t \omega )}{h}\right)^{\frac{1}{4} (-2
   b-2 \lambda -1)}  \Gamma \left(\frac{b}{2}+\frac{\lambda }{2}+\frac{1}{4}\right)
    \nonumber \\
   & \quad \times \left(-\frac{m^2 \omega ^2 \csc ^2(t \omega
   )}{h^2}\right)^{\lambda /2} e^{\frac{i m x^2 \omega  \cot (t \omega )}{2 h}} ~_1F_1\left(\frac{1}{4} (2 b+2
   \lambda +1);\lambda +1;-\frac{m^2 x^2 \omega ^2 \csc ^2(t \omega
   )}{4 a h^2-2 i h m \omega  \cot (t \omega )}\right).
    \end{align}
We now set $\hbar = m = \omega = 1$. We can now use this to plot several $|\Psi_{a,b}(x,t,\lambda)|^2$ in Figure \ref{fig:list04a}.

    \begin{figure}[h] 
   	\includegraphics[scale=0.467]{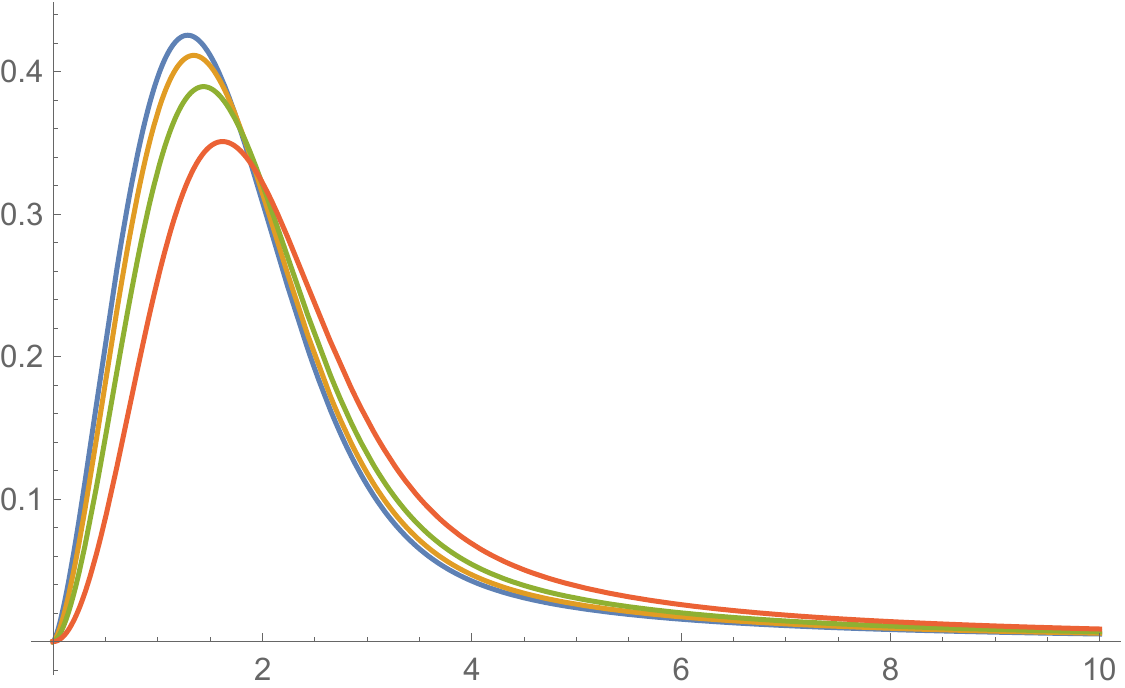} 
        \includegraphics[scale=0.467]{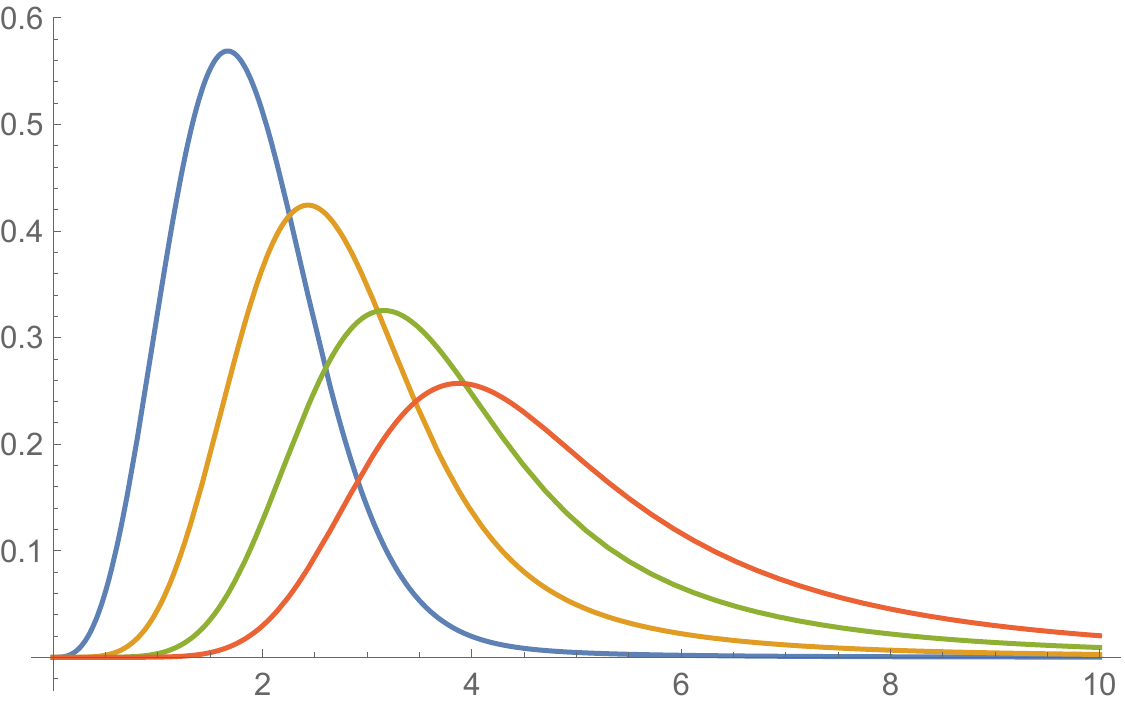} 
\caption{\underline{Left}: plot of $|\Psi_{1,1}(\frac{1}{5},x,1)|^2$ in blue, $|\Psi_{1,1}(\frac{1}{4},x,1)|^2$ in orange, $|\Psi_{1,1}(\frac{1}{3},x,1)|^2$ in green and $|\Psi_{1,1}(\frac{1}{2},x,1)|^2$ in red. \underline{Right}: plot of $|\Psi_{1,2}(1,x,1)|^2$ in blue, $|\Psi_{1,2}(2,x,1)|^2$ in orange, $|\Psi_{1,2}(3,x,1)|^2$ in green and $|\Psi_{1,2}(4,x,1)|^2$ in red.}
    \label{fig:list04a}
    \end{figure}

If we work with $b=1$, then we can write
\begin{align}
    |\Psi_{a,1}(x,t,\lambda)|^2 &= \frac{\sqrt{a} 2^{\frac{7}{4}-\lambda } \Gamma \left(\frac{1}{4} (2 \lambda
   +3)\right)^2 x^{2 \lambda +1} \left(4 a^2+\cot ^2(t)\right)^{-\lambda /2} \left|
   \csc (t)\right| ^{2 \lambda +\frac{1}{2}}}{\sqrt{\pi } \Gamma (\lambda +1)^2
   \left(\left(1-4 a^2\right) \cos (2 t)+4 a^2+1\right)^{3/4}} \nonumber \\
   &\quad \times \bigg| ~_1F_1\left(\frac{1}{4} (2 \lambda +3);\lambda +1;-\frac{x^2 \csc
   ^2(t)}{4 a-2 i \cot (t)}\right)\bigg|^2.
\end{align}
This expression is -- obviously -- positive and it also -- not so obviously -- integrates to $1$ over $x \in \R^+$, hence it is a PDF, as expected. However, it is not entirely clear if it is a recognizable `named' PDF. Identifying the PDF of $|\Psi_{a,b}(x,t,\lambda)|^2$ is a subject for future research. 

Up until this point, the computation of $\Psi(x,t)$ from $K(x,y;t,0)$ and $\Psi(x,0)$ has been performed with computer software. If \textit{instead} of \eqref{eq:Psi0freepositivechoice} we now try  
    \begin{align} \label{eq:Psi0RHO2}
        \tilde\Psi_{a,b}(y,0) = \frac{2^{\frac{1}{2} (2 b-1)} a^{\frac{1}{2} (2 b-1)} e^{-a y}
   y^{b-1}}{\sqrt{\Gamma (2 b-1)}},
    \end{align}
then advanced mathematical software cannot compute $\tilde{\Psi}_{a,b}(x,t) = \int_{\R^+} K(x,y;t,0) \tilde\Psi_{a,b}(y,0)dy$ except in very special cases. Therefore, we need to integrate this case `by hand'. 
\begin{lemma} \label{lem:auxRHO}
Define the integral
\begin{align}
    \Upsilon(\alpha,\beta,\gamma,\lambda,s) := \int_0^\infty \exp(-\alpha y^2 - \beta y) I_\lambda(\gamma y)y^{s-1}ds
\end{align}
for $\real(\alpha), \real(\beta), \real(\gamma), \real(\lambda) >0$ and $\real(s)>1$. Then one has that
\begin{align}
    \Upsilon(\alpha,\beta,\gamma,\lambda,s) = \sum_{n=0}^\infty &\frac{(-1)^{n+1} \gamma ^{\lambda } (i \gamma )^{2 n} 2^{-\lambda -2 n-1} \alpha ^{-\frac{\lambda
   }{2}-n-\frac{s}{2}-\frac{1}{2}}}{n! \Gamma (n+\lambda +1)}   \bigg[\sqrt{\alpha } \Gamma \bigg(\frac{1}{2} (2 n+s+\lambda )\bigg) \, _1F_1\bigg(\frac{1}{2} (2
   n+s+\lambda );\frac{1}{2};\frac{\beta ^2}{4 \alpha }\bigg) \nonumber \\
   & -\beta  \Gamma \bigg(\frac{1}{2} (2
   n+s+\lambda +1)\bigg) \, _1F_1\bigg(\frac{1}{2} (2 n+s+\lambda +1);\frac{3}{2};\frac{\beta
   ^2}{4 \alpha }\bigg)\bigg], \nonumber
\end{align}
where, we recall, $_1F_1$ is the confluent hypergeometric function.
\end{lemma}
\begin{proof}
We shall employ \cite[Eq. (2.1.16)]{titchmarshFourier}
\begin{align}
    \int_0^\infty f(x)g(x) x^{s-1}dx = \frac{1}{2\pi i} \int_{k-i\infty}^{k+i\infty} \mathfrak{F}(w)\mathfrak{G}(s-w)dw
\end{align}
where $\mathfrak{F}$ and $\mathfrak{G}$ are the Mellin transforms of $f$ and $g$, respectively, and $k \in \R$ needs to be chosen according to $f$ and $g$. We choose $f(x) = \exp(-\alpha x^2 - \beta x)$ and $g(x) = I_\lambda(\gamma x)$. The Mellin transforms can be found in \cite[$\mathsection$1.10]{oberhettinger}
\begin{align}
    \mathfrak{F}(w) = \frac{1}{2} \alpha ^{-\frac{w}{2}-\frac{1}{2}} \left(\sqrt{\alpha }
   \Gamma \left(\frac{w}{2}\right) \,
   _1F_1\left(\frac{w}{2};\frac{1}{2};\frac{\beta ^2}{4 \alpha
   }\right)-\beta  \Gamma \left(\frac{w+1}{2}\right) \,
   _1F_1\left(\frac{w+1}{2};\frac{3}{2};\frac{\beta ^2}{4 \alpha
   }\right)\right),
\end{align}
as well as
\begin{align}
    \mathfrak{G}(s-w) = \int_0^\infty x^{s-w-1} g(x) dx = \frac{\gamma ^{\lambda } 2^{s-w-1} (i \gamma )^{-\lambda -s+w}
   \Gamma \left(\frac{s}{2}-\frac{w}{2}+\frac{\lambda
   }{2}\right)}{\Gamma \left(-\frac{s}{2}+\frac{w}{2}+\frac{\lambda
   }{2}+1\right)}.
\end{align}
Let $0 < k < 1$ so that by closing the contour to the right we pick up the poles of $\Gamma(\frac{s}{2}-\frac{w}{2}+\frac{\lambda}{2})$ at $w=\lambda+2n+s$. Noting that the orientation in this contour is negative we see that the residues are given by
\begin{align}
    \Omega(n) = -\frac{1}{2\pi i} &\oint_{\gamma(n)} \frac{1}{2} \alpha ^{-\frac{w}{2}-\frac{1}{2}} \frac{\gamma ^{\lambda } 2^{s-w-1} (i \gamma )^{-\lambda -s+w}
   \Gamma \left(\frac{s}{2}-\frac{w}{2}+\frac{\lambda
   }{2}\right)}{\Gamma \left(-\frac{s}{2}+\frac{w}{2}+\frac{\lambda
   }{2}+1\right)} \nonumber \\  
   & \times \bigg(\sqrt{\alpha }
   \Gamma \left(\frac{w}{2}\right) \,
   _1F_1\left(\frac{w}{2};\frac{1}{2};\frac{\beta ^2}{4 \alpha
   }\right)-\beta  \Gamma \left(\frac{w+1}{2}\right) \,
   _1F_1\left(\frac{w+1}{2};\frac{3}{2};\frac{\beta ^2}{4 \alpha
   }\right)\bigg) dw,
\end{align}
where $\gamma(n)$ denotes a small circle with center $1+2n+s$. Recalling that the residues of $\Gamma(s)$ at $s=-n$ for $n \in \N^{\ge 0}$ are given by $\frac{(-1)^n}{n!}$ we obtain by Cauchy's residue theorem
\begin{align}
    \Omega(n) &= \frac{(-1)^{n+1} \gamma ^{\lambda } (i \gamma )^{2 n} 2^{-\lambda -2 n-1} \alpha ^{-\frac{\lambda
   }{2}-n-\frac{s}{2}-\frac{1}{2}}}{n! \Gamma (n+\lambda +1)}  \bigg[\sqrt{\alpha } \Gamma \bigg(\frac{1}{2} (2 n+s+\lambda )\bigg) \, _1F_1\bigg(\frac{1}{2} (2
   n+s+\lambda );\frac{1}{2};\frac{\beta ^2}{4 \alpha }\bigg) \nonumber \\
   & \quad -\beta  \Gamma \bigg(\frac{1}{2} (2
   n+s+\lambda +1)\bigg) \, _1F_1\bigg(\frac{1}{2} (2 n+s+\lambda +1);\frac{3}{2};\frac{\beta
   ^2}{4 \alpha }\bigg)\bigg]. \nonumber
\end{align}
We now sum over over all the residues in the rectangular contour so that $\Upsilon(\alpha,\beta,\gamma,\lambda,s) = \sum_{n \ge 0} \Omega(n)$, as claimed in the statement of the lemma.
\end{proof}
To illustrate this, we pick $\lambda=a=1$ as well as $b=2$. This leads to $\tilde\Psi_{1,2}(y,0)=2e^{-y}y$ and
\begin{align}
  {\tilde \Psi }_{1,2}(x,t) &= \int_0^\infty  {K(x,y;t,0){{\tilde \Psi }_{1,2}}(y,0)dy}  \nonumber \\
   &= \int_0^\infty  {i\sqrt {xy} (\csc t)\exp \left[ {\frac{1}{2}i({x^2} + {y^2})\cot t} \right]{I_1}\left( {ixy\csc t} \right)2{e^{ - y}}ydy}  \nonumber \\
   &= 2i\sqrt x (\csc t)\exp \left[ {\frac{{i\cot t}}{2}{x^2}} \right]\int_0^\infty  {{y^{3/2}}\exp \left[ {\frac{{i\cot t}}{2}{y^2} - y} \right]{I_1}\left( {ixy\csc t} \right)dy}  \nonumber \\
   &= 2i\sqrt x (\csc t)\exp \left[ {\frac{{i\cot t}}{2}{x^2}} \right] \Upsilon\bigg(-\frac{i\cot t}{2},1,ix \csc t, 1, \frac{5}{2}\bigg),
\end{align} 
and we can now employ Lemma \ref{lem:auxRHO} to compute $\Upsilon$ in terms of a power series.

    \begin{figure}[h] 
   	\includegraphics[scale=0.467]{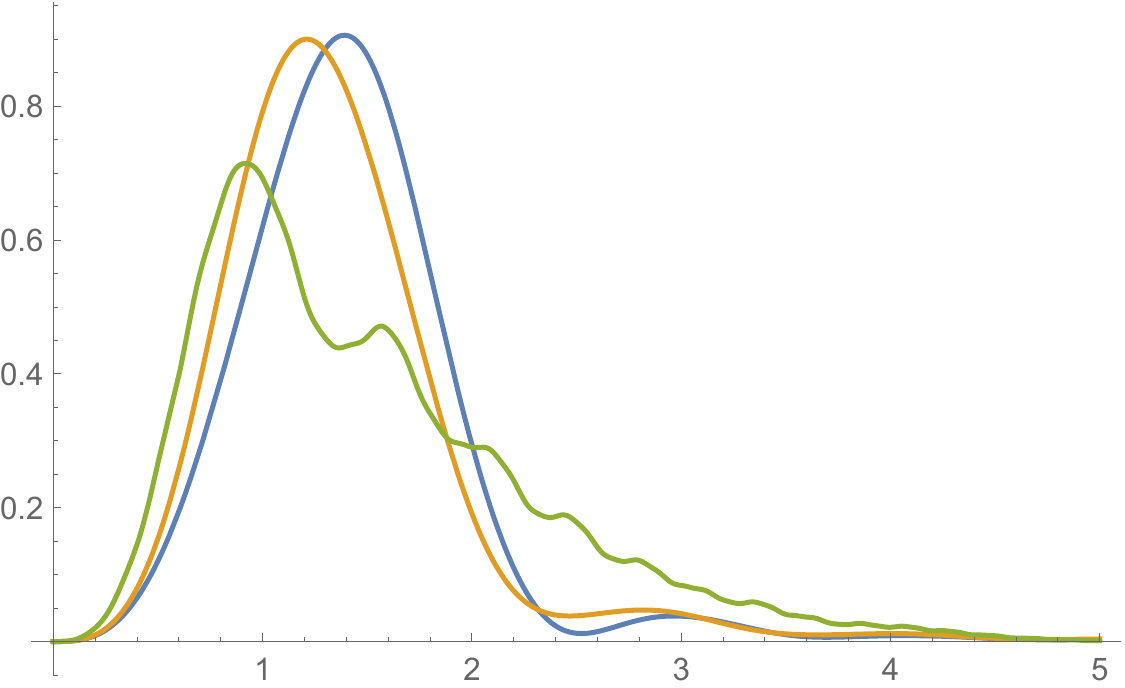} 
        \includegraphics[scale=0.467]{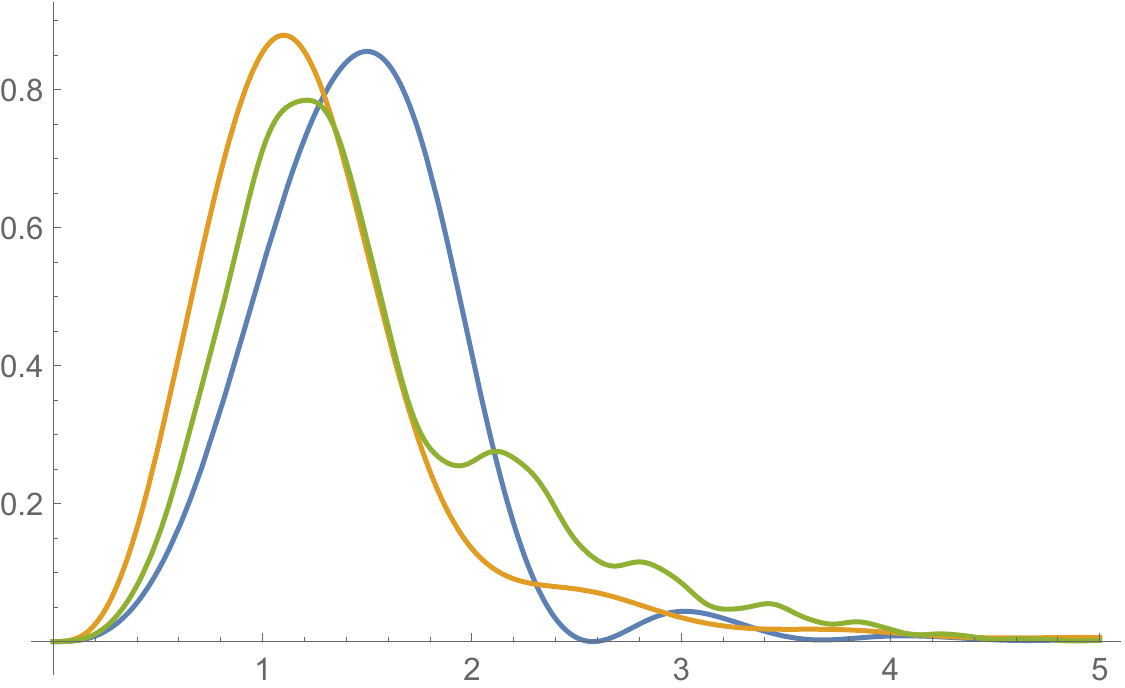} 
\caption{\underline{Left}: plot of $|\Psi_{1,2}(x,1)|^2$ in blue, $|\Psi_{1,2}(x,2)|^2$ in orange and $|\Psi_{1,2}(x,3)|^2$ in green. \underline{Right}: plot of $|\Psi_{1,2}(x,4)|^2$ in blue, $|\Psi_{1,2}(x,5)|^2$ in orange and $|\Psi_{1,2}(x,6)|^2$ in green.}
    \label{fig:list04b}
    \end{figure}

We illustrate ${\tilde \Psi }_{1,2}(x,t)$ for different values of $t$ in Figure \ref{fig:list04b}. This choice of initial wave function $\tilde\Psi$ is significant because \eqref{eq:Psi0RHO2} is an instance of a \textit{Gamma} distribution rather than a modified \textit{normal} distribution.

\subsection{The half linear potential} \label{sec:list05}
This potential takes the form
    \begin{align}
        V(x)= 
        \begin{cases}
        kx  \quad &\mbox{if $x>0$},  \\
        \infty \quad &\mbox{if $x \le 0$}.
        \end{cases}
    \end{align}
Here $k=q\varepsilon$ is a force and $q$ is a charge. The quantity $\varepsilon$ is given by the distance from $x=0$ to $V(x) = q \varepsilon x$. Our choice for the initial state is going to be 
    \begin{align} \label{eq:initialstatelinear01}
        \Psi_{a,b}(x,0) = \bigg(\frac{(2a)^{2b-1}}{\Gamma(2b-1)}\bigg)^{1/2} \exp(-ax) x^{b-1}
    \end{align}
for $\real(a) >0$ and $\real(b) > \frac{1}{2}$. We conjecture that the spectrum of this potential is discrete as we could not find this result in the literature. The normalized stationary states are given by 
  \begin{align} \label{eq:PsiNhalfLinear}
        \Psi_n(x) = \sqrt[6]{2} \bigg(\sqrt[3]{\frac{m k
   }{h^2}} \bigg\{\operatorname{Ai}'\bigg(-\sqrt[3]{\frac{2m k}{h^2}}\frac{
   E_n}{k}\bigg)\bigg\}^{-2}\bigg)^{1/2} \operatorname{Ai}\bigg(\sqrt[3]{\frac{2m k
   }{h^2}} \bigg(x-\frac{E_n}{k}\bigg)\bigg),
    \end{align}
where $\operatorname{Ai}$ is the Airy function and $\operatorname{Ai}'$ its derivative \cite{gradryz}, and the eigenvalues are 
    \begin{align} \label{eq:eNhalfLinear}
    E_n = - \Upsilon_n \bigg(\frac{\hbar^2 k^2}{2m}\bigg)^{1/3},
    \end{align}
    with $\Upsilon_n$ denoting the zeros of the $\operatorname{Ai}$ function. The first few zeros are 
    \begin{align}
    \{ -2.33811, -4.08795, -5.52056, -6.78671, -7.94413, -9.02265, \ldots \}.
    \end{align}
One can verify that \eqref{eq:PsiNhalfLinear} and \eqref{eq:eNhalfLinear} satisfy the TISE equation \eqref{eq:PsixtPsin}. We now plot $\Psi_n$ and $E_n$ in Figure \ref{fig:list05a}.

\begin{figure}[h]
    \includegraphics[scale=0.467]{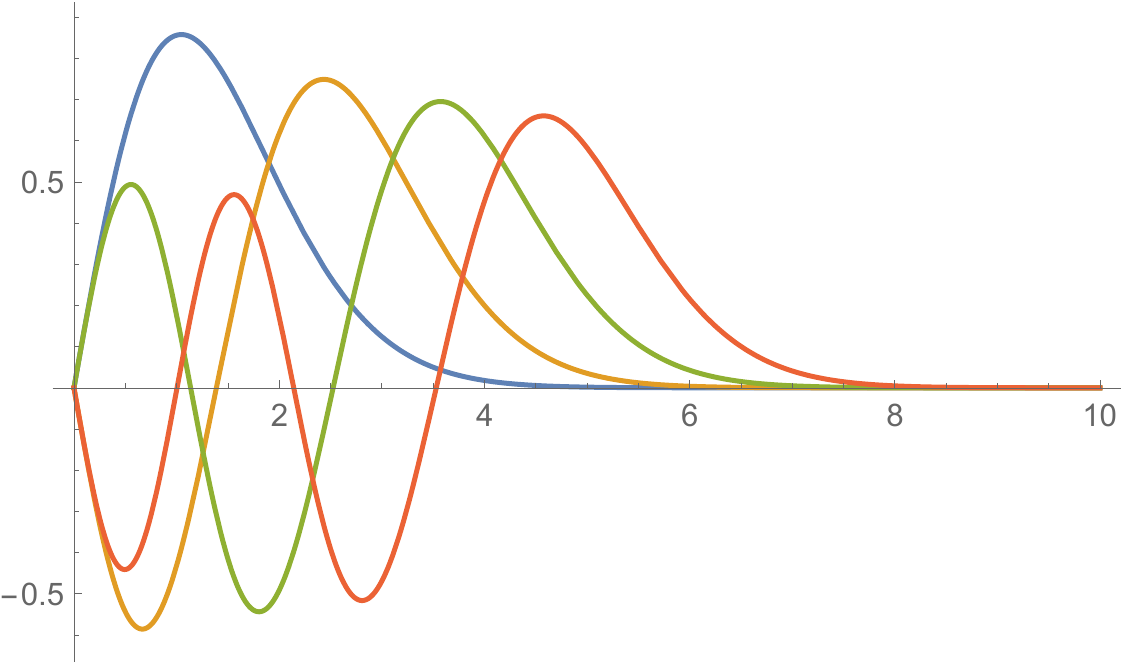} 
    \includegraphics[scale=0.467]{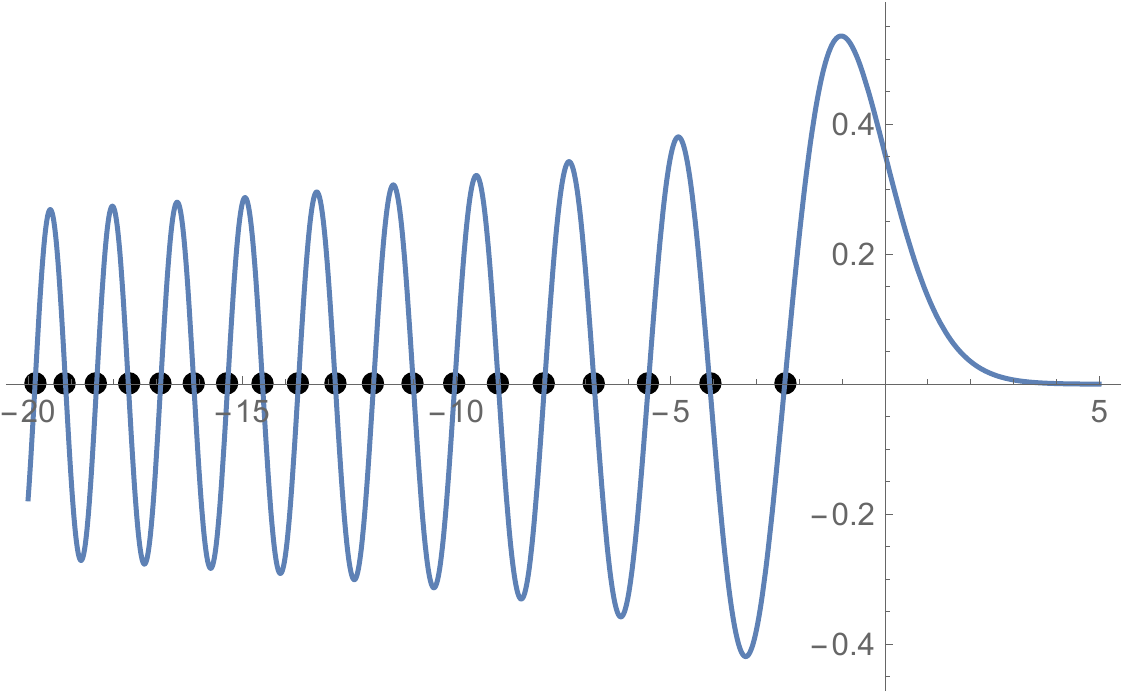} 
\caption{\underline{Left}: plot of $\Psi_1(x)$ in blue, $\Psi_2(x)$ in orange, $\Psi_3(x)$ in green, and $\Psi_4(x)$ in blue. \underline{Right}: plot of $\operatorname{Ai}(x)$ for $-20 \le x \le 5$ with $\Upsilon_n$ marked out in black.}
    \label{fig:list05a}
\end{figure}

In general, the wave function for $t>0$ is given by
\begin{align} \label{eq:waveAiry}
        \Psi_{a,b}(x,t) = \sum_{n=1}^\infty c_n(a,b) \Psi_n(x) e^{-iE_nt/\hbar},
\end{align}
where the normalizing constants $c_n$ can be found by using
\begin{align}
    c_n(a,b) &= \int_{\R^+}\Psi_n^*(x) \Psi_{a,b}(x,0)dx .
\end{align}
This integral is difficult to evaluate analytically due to the presence of the term $x-\frac{E_n}{k}$ in the argument of the Airy function, but the values can easily be found through numerical integration. For example the first four instances with $a=b=1$ are given by
\begin{align}
c_{1}(1,1) \approx 0.714614,\; c_{2}(1,1) \approx -0.230831, \; c_{3}(1,1) \approx 0.227638 \quad \textnormal{and} \quad c_{4}(1,1) \approx -0.168829.
\end{align}
Now let us set the truncated wave function
\begin{align}
\Psi_{N,a,b}(x,0) = \sum_{n=1}^N c_n(a,b) \Psi_n(x) = \sum_{n=1}^N \bigg(\int_{\R^+}\Psi_n^*(y) \Psi_{a,b}(y,0)dy\bigg) \Psi_n(x).
\end{align}
Letting $N \to \infty$ yields $\Psi_{N,a,b}(x,0) \to \Psi_{a,b}(x,0)$. We can illustrate this by taking $N \in \{5,10,20\}$ and plotting the resulting approximations.

\begin{figure}[h] 
   	\includegraphics[scale=0.467]{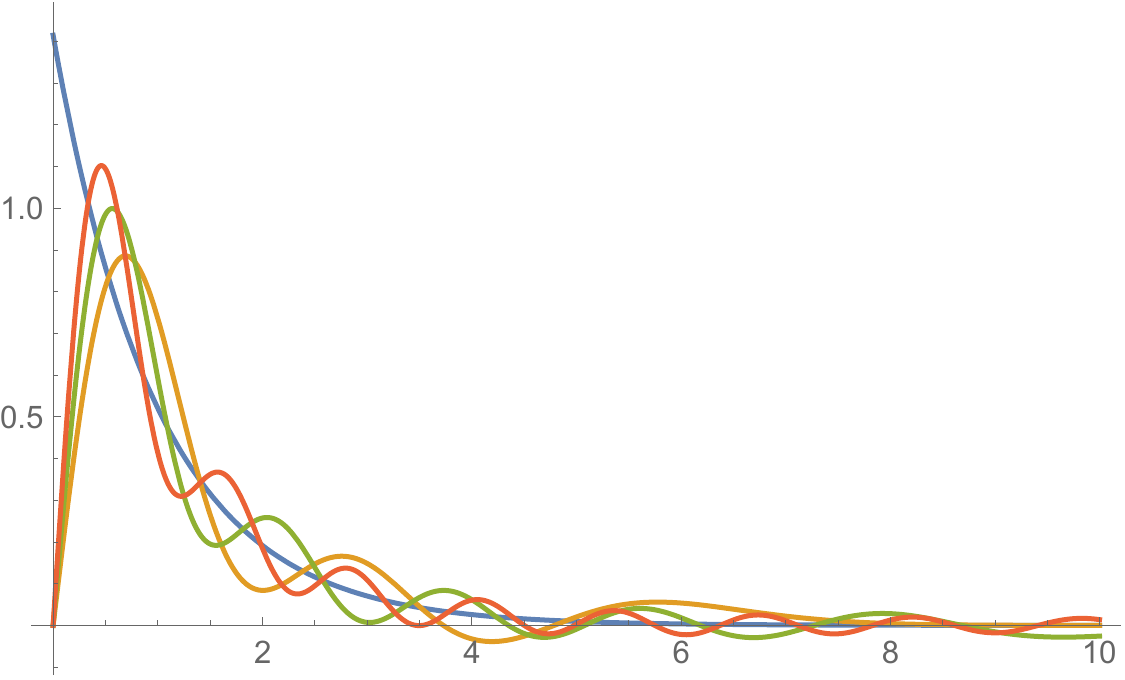} 
        \includegraphics[scale=0.467]{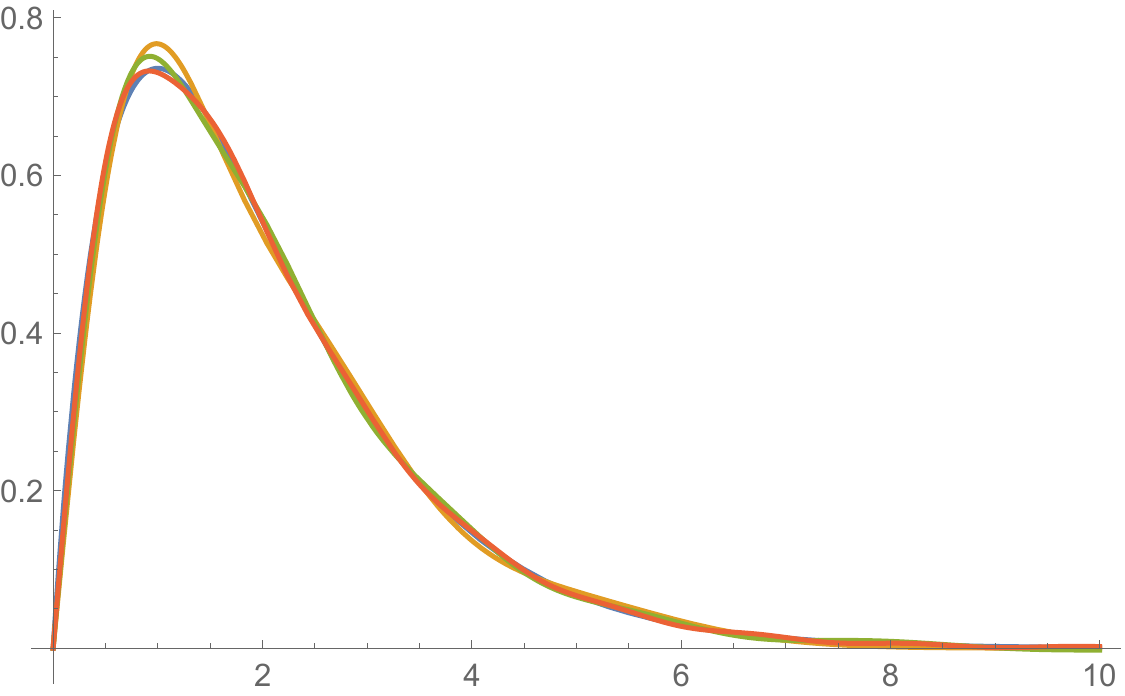} 
\caption{\underline{Left}: plot of $\Psi_{0,1,1}(x)$ in blue, $\Psi_{5,1,1}(x,0)$ in orange, $\Psi_{10,1,1}(x,0)$ in green and $\Psi_{20,1,1}(x,0)$ in red. \underline{Right}: plot of $\Psi_{0,1,2}(x)$ in blue, $\Psi_{5,1,2}(x,0)$ in orange, $\Psi_{10,1,2}(x,0)$ in green and ${\Psi}_{20,1,2}(x,0)$ in red.}
\end{figure}

We can also plot the resulting distribution $|\Psi_{a,b}(x,t)|^2= 
\lim_{N \to \infty}|\Psi_{N,a,b}(x,t)|^2$. The resulting distributions $|\Psi_{a,b}(x,t)|^2$ is a topic for future work. 

\begin{figure}[h] 
   	\includegraphics[scale=0.467]{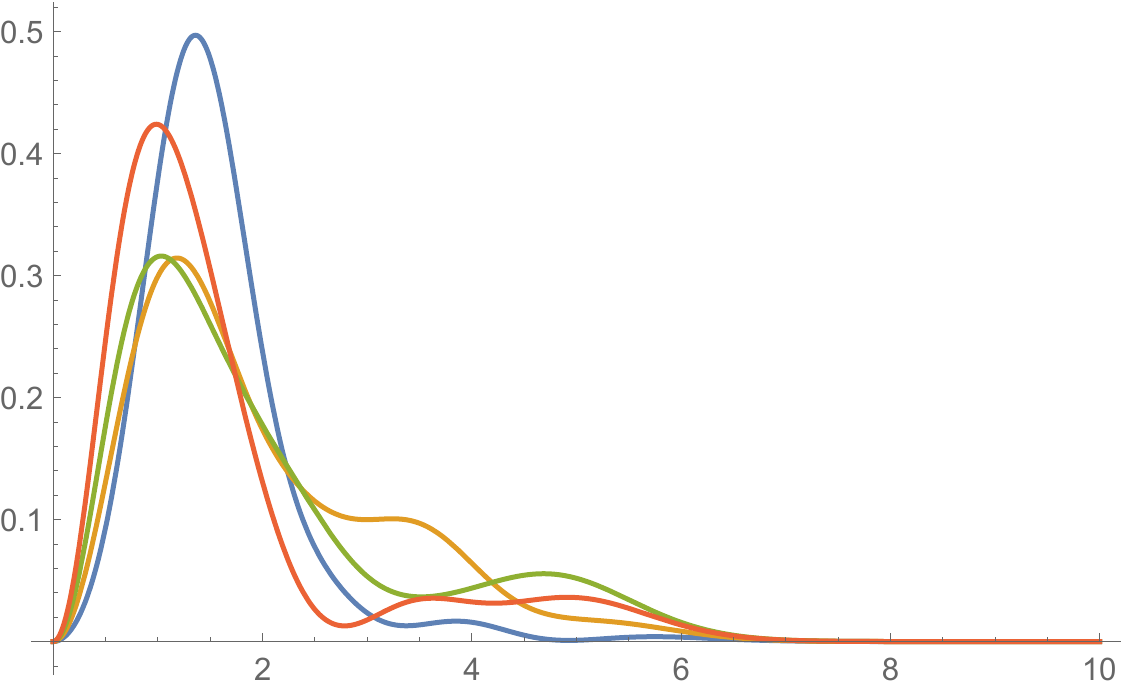} 
        \includegraphics[scale=0.467]{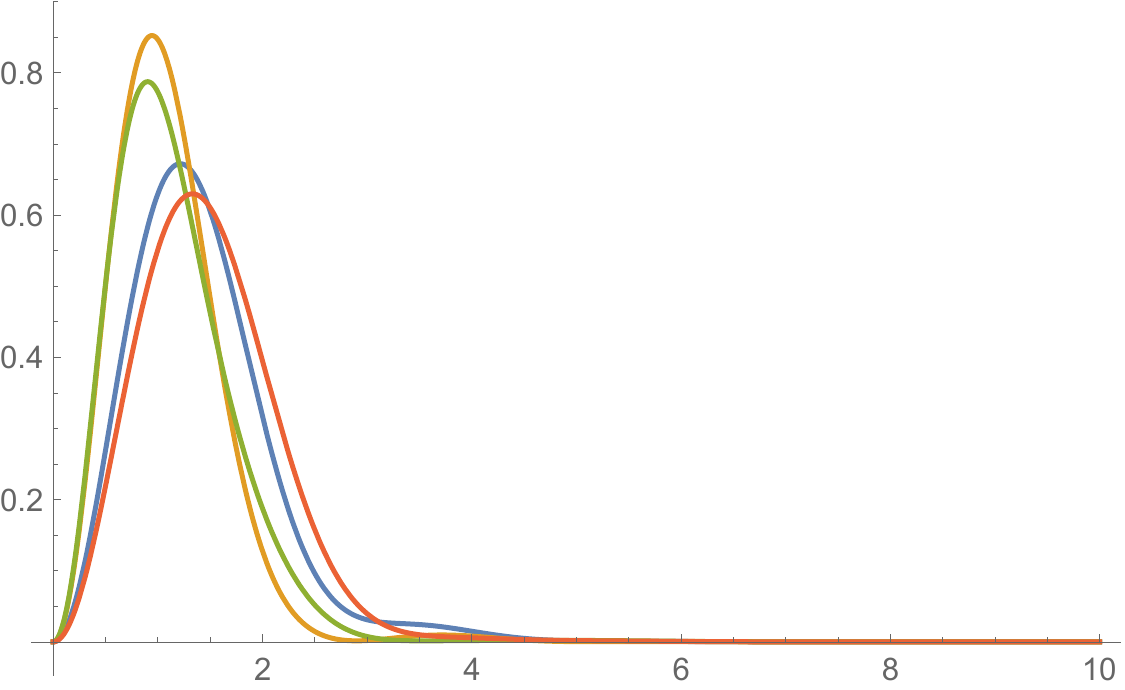} 
\caption{\underline{Left}: plot of $|\Psi_{1,1}(x,1)|^2$ in blue, $|\Psi_{1,1}(x,2)|^2$ in orange, $|\Psi_{1,1}(x,3)|^2$ in green and $|\Psi_{1,1}(x,4)|^2$ in red. \underline{Right}: plot of $|\Psi_{1,2}(x,1)|^2$ in blue, $|\Psi_{1,2}(x,2)|^2$ in orange and $|\Psi_{1,2}(x,3)|^2$ in green and $|\Psi_{1,2}(x,4)|^2$ in red. In both instances the infinite series \eqref{eq:waveAiry} was truncated with five terms.}
\end{figure}


\subsection{The unrestricted linear potential} \label{sec:list06}
Now we have $V(x) = kx$ with $x \in \R$. The Feynman propagator \cite[$\mathsection$6.2.1.4]{handbookPI} is
    \begin{align} \label{eq:LinearPropagator}
    K(x'',x';t'',t') &= \int_{x(t') = x'}^{x(t'') = x''} \mathcal{D}x(t) \exp\bigg[\frac{im}{2\hbar} \int_{t'}^{t''}  ({\dot x}^2 - k x) dt\bigg]  \nonumber \\
    &= \int_{\R} dE e^{-iET/\hbar} \bigg(\frac{2m}{\hbar^2 \sqrt{k}}\bigg)^{2/3}  \operatorname{Ai}\bigg[\bigg(x''-\frac{E}{k}\bigg)\bigg(\frac{2mk}{\hbar^2}\bigg)^{1/3}\bigg] \operatorname{Ai}\bigg[\bigg(x'-\frac{E}{k}\bigg)\bigg(\frac{2mk}{\hbar^2}\bigg)^{1/3}\bigg] \nonumber \\
    &= \bigg(\frac{m}{2 \pi i \hbar T}\bigg)^{1/2} \exp \bigg[ \frac{i}{\hbar} \bigg(\frac{m}{2}\frac{(x''-x')^2}{T} - \frac{kT}{2}(x''+x') - \frac{k^2T^3}{24m} \bigg) \bigg].
\end{align}
Our choice for the initial state is going to be 
    \begin{align}
        \Psi_{a,b}(x,0) = \begin{cases}
        \big(\frac{(2a)^{2b-1}}{\Gamma(2b-1)}\big)^{1/2} \exp(-ax) x^{b-1}, \quad &\mbox{$x \ge 0$}, \\
        0, \quad &\mbox{$x < 0$}.
        \end{cases}
    \end{align}
for $\real(a) >0$ and $\real(b) > \frac{1}{2}$. Note how this is not the same as \eqref{eq:initialstatelinear01}. When $a=1$ the propagator integral can be evaluated by the use of mathematical software and it is
\begin{align}
    \Psi_{1,b}(x,t) &= \int_{\R} K(x,y;t,0)\Psi_{1,b}(y,0) dy \nonumber \\
    &= \frac{2^{\frac{3 b}{2}-2} \left(-\frac{i}{t}\right)^{-b/2} e^{-\frac{i \left(t^4+12 t^2
   x-12 x^2\right)}{24 t}}}{\sqrt{\pi } \sqrt{\Gamma (2 b-1)}} \bigg[ \sqrt{-\frac{i}{t}} \Gamma \left(\frac{b}{2}\right) \,
   _1F_1\bigg(\frac{b}{2};\frac{1}{2};\frac{i \left(2 t+i \left(t^2+2
   x\right)\right)^2}{8 t}\bigg) \nonumber \\
   & \quad -\frac{\Gamma \left(\frac{b+1}{2}\right) \left(2 t+i \left(t^2+2 x\right)\right)}{\sqrt{2} t} \,
   _1F_1\bigg(\frac{b+1}{2};\frac{3}{2};\frac{i \left(2 t+i \left(t^2+2
   x\right)\right)^2}{8 t}\bigg)
   \bigg],
\end{align}
provided $\real(b)>0$. Different values of $|\Psi_{1,b}(x,t)|^2$ are plotted in Figure \ref{fig:list06}.

\begin{figure}[h] 
   	\includegraphics[scale=0.467]{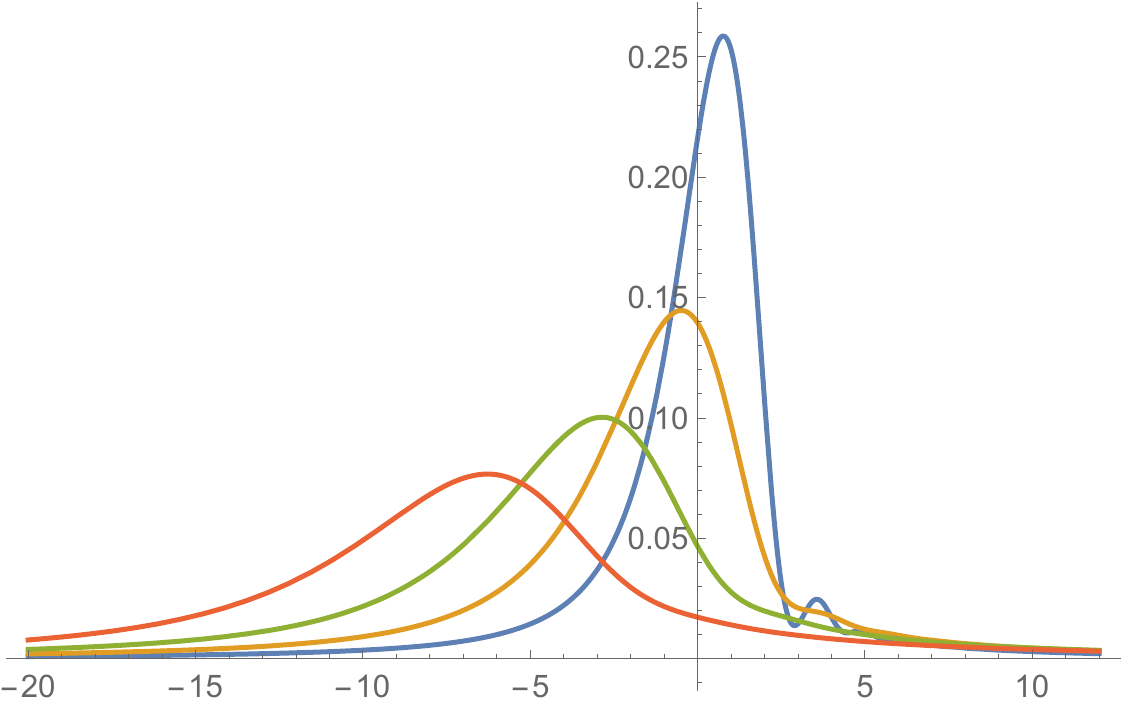} 
        \includegraphics[scale=0.467]{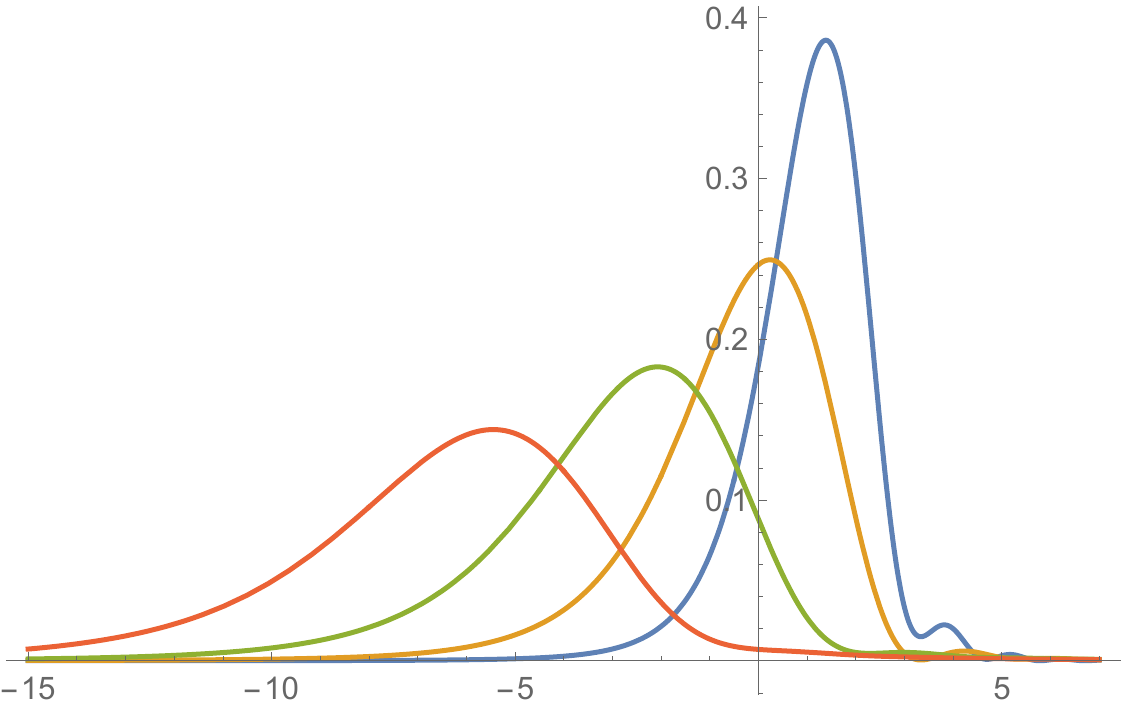} 
\caption{\underline{Left}: plot of $|\Psi_{1,1}(x,1)|^2$ in blue, $|\Psi_{1,1}(x,2)|^2$ in orange, $|\Psi_{1,1}(x,3)|^2$ in green and $|\Psi_{1,1}(x,4)|^2$ in red. \underline{Right}: plot of $|\Psi_{1,2}(x,1)|^2$ in blue, $|\Psi_{1,2}(x,2)|^2$ in orange, $|\Psi_{1,2}(x,3)|^2$ in green and $|\Psi_{1,2}(x,4)|^2$ in red.}
    \label{fig:list06}
\end{figure}

\subsection{The infinite square well potential} \label{sec:list07}
The infinite potential is a very early example of quantum mechanics textbook exercise \cite{griffiths}. The potential is
\begin{align}
    V(x) = \begin{cases}
    0, \quad &\mbox{if $0 < x < a$}, \\
    \infty, \quad &\mbox{otherwise}.
    \end{cases}
\end{align}
This system is simple enough that we can bypass the kernel approach. The steady states and their eigenvalues are
\begin{align}
    \psi_n(x) = \sqrt{\frac{2}{a}} \sin \frac{n \pi x}{a} \quad \textnormal{and} \quad E_n = \frac{n^2 \pi^2 \hbar^2}{2ma^2}.    
\end{align}
The prototype example of an initial state for the infinite square well is
\begin{align}
    \Psi_a(x,0) = \begin{cases}
    \frac{2\sqrt{3}}{a^{3/2}}x, \quad &\mbox{for $0 \le x \le a/2$}, \nonumber \\
    \frac{2\sqrt{3}}{a^{3/2}}(x-a), \quad &\mbox{for $a/2 \le x \le a$}.
    \end{cases}
\end{align}
The wave function can easily be found via \eqref{eq:PsixtPsin} by first noting that the constants are given $c_n = \int_0^a \Psi_a(x,0)dx = \frac{4\sqrt{6}}{n^2 \pi^2} \sin \frac{n \pi}{2}$. Therefore
\begin{align}
    \Psi_a(x,t) = \frac{8}{\pi^2} \sqrt{\frac{3}{a}} \bigg\{\sum_{n \equiv 1 \modu 4} - \sum_{n \equiv 3 \modu 4} \bigg\} \frac{1}{n^2}  \sin \bigg(\frac{n \pi x}{a}\bigg) \exp\bigg(-\frac{it}{\hbar}\frac{n^2 \pi^2 \hbar^2}{2ma^2}\bigg). 
\end{align}
In Figure \ref{fig:list07} we plot several instances of $|\Psi_a(x,t)|^2$ with $m=\hbar=1$.

\begin{figure}[h] 
   	\includegraphics[scale=0.467]{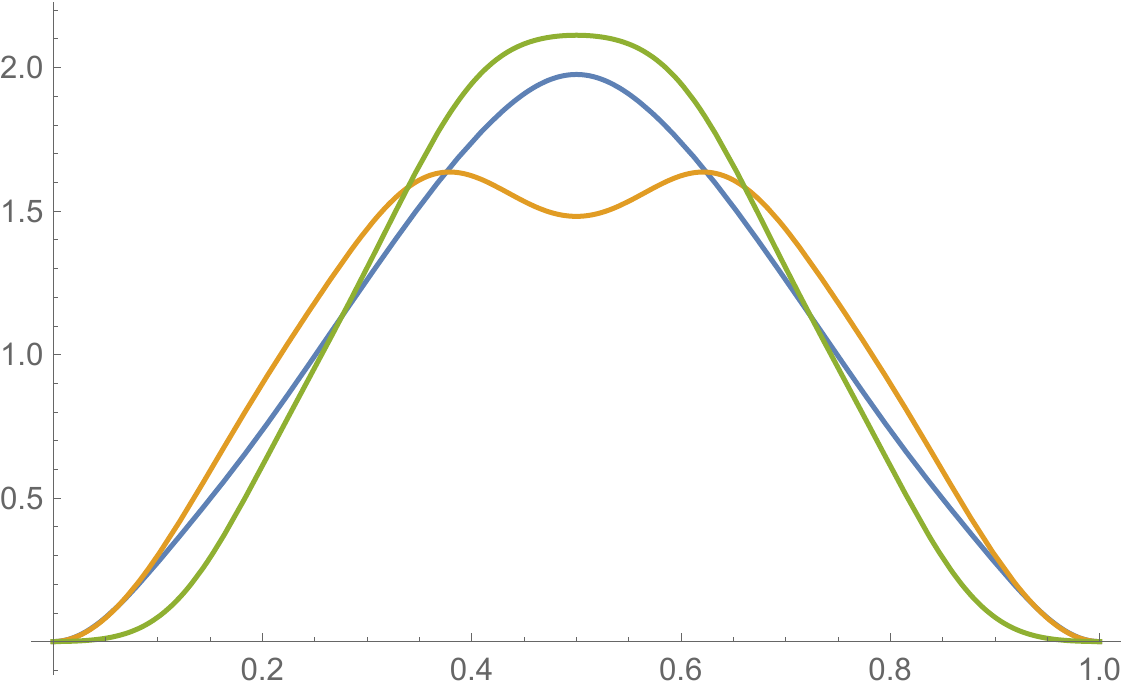} 
        \includegraphics[scale=0.467]{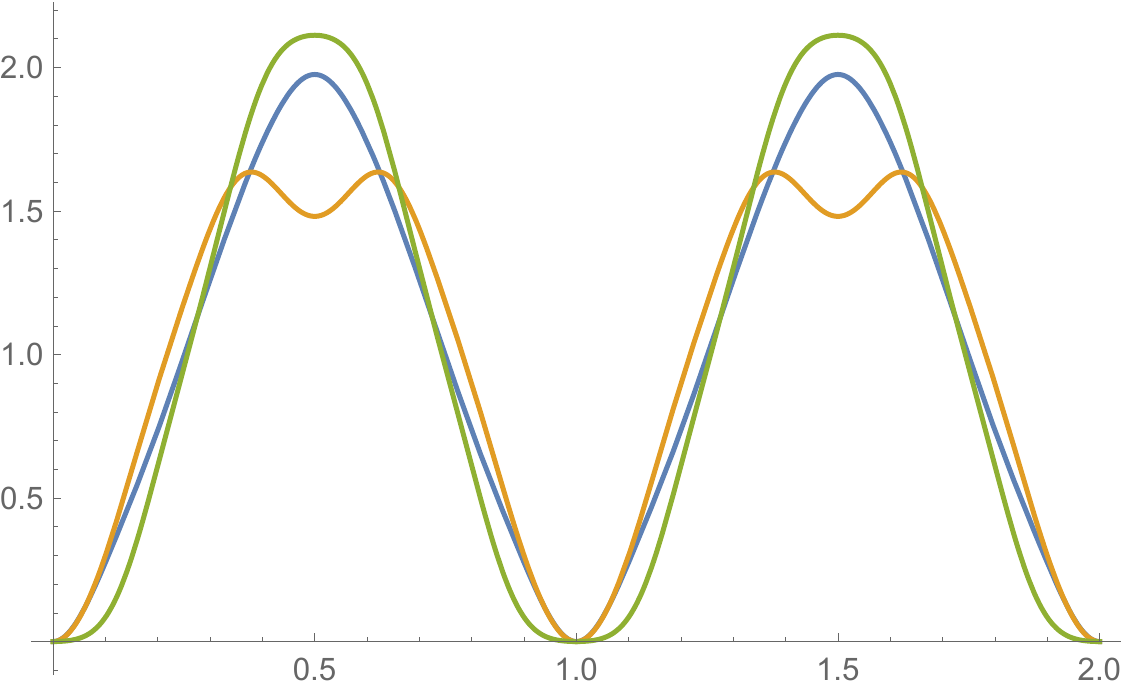} 
\caption{\underline{Left}: plot of $|\Psi_{1}(x,1)|^2$ in blue, $|\Psi_{1}(x,2)|^2$ in orange and $|\Psi_{1}(x,3)|^2$ in green with. \underline{Right}: same but with $a=2$.}
    \label{fig:list07}
\end{figure}

\subsection{The P\"{o}schl-Teller potential} \label{sec:list08}
In one dimension, the P\"{o}schl-Teller potential has the form
    \begin{align} \label{eq:potentialTeller}
        V(x)= \frac{\hbar^2}{2m}\bigg(\frac{\alpha^2-\frac{1}{4}}{\sin^2 x} + \frac{\beta^2-\frac{1}{4}}{\cos^2 x}\bigg)
    \end{align}
for $0< x < \frac{\pi}{2}$ and for $\alpha, \beta \in \R$. Its associated Feynman propagator \cite[$\mathsection$6.5.1.1]{handbookPI} is
    \begin{align} \label{eq:TellerPropagator}
    K(x'',x';t'',t') &= \int_{x(t') = x'}^{x(t'') = x''} \mathcal{D}x(t) \exp\bigg[\frac{i}{\hbar} \int_{t'}^{t''} \bigg(\frac{m}{2}{\dot x}^2  - \frac{\hbar^2}{2m}\bigg(\frac{\alpha^2-\frac{1}{4}}{\sin^2 x} + \frac{\beta^2-\frac{1}{4}}{\cos^2 x}\bigg) \bigg)dt\bigg] \nonumber \\
    &= \sum_{n=0}^\infty \Psi(n,\alpha,\beta,x'')\Psi^*(n,\alpha,\beta,x') e^{-iTE_n/\hbar}.
\end{align}
This is a discrete spectrum with eigenvalues
\begin{align}
        E_n = \frac{\hbar^2}{2m}(\alpha+\beta+2n+1)^2,
\end{align}
and eigenstates given by
\begin{align}
    \Psi(n,\alpha,\beta,x) &= \bigg(2(\alpha+\beta+2n+1)\frac{n!\Gamma(\alpha+\beta+n+1)}{\Gamma(\alpha+n+1)\Gamma(\beta+n+1)}\bigg)^{1/2}  (\sin x)^{\alpha+1/2}(\cos x)^{\beta+1/2} P_{n}^{(\alpha,\beta)}(\cos 2x).
\end{align}
Here $P_{n}^{(\alpha,\beta)}(x)$ are the Jacobi polynomials \cite{gradryz, watson}. Let us set the initial state to be
\begin{align} \label{eq:initialTeller}
    \Psi(x,0) = 2 \sqrt{\frac{\pi}{\pi^2-\sin(\pi^2)}} \sin(\pi x).
\end{align}
The constants associated to this initial state are given by
\begin{align}
    c_n(\alpha,\beta) = \int_0^{\pi/2} \Psi^*(n,\alpha,\beta,x) \Psi(x,0) dx.
\end{align}
The presence of the Jacobi polynomials makes this a difficult integral to compute and we shall therefore proceed numerically. To simplify the fractions in \eqref{eq:potentialTeller} we set $\alpha = \beta = \frac{\sqrt{5}}{2}$. In Figure \ref{fig:list08a} we approximate $\Psi(x,0)$ with $\sum_{n=0}^{M} c_n(\alpha,\beta) \Psi(n,\alpha,\beta,x)$ as $M \to \infty$.

\begin{figure}[h] 
   	\includegraphics[scale=0.467]{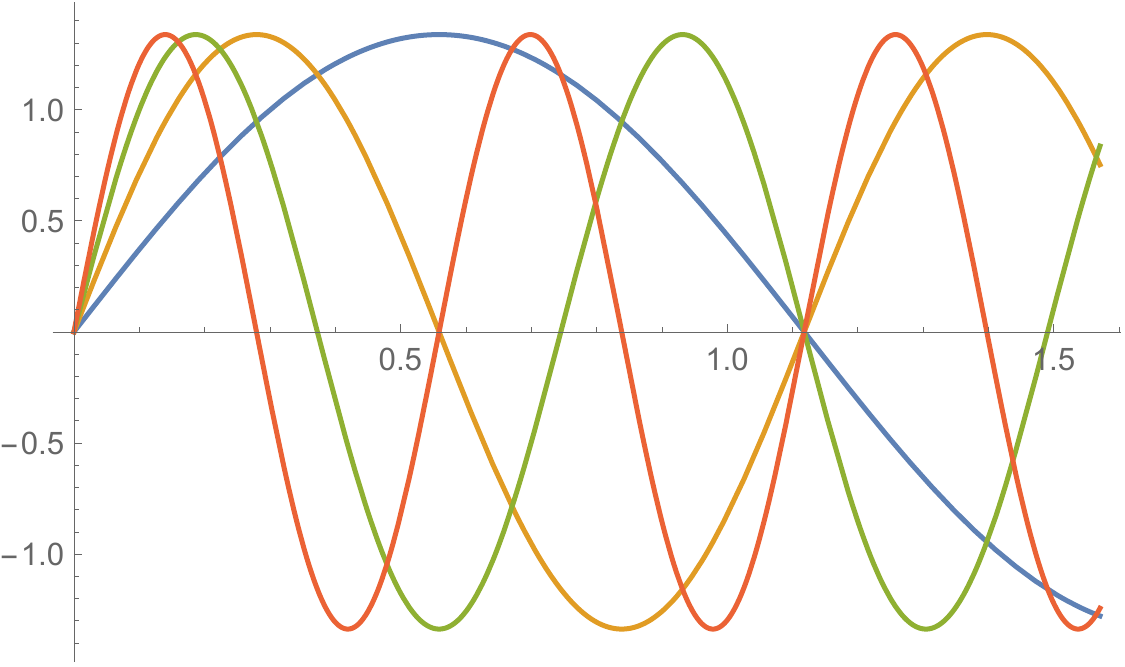} 
        \includegraphics[scale=0.467]{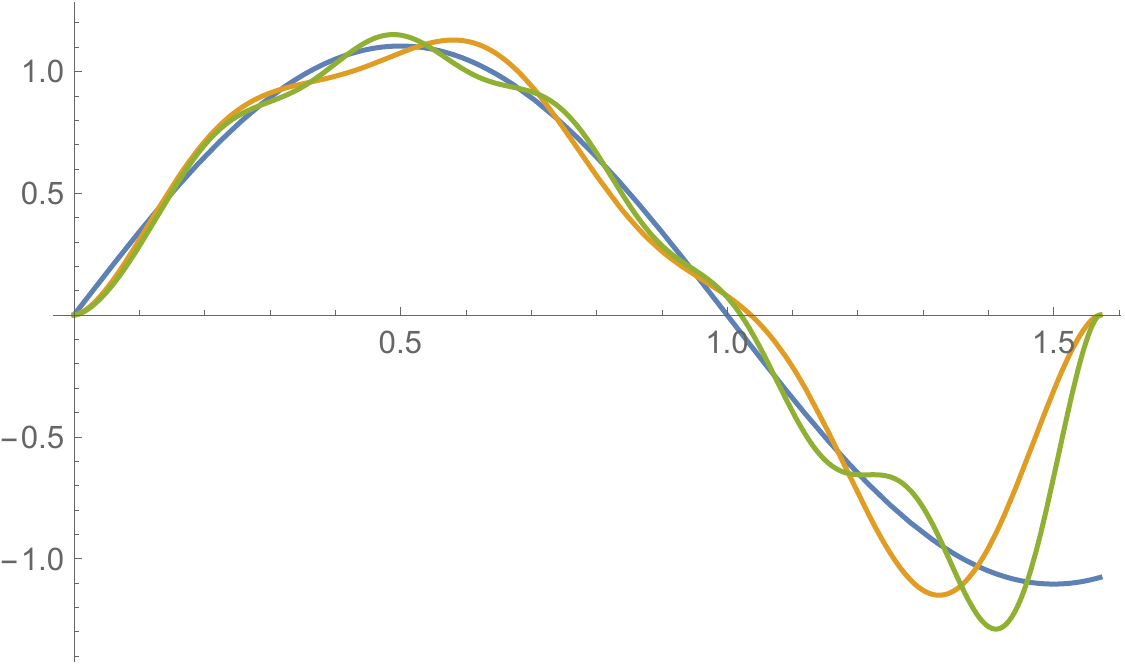} 
\caption{\underline{Left}: plot of $\Psi_n(x)$ for $n=1$ in blue, $n=2$ in orange, $n=3$ in green and $n=4$ in red. \underline{Right}: Plot of $\Psi(x,0)$ in blue and $\sum_{n=0}^5 c_n(\alpha,\beta) \Psi(n,\alpha,\beta,x)$ in orange and $\sum_{n=0}^{10} c_n(\alpha,\beta) \Psi(n,\alpha,\beta,x)$ in green.}
    \label{fig:list08a}
\end{figure}

We define the truncation
\begin{align}
    \Psi(x,t,M)= \sum_{n=1}^M c_n(\alpha,\beta) \Psi(n,\alpha,\beta,x) e^{-E_n t/\hbar}
\end{align}
and plot $|\Psi(x,t,M)|^2$ in Figure \ref{fig:list08b}. We have $\Psi(x,t,M) \to \Psi(x,t)$ as $M \to \infty$.

\begin{figure}[h] 
   	\includegraphics[scale=0.467]{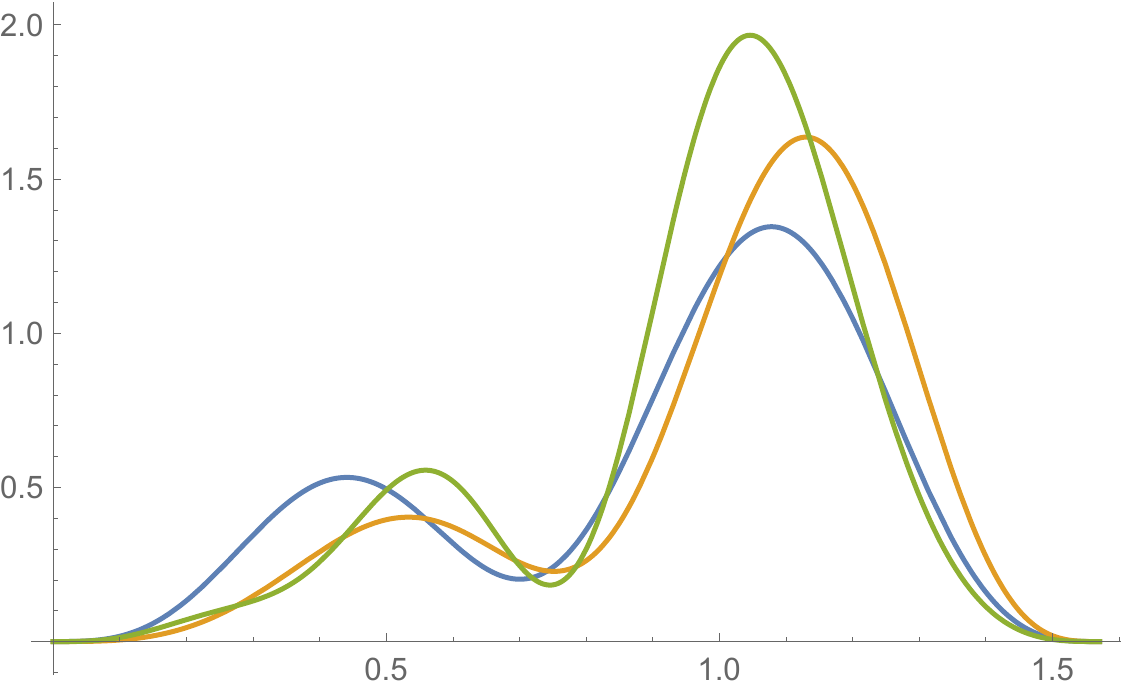} 
        \includegraphics[scale=0.467]{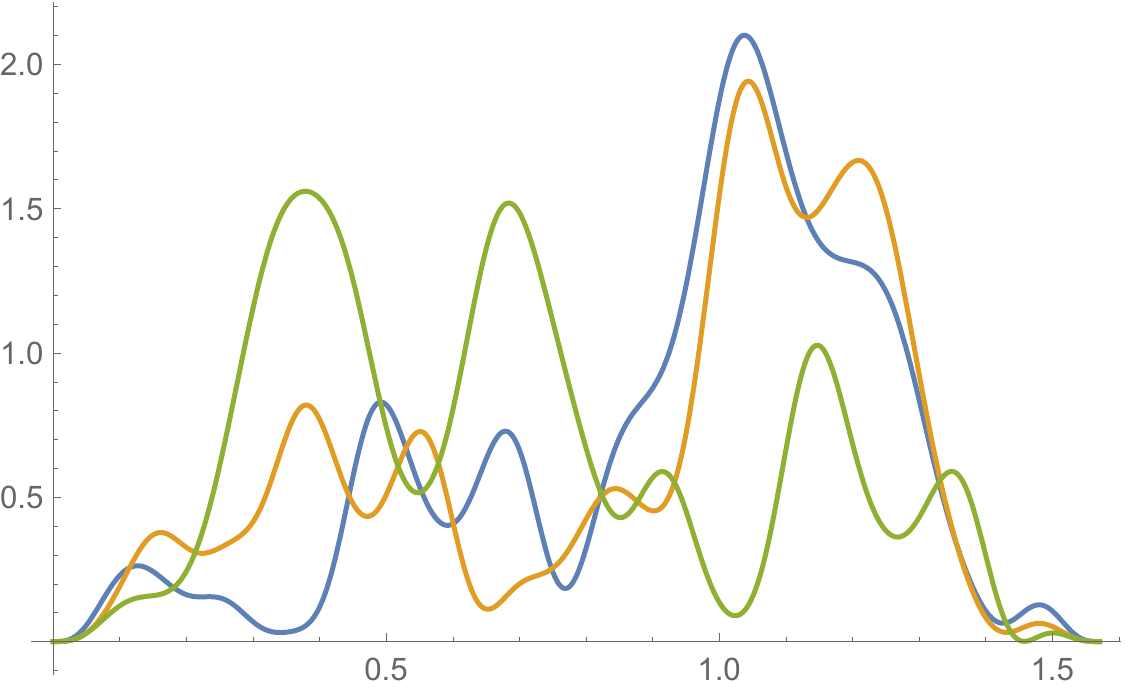} 
\caption{\underline{Left}: plot of $|\Psi(x,1,1)|^2$ in blue,$|\Psi(x,1,2)|^2$ in orange and $|\Psi(x,1,3)|^2$ in green. \underline{Right}: Plot of $|\Psi(x,1,15)|^2$ in blue and $|\Psi(x,2,15)|^2$ in orange and $|\Psi(x,3,15)|^2$ in green.}
    \label{fig:list08b}
\end{figure}

\subsection{The Coulomb potential} \label{sec:list09}
In the case of only one dimension the Coulomb, or hydrogen atom, potential is given by
    \begin{align}
        V_{\operatorname{1DCoulomb}}(x)= -\frac{e_1e_2}{|x|}
    \end{align}
for $x>0$ and for $e_1, e_2 \in \R^+$. This potential is interesting as it is an instance of a mixed spectrum: it has both discrete and continuous spectrum. The Feynman propagator is \cite[$\mathsection$6.8.1]{handbookPI}
    \begin{align} \label{eq:CoulombPropagator}
    K(x'',x';t'',t') &= \int_{x(t') = x'}^{x(t'') = x''} \mathcal{D}x(t) \exp\bigg[\frac{i}{\hbar} \int_{t_0}^t \bigg(\frac{m}{2}{\dot x}^2  +\frac{e_1e_2}{|x|} \bigg)dt\bigg] \nonumber \\
    &= \sum_{n=0}^\infty \Psi_n(x'')\Psi_n^*(x') e^{-iTE_n/\hbar} + \int_{\R} \tilde \Psi_k(x'') \tilde \Psi_k^*(x')e^{-iT \tilde E_k/\hbar}dk .
\end{align}
The bound (discrete) eigenstates and eigenvalues are given by
\begin{align} \label{discretecoulomb}
    \Psi_n(x) &= \bigg(\frac{n!}{a(n+1)!}\bigg)^{1/2} \frac{2x}{a(n+1)^2} \exp\bigg[-\frac{x}{a(n+1)}\bigg] L_n^{(1)}\bigg(\frac{2x}{a(n+1)}\bigg) \quad \textnormal{and} \quad E_n = -\frac{m(e_1e_2)^2}{2\hbar^2(n+1)^2},
\end{align}
whereas the scattering (continuous) states and elements of the spectrum are
    \begin{align} \label{continuouscoulomb}
        \tilde \Psi_k(x) = \frac{\Gamma(1-\frac{i}{ak})}{\sqrt{2\pi}} \exp\bigg(\frac{\pi}{2ak}\bigg)M_{i/(ak),1/2}(-2ikx) \quad \textnormal{and} \quad \tilde E_k = -\frac{k^2 \hbar^2}{2m}.
    \end{align}
Our initial state will be set to
\begin{align}
    \Psi(y,0) = \frac{2}{\sqrt{3}}y^2 e^{-y}.
\end{align}
At $t>0$ the wave function will of course be $\Psi(x,t) = \int_{\R^+} K(x,y;t,0)\Psi(y,0)dy$, however it is very difficult to compute the resulting integral analytically. This is because $n$ appears as an \textit{argument} in the Laguerre polynomial as $\frac{2x}{a(n+1)}$ and $k$ also appears as an \textit{argument} inside the Whitaker function as $-2ikx$, see \cite{gradryz}. We plot $|\Psi(x,t,M)|^2$ where $\Psi(x,t,M)=\int_{\R^+} K(x,y;t,0,M)\Psi(y,0)dy$ with $K(x,y;t,0,M)$ indicating truncation of $n$-sum at $M$ and truncation of $k$-integral as $\int_{-M}^M$ in Figure \ref{fig:list09}.

\begin{figure}[h] 
   	\includegraphics[scale=0.467]{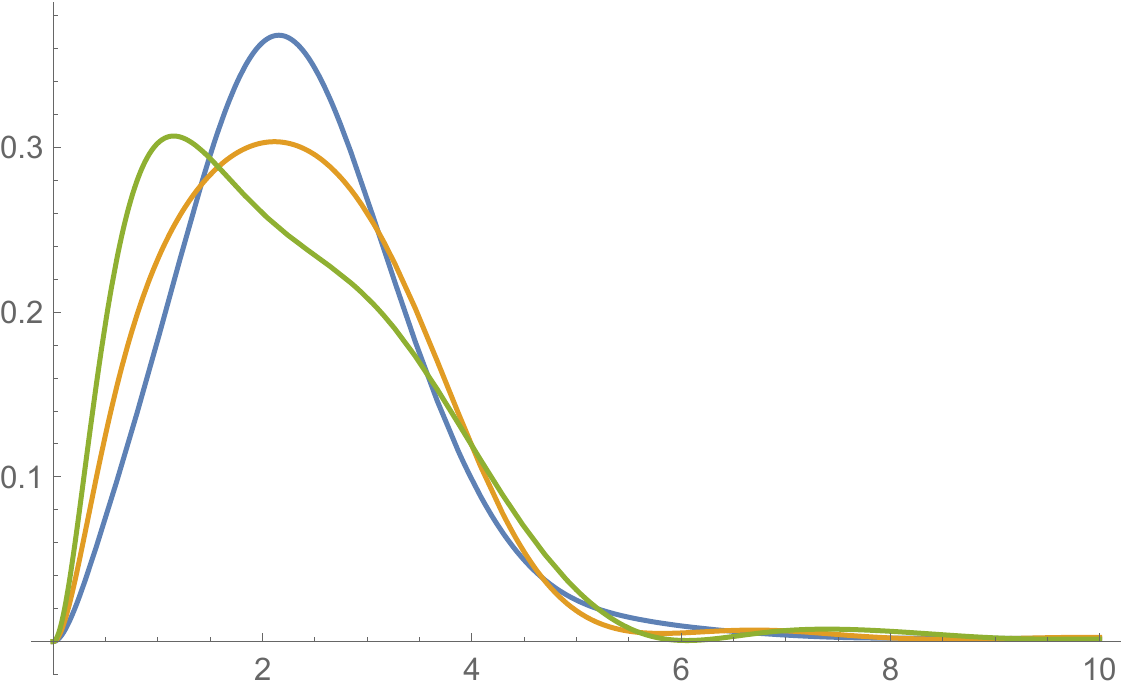} 
        \includegraphics[scale=0.467]{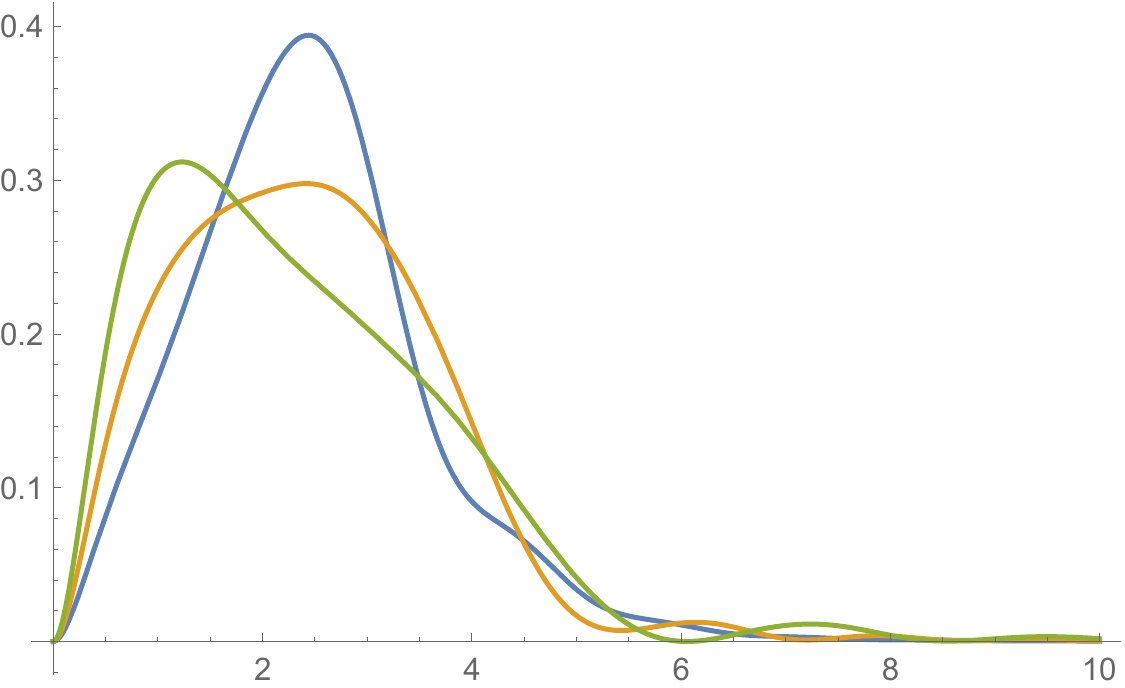} 
\caption{\underline{Left}: plot of $|\Psi(x,1)|^2$ in blue, $|\Psi(x,2)|^2$ in orange, and $|\Psi(x,3)|^2$ in green, with truncation of $\sum_{n=0}^2$ in \eqref{discretecoulomb} and $\int_{-2}^2$ in \eqref{continuouscoulomb}. \underline{Right}: plot of $|\Psi(x,1)|^2$ in blue, $|\Psi(x,2)|^2$ in orange, and $|\Psi(x,3)|^2$ in green, with truncation of $\sum_{n=0}^5$ in \eqref{discretecoulomb} and $\int_{-5}^5$ in \eqref{continuouscoulomb}.}
    \label{fig:list09}
    \end{figure}
\subsection{The Dirac $\delta$ potential and the Laplace distribution} \label{sec:list10}
\noindent Suppose we consider the potential $V(x) = - a \delta(x)$, where $\delta$ is the Dirac delta function and $a>0$. Then from \cite[Problem 2.47]{griffiths}  we have that
\begin{align}
\psi(a,x,t) = \frac{\sqrt{ma}}{\hbar} \exp(-ma |x|/\hbar^2)\exp(-iEt/\hbar)
\end{align}
where $E = -ma^2/(2\hbar^2)$. This bypassed having to deal with the Feynman kernel. At $t=0$ we see that $\psi(a,x,0)= \sqrt{a}e^{-a|x|}$ and hence $\int_{\R} |\psi(a,x,0)|^2 dx = 1$, as it should be. Moreover, we also have that $|\psi(a,x,t)|^2 = |a|e^{-2a|x|} = |\psi(a,x)|^2$, since there is no $t$ dependence but we still have $a$ as a `wiggle' parameter. Now, the PDF of the Laplace distribution is given by
\begin{align} \label{eqf:LaplaceDistribution}
    f(\mu,b;x) = \frac{1}{2b} \exp\bigg(-\frac{|x-\mu|}{b}\bigg).
\end{align}
We write $X \sim \mathfrak{L}(\mu,b)$. Therefore we may identify
\begin{align}
|\psi(2,x)|^2 = f(0,\tfrac{1}{4};x) .
\end{align}
In general one will get that
\begin{align}
f(0,b;x) = \frac{1}{2b} \exp\bigg(-\frac{|x|}{b}\bigg) = \bigg|\psi\bigg(\frac{1}{2b},x\bigg)\bigg|^2 \sim \mathfrak{L}(0,b).
\end{align}
In Figure \ref{fig:list10} we plot several instances of $|\psi(\frac{1}{2b},x)|^2$.

\begin{figure}[H] 
   	\includegraphics[scale=0.467]{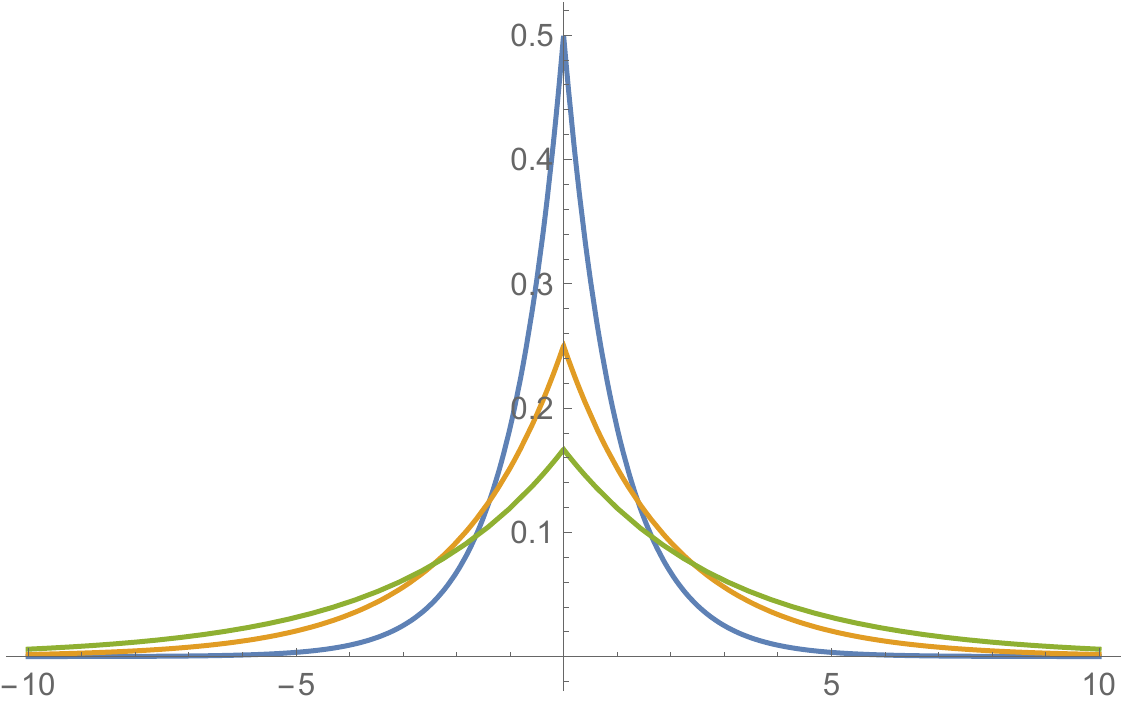} 
        \includegraphics[scale=0.467]{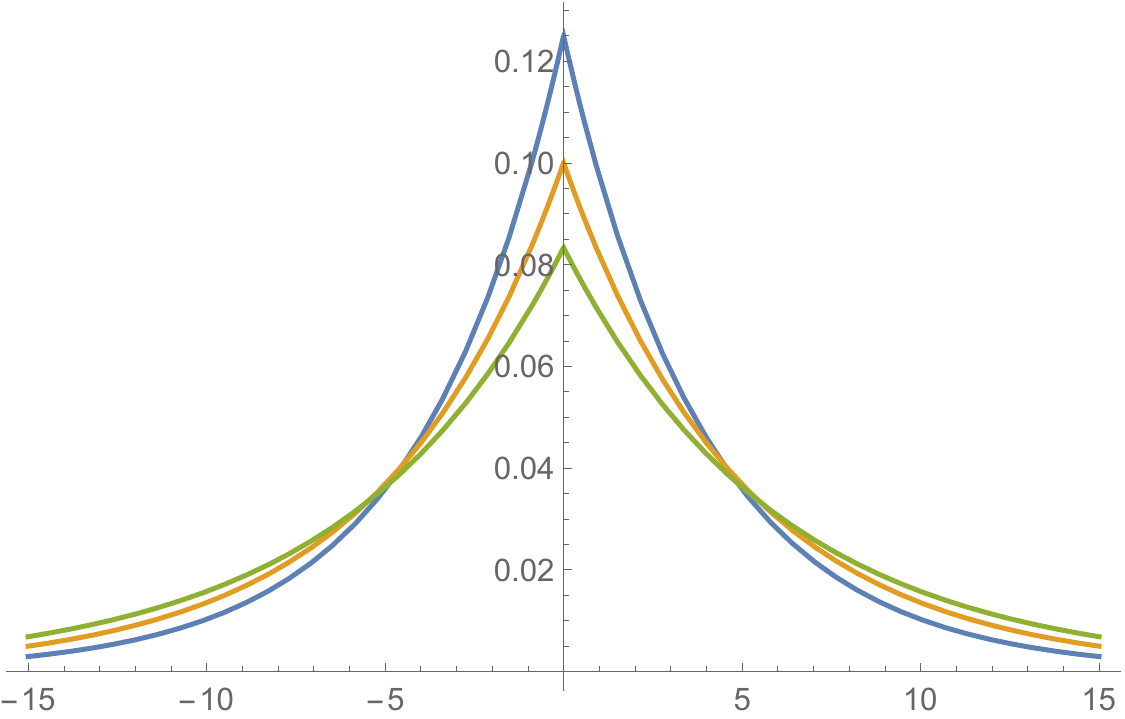} 
\caption{\underline{Left}: plot of $|\psi(\frac{1}{2b},x)|^2$ with $b=1$ in blue, $b=2$ in orange and $b=3$ in gree. \underline{Right}: plot of $|\psi(\frac{1}{2b},x)|^2$ with $b=4$ in blue, $b=5$ in orange and $b=6$ in green.}
    \label{fig:list10}
\end{figure}

\subsection{Summary table} \label{sec:summary}
We now summarize our finding from the previous sections in the following table.

\begin{center}
\begin{longtable}{|c|c|c|c|c|c|c|}
\hline
\textbf{Potential} & \textbf{Range} & \textbf{Initial state} & \textbf{Special} & \textbf{Resulting wave} & \textbf{Resulting distribution} \\
\hline
$V(x)=0$ & $x \in \R$ & General Gaussian & None & $_1F_1(\alpha,\frac{1}{2},f(x,t))+ ~_1F_1(\alpha,\frac{3}{2},f(x,t))$ & Weighted sum of $\chi$'s \\
\hline
$V(x)=0$ & $x >0$ & Half Gaussian & None & $~_1F_1(\alpha,\frac{3}{2},f(x,t))$ & Maxwell-Boltzman \\
\hline
$V(x)=0$ & $-b < x < b$ & Short Gaussian & None & $\sum_n \sin(a(n)) \exp(ib(n) \operatorname{erf}(c(n))$ & Weighted $\sum$ of wrapped $\mathcal{N}$ \\
\hline
$V(x)=\frac{1}{2}m\omega^2x^2$ & $x \in \R$ & Gaussian & Hermite & Gaussian amplitude & Normal $\mathcal{N}$\\
\hline
$V(x)=\frac{\lambda^2-\frac{1}{4}}{2mx^2}$ & $x >0$ & Half Gaussian & Laguerre & $~_1F_1(A,\lambda+1,f(x,t))$ & Unknown \\
\hline
$V(x)=\frac{\lambda^2-\frac{1}{4}}{2mx^2}$ & $x >0$ & Gamma & Laguerre & $_1F_1(\alpha,\frac{1}{2},f(x,t))+ ~_1F_1(\alpha,\frac{3}{2},f(x,t))$ & Unknown \\
\hline
$V(x) = \{kx,\infty\}$ & $x >0$ & Gamma & Airy & Unknown & Unknown \\
\hline
$V(x)=kx$ & $x\in \R$ & Half Gamma & None & $_1F_1(\alpha,\frac{1}{2},f(x,t))+ ~_1F_1(\alpha,\frac{3}{2},f(x,t))$ & Unknown \\
\hline
$V(x)=\{0, \infty\}$ & $0< x < a$ & Triangular & None & Twisted Dirichlet character sums & Unknown \\
\hline
$V(x)=\frac{1}{\sin^2 x}$ & $x \in [0,\frac{\pi}{2}]$ & Periodic & Jacobi & Unknown & Unknown \\
\hline
$V(x)=\frac{e}{x}$ & $x \in \R^+$ & Gamma & Laguerre & Not always tractable & Unknown \\
\hline
$V(x)=-a\delta(x)$ & $x \in \R$ & Laplace & None & $\exp(-|x|)$ & Laplace $\mathfrak{L}$\\
\hline
\end{longtable}
\label{table:fulltable}
\end{center}

We stress again that although we cannot always identify the exact nature of the distributions, we can for the most part write their closed analytic forms, if needed in terms of infinite series or special functions.

\section{Probability distributions from given potentials} \label{sec:inverselist}

In Sections \ref{sec:list01} to \ref{sec:list10} we illustrated how the usual potentials (free particle, harmonic oscillator, centrigufal, Coulomb, etc...) give rise to certain random variables and, when possible, we have elucidated the PDF of said random variables. The process has been such that we have depended on what the potential and initial state imply for the resulting PDF. We now propose a look at reverse engineering this process and attempting to identify the potential that induces a desired PDF. The idea is instead of having our starting point be a potential, it will now be ground state, i.e. the $n=0$ case of the eigenstates $\Psi_n(x)$, see \cite{susyQM}. Without loss of generality we choose the ground state enegery $E_0$ of a Hamiltonian $\hat H$ to be 0. The Schr\"{o}dinger equation for the ground state wave function $\Psi_0(x)$ will be
\begin{align}
    -\frac{\hbar^2}{2m} \frac{d^2}{dx^2}\Psi_0(x) + V(x) \Psi_0(x) = 0 \times \Psi_0(x) = 0,
\end{align}
which, upon rearrangement, yields
\begin{align} \label{eq:susyFormula}
    V(x) = \frac{\hbar^2}{2m} \frac{\Psi_0''(x)}{\Psi_0(x)}.
\end{align}
This is a simplified version of the TISE \eqref{eq:TISE} as the right-hand is now is zero. 
Let us choose $\Psi_0$ to be
\begin{align} \label{eq:targetNormal}
    \Psi_0(x) = \frac{2^{1/4}a^{1/4}}{\pi^{1/4}} \exp(-ax^2)
\end{align}
with $a>0$ as well as support $x \in \R$. In this case, a straightforward calculation shows that
\begin{align}
    V(x) = \frac{2a^2 \hbar^2}{m}x^2 - \frac{a\hbar^2}{m}.
\end{align}
Ignoring the last constant term, this clearly shows a harmonic oscillator potential. Making the identification $a = \frac{m\omega}{2\hbar}$ leads naturally to $V(x) = \frac{1}{2}m\omega^2 x^2$. This goes on to show that if our target is the normal distribution \eqref{eq:targetNormal}, then a potentially desirable \textit{candidate} for a potential is the harmonic oscillator. The caveat here is that we still have to find a way, typically through lowering and raising operators, to evolve from ground states $\Psi_0(x)$ to excited states $\Psi_n(x)$ with $n \ge 1$ and these excited states must still maintain the desired \textit{structure} of the PDF, in this case the normal PDF. The excited states for the harmonic oscillator in \eqref{eq:HermitePsin} show that indeed the structure is maintained as we still have $\exp(- ax^2)$ twisted by the Hermite polynomials $H_n(\sqrt{\frac{m\omega}{\hbar}}x)$. While the presence of these Hermite polynomials might seem undesirable, they are in fact essential to eventually obtain a normal amplitude for $\psi(x,t)$ through the effect of the Mehler formula.

Another example involves the ground state
\begin{align} \label{eq:targetChi}
    \Psi_0(x) = N \exp(-ax^2)x^{b-1} 
\end{align}
with a suitable normalization constant $N$ and support $x \in \R^+$. Substituting \eqref{eq:targetChi} in \eqref{eq:susyFormula} leads to
\begin{align}
        V(x) = \frac{2a^2\hbar^2}{m}x^2 + \frac{(b-1)(b-2)\hbar^2}{2mx^2} + \frac{a\hbar^2}{m} - \frac{2ab\hbar^2}{m}.
\end{align}
The same identification $a = \frac{m\omega}{2\hbar}$ followed by $b = \frac{3 \pm 2\lambda}{2}$ and dropping the constants terms that neither depend on $x^2$ nor on $x^{-2}$ yields exactly \eqref{eq:VRHO}. Therefore, if our target is the $\chi$ distribution, then we can conclude that the radial harmonic oscillator is a strong \textit{candidate} to produce its PDF. This is indeed the case when $b=1,2$ and letting $\omega \to 0^+$ which produces the situation of the free particle restricted to positive values of $x$. The associated PDF for a suitable time $t$ was seen to be the Maxwell-Boltzman distribution which corresponds to the $\chi$ distribution with three degrees of freedom.

Lastly, we now consider the log normal distribution whose PDF is
\begin{align}
    g(x,\mu,\sigma) = \frac{1}{x\sigma\sqrt{2\pi}}\exp\bigg(-\frac{(\log x-\mu)^2}{2\sigma^2}\bigg)
\end{align}
for $x>0$ and mean $\mu \in \R$ as well as deviation $\sigma > 0$. A reasonable ansatz for the ground state would then be
\begin{align}  \label{eq:targetLognormal}
    \Psi_0(x) = \frac{N}{x}\exp(-a (\log x -b)^2),
\end{align}
for a suitable normalization constant $N$ and valid $a$ and $b$. Repeating the same process and substituting \eqref{eq:targetLognormal} in \eqref{eq:susyFormula} yields a potential of the form
\begin{align} \label{eq:lognormalpotential}
    V(x) = \frac{h^2 \left(2 a^2 b^2-3 a b-a+1\right)}{m x^2} -\frac{a h^2 (4 a b-3) \log x}{m x^2} + \frac{2 a^2 h^2 \log ^2 x}{m x^2}.
\end{align}
The first term corresponds to the radial harmonic oscillator, however the last two terms have structure $\frac{\log x}{x^2}$ and $\frac{\log^2 x}{x^2}$ and therefore their propagators would have to be computed separately and this is a considerable analytical effort. In Appendix \ref{sec:appendix1} we shall give some examples on how to compute path integrals and potential techniques to compute or approximate path integrals associated with potentials such as \eqref{eq:lognormalpotential}.

\subsection{Projection to ground states} 
The propagators for the Hamiltonians whose ground states are probability distributions may not have closed form, however they can still be used to generate the quantum states 
corresponding to the ground states through a variational method. The variational method starts with a quantum state $\ket{\phi(x)}$ that approximates the ground state $\ket{\phi_{0}(x)}$, that is, the inner product $\bket{ \phi(x)}{ \phi_{0}(x)}^{2} $ is sufficiently large and uses the Hamiltonian simulation to project $\ket{\phi(x)}$ onto $\ket{\phi_{0}(x)}$. 

In section \ref{sec:amplitudeencoding} we provide methods to create the initial states for the variational method that have high overlap with the target ground state and discuss Hamiltonian simulation algorithms in section \ref{sec:evolution}. In this section, for completeness, we describe the standard method for projecting onto the ground state given a suitable initial state and a Hamiltonian simulation algorithm. 

The initial state $\ket{\phi_{0}(x)}$ can be decomposed in the eigenbasis of the Hamiltonian $\hat H$ as follows, 
\al{ 
\ket{\phi(x)} = c \ket{\phi_{0}(x)} + \sum_{i\geq 1} \alpha_{i} \ket{\phi_{i}(x)} . 
} 
The $\ket{\phi_{i}(x)}$ are the excited states for the Hamiltonian, such a decomposition is possible as the eigenfunctions are a complete basis for the Hilbert space in the discrete setting.  In order to project to the ground state it suffices to resolve the eigenvalues and to measure the energy, that is to append an auxiliary register with an estimate for the corresponding eigenvalue 
\[
c \ket{\phi_{0}(x), \overline{\lambda_{0}} } + \sum_{i\geq 1} \alpha_{i} \ket{\phi_{i}(x), \overline{\lambda_{i}}}, 
\]
where the estimate is such that 
$|\lambda_{i} - \overline{\lambda_{i}}|  \leq \epsilon$ with high probability. A measurement of the auxiliary register then projects onto the ground state 
$\ket{\phi_{0}(x)} $ with probability at least $c^{2}$. 

It is well known the phase estimation algorithm \cite{Kitaev95} provides an efficient method for the energy measurement, in order to resolve the eigenvalues to an additive error of $\epsilon$ the phase estimation algorithm requires Hamiltonian simulation up to time $t= O(1/\epsilon)$. In the general case, Hamiltonian simulation in the phase estimation algorithm is performed by iteratively composing the circuits $t$ times, however for the Hamiltonians of the type $\hat H= \Delta + V(x)$ with closed form propagators that we consider in this paper, fast forwarding may be possible yielding even more efficient ground state projection algorithms. 

The resolution $\epsilon$ required for the ground state projection is on the order of the spectral gap of $\hat H$. Again, for the analytically solvable Hamiltonians, the spectral gap may be computable in closed form, for example in the case of the harmonic oscillators the eigenvalues $\lambda_{k} = 2k+1$ are known to be the odd integers. 
The result on ground state projection is stated below. 

\begin{claim} 
Given Hamiltonian $\hat H$  with spectral gap $\Delta(\hat H)$ and cost $C_{\hat H}(t)$ for simulating $\hat H$ up to time $t$ and an initial state $\ket{\phi(x)}$ that can be prepared with cost $C_{0}$, there is a quantum algorithm that outputs $\ket{\phi_{0}(x)}$ with probability at least $\bket{\phi(x)}{ \phi_{0}(x)}^{2} $. The complexity of the algorithm is $\widetilde{O}( C_{\hat H}(1/\Delta(\hat H) + C_{0})$ where the $\widetilde{O}$ hides factors logarithmic in the number of qubits. 
\end{claim} 

The ground state projection algorithm does not output a sampler from the corresponding probability distribution but an amplitude encoding for its probability density function. 
\begin{definition} 
The $N$ qubit amplitude encoding of a probability density function for distribution $D$ with density function $p(x)$ supported on the interval $[-L, L]$ is the quantum state, 
\[
\ket{p} = \frac{1}{Z} \sum_{i=0}^{2^{N}-1} p\bigg(-L+ \frac{(1+2i)L}{2^{N}} \bigg)  \ket{i} ,
\] 
where the normalization factor is $Z= \sum_{i=0}^{2^{N}-1}  p(-L+ \frac{(1+2i)L}{2^{N}})^{2}$.
\end{definition} 
In section \ref{sec:list} we provided algorithms for sampling from probability distributions that is creating quantum states of the form 
$ \sum_{i=0}^{2^{N}-1} \sqrt{p(-L+ \frac{(1+2i)L}{2^{N}} )}  \ket{i}$, measuring these states samples from the probability distributions.  
The amplitude encoding of the probability density function on the other hand is more natural in some settings as it allows for linear combinations of density functions and the computation of expectations with respect to the probability distribution using quantum inner product estimation. 

\section{Amplitude encodings and ladder state preparation} \label{sec:amplitudeencoding}

The initial state for the variational method for generating quantum representations of probability distributions is a coarse approximation to the 
probability distribution using a set of discretized levels. In this section, we define the quantum encodings for probability distributions 
and provide efficient circuits for preparing the initial states that we termed ladder states in Section \ref{sec:intro03}. 

\subsection{Ladder states} \label{sec:ladder} 
The method for preparing amplitude encodings of probability distributions starts with a coarse discrete approximation of the 
probability density function followed by time evolution or ground state projection for a suitable Hamiltonian. We introduce two family of ladder states that 
are used as initial states for our method and provide simple quantum circuits to prepare these ladder states. 

The family of monotone ladder states is a coarse discrete approximation to the amplitude encoding a symmetric distribution with mode $0$, it can be used as initial state for 
other symmetric probability distributions including but not limited to Gaussian distributions, logistic distributions, Cauchy distributions, Voigt distributions, student's distributions, and uniform distributions. The monotone ladder states are defined as as follows. 
\begin{definition} \label{ladder}
A monotone ladder state with $K=2^{k}$ levels on $N$ qubits is defined to be a quantum state,  
\al{ 
\ket{L_{K}(f)} = \frac{1}{Z} \sumtwo_{ b \in \{0, 1\}^{k},  c \in \{0, 1\}^{N-k-1} } (f(b) \ket{0} \otimes \ket{b} \otimes \ket{c}  +  f(2^{k} -1 - b) \ket{1} \otimes \ket{b} \otimes \ket{c} ), 
} 
where $f: \{ 0, 1\}^{k} \to \R$ is a strictly increasing function, that is $f(x)> f(y)$ if $x>y$ and $Z$ is a normalizing factor. 
\end{definition} 
\noindent We construct monotone ladder states with $K=2^{k}$ levels on $k$ qubits that with a total gate complexity $4k$, that is logarithmic in the number of levels. 
The next claim describes the implementation of 2-qubit controlled rotation gates that are used for this construction.

\begin{claim} 
Define $R(\theta)$ to be the following family of two qubit gates, 
\begin{equation} 
R(\theta) =  \begin{pmatrix} 
& \cos(\theta)  & -\sin(\theta) & 0 & 0 \\  
&\sin(\theta)  & \cos(\theta) &0  &0 \\
&0 & 0  & \sin(\theta)   & -\cos(\theta) \\
&0 & 0 & \cos(\theta)     & \sin(\theta) 
\end{pmatrix}. 
\label{btheta} 
\end{equation}  
The $R(\theta)$ gate can be implemented using 2 X gates and 2 controlled rotation gates.  
\end{claim} 
\begin{proof} 
Let $C_{1}(\theta)$ be the two qubit controlled rotation gate where a rotation by angle $\theta$ is applied to qubit 2 conditioned on qubit 1 being 1, that is, 
\begin{equation} 
C_{1}(\theta) =  \begin{pmatrix} 
&1  & 0 & 0 & 0 \\  
&0  & 1 &0  &0 \\
&0 & 0  & \cos(\theta)   & -\sin(\theta) \\
&0 & 0 & \sin(\theta)     & \cos(\theta) 
\end{pmatrix}. 
\label{btheta2} 
\end{equation}  
The circuit $C_{2}(\theta)=(XI) C_{1}(\theta)  (XI)$ applies a controlled rotation conditioned on qubit 1 being equal to $0$. 
The identity $R(\theta) = C_{1}(\pi/2-\theta)  C_{2}(\theta)$ holds providing an implementation using $4$ gates as claimed. 
\end{proof} 

The following claim constructs ladder states with $K=2^{k}$ levels with a total gate complexity $O(\log K)$. In addition, the angles for the $R(\theta)$ gates 
can be chosen to be simple fractional multiples of $\pi$ providing a simple state preparation procedure that we can use to initialize the variational method.

\begin{claim} 
There exist monotone ladder states with $K=2^{k}$ levels on $N$ qubits that can be prepared by a quantum circuit with $N+3k$ one and two qubit gates and has depth $k$. 
\end{claim} 
\begin{proof} 
Select $k$ angles such that $\pi/2> \theta_{1} > \theta_{2} \cdots > \theta_{k}>\pi/4$ such that $\cot(\theta_{i}) < \cot^2(\theta_{i+1})$ for $i < k$. These angles can be chosen to be integral multiples of $\pi/2^{j}$ for small $j$ for convenience. With this choice of angles, we claim that the inequality $\cot(\theta_{i}) < \prod_{j = i+1}^{k} \cot(\theta_j)$ holds for all $i < k$. The inequality is clearly true for $i = k-1$, and assuming that $\cot(\theta_{t+1}) < \prod_{j = t+2}^{k} \cot(\theta_j)$ for some $t < k-1$, it follows that $\cot(\theta_{t}) < \cot^2(\theta_{t+1}) < \cot(\theta_{t+1}) \prod_{j = t+2}^{k} \cot(\theta_j) = \prod_{j = t+1}^{k} \cot(\theta_j)$.

A monotone ladder state on $N$ qubits with $K=2^{k}$ levels can be generated as follows, 
\all{ 
\ket{L_{K}} = \prod_{j = 1}^k R_{1 (j+1)}(\theta_{j}) (H \otimes I^k \otimes H^{N-k-1} ) \ket{0^{N}}, 
} {eq5} 
where the $R$ gate is conditioned on the first qubit and is applied sequentially to qubits numbered 2 through $k+1$. Hadamard gates are applied to the remaining $N-k-1$ qubits. The corresponding quantum circuit requires $N+3k$ gates and has depth $k$. 

The amplitudes of the state generated in equation \eqref{eq5} for standard basis state $\ket{1, b, c}$ (with $b \in \{0,1\}^{k}$ and $c \in \{0,1\}^{N-k+1}$) are proportional to 
\[f(b) = \prod_{i \not \in [b]} \sin(\theta_{i}) \prod_{i \in [b]} \cos(\theta_{i}).\] It follows from the choice of the $\theta_{i}$ that $f(b)$ is a strictly decreasing function, so the amplitudes of $\ket{1,b,c}$ are strictly decreasing in $b$. Indeed, suppose that $x,y \in \left\{0,1\right\}^k$ and $x < y$. If $x_i$ denotes the $i^{\text{th}}$ digit of $x$ starting from the right, then there exists some $m \in \left\{1,\dots,k\right\}$ such that $x_m = 0$, $y_m = 1$, and $x_j = y_j$ for all $j$ with $1 \le j < m$. Therefore \[\frac{f(y)}{f(x)} = \frac{\prod_{i \not \in [y]} \sin(\theta_{i}) \prod_{i \in [y]} \cos(\theta_{i})}{\prod_{i \not \in [x]} \sin(\theta_{i}) \prod_{i \in [x]} \cos(\theta_{i})} \le \cot(\theta_m)\prod_{j > m} \tan(\theta_j) < 1.\] Thus $f(x) > f(y)$ for all $x, y$ with $x < y$. Further the use of the $R_{1(j+1)}$ ensures that the amplitudes of $\ket{0, b, c}$ are proportional to 
$f(2^{k} -1 - b) = \prod_{i \in [b]} \sin(\theta_{i}) \prod_{i \not \in [b]} \cos(\theta_{i})$. Therefore, the amplitudes of $\ket{0, b, c}$ are strictly increasing in $b$. Thus, the circuit described in equation \eqref{eq5} prepares a ladder state as defined in Definition \ref{ladder}. 
\end{proof} 

\noindent We provided a general preparation procedure for ladder states, the angles $(3\pi/8, \pi/4, \pi/8)$ or $(\pi/3, \pi/4, \pi/6)$ could be  used to generate specific 8 level ladder states as initial states for quantum variational methods to prepare states corresponding to the normal distribution. The property required from the ladder state is that it should have a high overlap with the state representing the probability distribution, in some cases even a non-monotone ladder state can serve as a good initial state for the variational method. 

The above construction shows that monotone ladder states can be prepared using $O(\log K)$ gates where $K$ is the number of levels in the ladder. However, it does not provide much control over the amplitudes of the different levels in the ladder. It is well known that that arbitrary ladder states can be prepared by quantum circuits with $O(K)$ gates. 
A depth $O(K)$ construction follows from the work of Plesch and Brukner \cite{plesch} on the preparation of arbitrary quantum states. 

\begin{claim}\label{claim:arbitrary_state}\cite{plesch}
 An arbitrary quantum state on $n$ qubits can be prepared by a quantum circuit with $O(2^n)$ gates and depth $O(2^n)$.
\end{claim}
\noindent In the context of variational probability distribution loading, the initial state can be a $K$ level ladder approximation the probability distribution. The amplitudes are 
determined as the average values $f(k) = \sum_{2Lk/K \leq x \leq 2L(k+1)/K} p(x)$ and are prepared using the method in claim \ref{claim:arbitrary_state}. 

\begin{claim} 
Every ladder state with $K=2^{k}$ levels on $N$ qubits can be prepared by a quantum circuit with $N+O(K)$ gates and depth $O(K)$.
\end{claim} 
\begin{proof} 
Let $f: \{ 0, 1\}^{k} \to \R$ be any strictly increasing function. Define $Q_f$ to be a quantum circuit which prepares the quantum state \al{ 
\frac{1}{Z} \sum_{ b \in \{0, 1\}^{k}} (f(b) \ket{0} \otimes \ket{b}  +  f(2^{k} -1 - b) \ket{1} \otimes \ket{b}) 
} on $k+1$ qubits starting from $\ket{0^{k+1}}$, where $Z$ is a normalization factor and $Q_f$ has $O(2^{k+1}) = O(K)$ gates and depth $O(2^{k+1}) = O(K)$. The existence of $Q_f$ follows from Claim~\ref{claim:arbitrary_state}.

A ladder state on $N$ qubits with $K=2^{k}$ levels can be generated as follows 
\al{ 
\ket{L_{K}(f)} =    (Q_f \otimes H^{N-k-1} ) \ket{0^{N}}, 
} where the $Q_f$ gate is applied to the first $k+1$ qubits. Hadamard gates are applied to the remaining $N-k-1$ qubits. The corresponding quantum circuit requires $N+O(K)$ gates and has depth $O(K)$. 
\end{proof} 

The above construction of an arbitrary ladder provides a coarse approximation to the probability distribution, the variational method refines it to a better approximation to the distribution. In the next subsection, we show for completeness that any probability distribution can be approximated to arbitrarily high precision with a discrete approximation on a finite number of qubits.

\subsection{Approximating probability distributions with ladder states} \label{sec:ApproximatingLadders}

In order to approximate a symmetric probability density function $p(x)$ on an interval $[-L, L]$ within an error of less than $\epsilon$ using a state $\ket{\gamma}$ prepared by our variational method, it suffices to choose a sufficiently large number of levels $K$ with respect to $\epsilon$, $p$, and $L$. We show how to find a sufficiently large number of levels in the following proposition. Note that the number of levels that we find in the next result depends only on $\epsilon$, $p$, and $L$ (and not on the number of qubits $N$).

\begin{proposition} \label{prop:PDFLL}
Suppose that $p(x)$ is an even probability density function that is positive on $[-L,L]$ with $p'(x)$ continuous and $p'(x) < 0$ for $x > 0$. For every $\epsilon > 0$ and every number of qubits $N$, there exists a ladder state $\ket{\gamma}$ on $K = 2^k$ levels with $\norm{\gamma-p}_2 < \epsilon$ such that $k$ and $K$ depend only on $\epsilon$, $p$, and $L$. 
\end{proposition}

\begin{proof}
It suffices to prove this for $\epsilon \le 1$, so we suppose for the rest of the proof that $\epsilon \le 1$. Since $p'(x)$ is continuous and $[-L, L]$ is a closed and bounded interval, there exists a constant $M$ such that $|p'(x)| < M$ for all $x \in [-L, L]$. 

By definition, the $N$-qubit amplitude encoding of the probability distribution $D$ with density function $p(x)$ on the interval $[-L, L]$ is the quantum state  
\[\ket{p} = \frac{1} { Z} \sum_{i=0}^{2^{N}-1} p \bigg(-L+ \frac{(1+2i)L}{2^{N}} \bigg) \ket{i} 
\] 
with the normalization factor 
\[
Z = \bigg \{ \sum_{i=0}^{2^{N}-1}  p\bigg(-L+ \frac{(1+2i)L}{2^{N}}\bigg)^{2} \bigg \} ^{1/2}.
\]  
We initialize the ladder state with 
\[
k = \bigg\lceil \log_2 \left (\frac{M L}{\epsilon \min_{x \in [-L,L]} p(x)} \right ) \bigg\rceil,
\] so the number of levels is 
\[
K = 2^k \ge  \frac{M L 2^{N/2}}{\epsilon Z} \quad \textnormal{since} \quad Z \ge 2^{N/2} \min_{x \in [-L,L]} p(x).
\]
There are $2^{k+1}$ equally-spaced steps in this ladder state, each step with $2^{N-k-1}$ equal amplitudes. Note that the middle two steps are at the same level, so they will appear as a single step with $2^{N-k}$ equal amplitudes. 

It remains to define the strictly increasing function $f: \{ 0, 1\}^{k} \to \R$ which specifies the levels of our ladder state  
\[
\ket{L_{K}(f)} = \sumtwo_{\substack{b \in \{0, 1\}^{k} , \;  c \in \{0, 1\}^{N-k-1}}} (f(b) \ket{0} \otimes \ket{b} \otimes \ket{c}  +  f(2^{k} -1 - b) \ket{1} \otimes \ket{b} \otimes \ket{c} ).
\] 
For each $x \in \left\{0,1\right\}^k$, first define \[g(x) = \sqrt{\sum_{i = 0}^{2^{N-k-1}-1} p \left (-L+ \frac{(1+2(x2^{N-k-1}+i))L}{2^{N}} \right )^2}.\] Then we can take $f(x) = \frac{g(x)}{Y}$ with the normalization factor \[Y = \bigg(2 \sum_{i=0}^{2^{k}-1}  g(i)^2 \bigg) ^{1/2}.\] Observe that $Y = Z$ by definition. Moreover for any $i \in \left\{0,1\right\}^N$, \[|\ket{L_{K}(f)}_i-\ket{p}_i| < \frac{M L}{Z 2^k}\] since $|p'(x)| < M$ for all $x \in [-L, L]$. Thus \[\norm{L_K(f) - p}_2 < \frac{M L}{Z 2^k} 2^{N/2} \le \epsilon,\]
as it was to be shown.
\end{proof}

The conditions in Proposition \ref{prop:PDFLL} apply to a number of symmetric probability distributions including but not limited to Gaussian distributions, logistic distributions, Cauchy distributions, Voigt distributions, student's distributions, and uniform distributions. Note that we can handle non-zero means by shifting the parameters in the proof.

\subsection{Explicit ladder circuits}
We describe explicitly the quantum circuits used for constructing the ladder states used in the two examples in equations \eqref{lad1} and \eqref{eq:exampleinitialGaussian} in the introduction for generating the quantum states representing the normal distribution. For the description of these circuits, we restrict ourselves to a fixed number of qubits which is taken to be $N=6$ for concreteness. It is straightforward to generalize to a larger number of qubits. 

The ladder state with two levels corresponding to equation \eqref{lad1} has two levels, for $N=6$ qubits is given as,  
\al{ 
\ket{L_{2} } = \frac{1}{\sqrt{2}} ( \ket{01} + \ket{10} )  \otimes H^{\otimes 4} \ket{0000}.  
} 
The state $\frac{1}{\sqrt{2}} ( \ket{01} + \ket{10} )$ on the first two qubits can be generated starting with the standard basis state $\ket{01}$ and applying a Hadamard gate $H$ on the first qubit followed by a 
$\operatorname{CNOT}_{12}$ gate with the first qubit as control. The Hadamard gate is then applied to the 4 remaining qubits to obtain the ladder state $\ket{L_{2}}$ in equation \eqref{lad1}. More generally, the two level ladder state on $N$ qubits is generated as $\frac{1}{\sqrt{2}} ( \ket{01} + \ket{10} )$ the first two qubits followed by an application of Hadamard gates to the remaining $(N-2)$ qubits.

The ladder corresponding to the initial state in equation \eqref{eq:exampleinitialGaussian} has four levels, for $N=6$ qubits it is given as, 
\al{ 
\ket{L_{4} } = \frac{1}{4} ( \ket{010} + \sqrt{7} \ket{011} + \sqrt{7} \ket{100} + \ket{101} ) \otimes H^{\otimes 3} \ket{000}.  
} 
Let $H'\ket{0} = \frac{1}{2\sqrt{2}} \ket{0} + \frac{\sqrt{7}}{2\sqrt{2}} \ket{1} $ be a rotation gate that applies rotation by angle $\theta$ such that $\sin^{2}(\theta) = 1/8$.  
The circuit for generating the state $\ket{L_{4}}$ in equation \eqref{eq:exampleinitialGaussian} is given as follows. Starting  with initial state $\ket{010}$ apply a Hadamard to first qubit and 
$\operatorname{CNOT}_{12}$ with first qubit as control to get $\frac{1}{\sqrt{2}} ( \ket{010} + \ket{100})$. Then apply the rotation gate $H'$ on the third qubit to obtain, 
\al{ 
\frac{1}{4} ( \ket{010} + \sqrt{7} \ket{011} + \ket{100} + \sqrt{7} \ket{101} ).
} 
A $\operatorname{CNOT}_{13}$ gate with the first qubit as control yields the state $\frac{1}{4} ( \ket{010} + \sqrt{7} \ket{011} + \sqrt{7} \ket{100} + \ket{101} )$.  
The Hadamard gate is then applied to the 3 remaining qubits to obtain the ladder state $\ket{L_{4}}$in equation \eqref{lad1}. 
More generally, the state 
$\frac{1}{4} ( \ket{010} + \sqrt{7} \ket{011} + \sqrt{7} \ket{100} + \ket{101} )$ is generated on the first three qubits while the Hadamard gates are 
applied to the remaining $(N-3)$ qubits.  

These examples serve for illustrative purposes, ladder states approximating a given distribution with a fixed number of levels and a fixed number of qubits can be prepared using the techniques described in Section \ref{sec:ladder}. 


\section{Hamiltonian simulation algorithms} \label{sec:evolution} 

In Section \ref{sec:amplitudeencoding} we discussed the techniques for loading the ladder states that are coarse approximations for the desired PDFs. We now address the evolution of the ladder into the PDF as illustrated in Section \ref{sec:intro03}. The two proposed techniques (Trotter based simulation and variational quantum evolution) are well established algorithms that should not hinder the overall application of probability loading and should not add global overhead. We begin with Trotter simulation and then move on to variational techniques.

\subsection{Trotter based simulation}
The methods described above all reduce to the problem of quantum Hamiltonian simulation, where the goal is to approximate the state of the wave function given an initial state and a fixed Hamiltonian that is applied to the initial state for a certain time. Following \cite{nielsenchuang} solutions are obtained by approximating the state with a digital representation, then discretizing the Schr\"{o}dinger equation in space and time such that an iterative application of a well chosen strategy propagates the wave from the initial to the final state. A critical aspect for simulation to be successful is that the error in this strategy is bounded and it is known to not grow faster than some small power of the number of iterations. 

The Hamiltonians that we need to simulate are of the form $\hat H=\Delta + V(x)$, they are obtained by discretizing one dimensional Schr\"{o}dinger operators. The matrix  $\Delta$ is a  rescaled Laplacian matrix for the path and $V(x)$ is a diagonal matrix representing the potential term. This class of Hamiltonians is a sum of two terms $\Delta$ and $V(x)$ that are diagonal in the standard and the Fourier bases respectively. This allows us to obtain fairly efficient simulation algorithms for this class of Hamiltonians using the Trotter formula. 
\begin{claim} \label{trotter} 
\cite[$\mathsection$4.7.2]{nielsenchuang} One has that $e^{i(A+B)d t } = e^{iAdt/2 }e^{iB dt }e^{iAdt/2 } + O(dt^{3})$.
\end{claim} 
Let us make more precise the discretization of the Hamiltonians $\hat H$. Suppose that the number of qubits used by the algorithm is $N$ and the probability distribution is being approximated on the interval of size $2L$, 
the interval could be $[-L, L]$ for symmetric distributions and $[0, 2L]$ for positive valued distributions.  
The discretization step size is $\delta = 2L/2^{N}$. 

In order to approximate the second derivative operator in the continuous setting $\Delta$ is defined as $\Delta = \frac{1}{\delta^{2}} \mathcal{L}$ where $\mathcal{L}$ is the Laplacian matrix for the path with $2^{N}$ nodes with entries $\mathcal{L}_{ii}= 2$, $\mathcal{L}_{ij} = -1$ if $|i-j|=1$ and $\mathcal{L}_{ij}=0$ otherwise. The matrix $\Delta$ is diagonal in the Fourier basis and its eigenvalues are $\lambda_{j} = 2(1-\cos(2\pi i/2^{N}))$. 

The error bounds and resource requirements for the Hamiltonians of the form $\hat H=\Delta + V(x)$ are analyzed in the following claim. 
\begin{claim} 
A Hamiltonian of the form $\hat H=\Delta + V(x)$  can be simulated for time $t$ with error $\epsilon$ using 
$O((C_{V} + C_{\Delta} + \log^{2} N ) t^{3/2}/\epsilon^{1/2})$ gates where $C_V$ and $C_\Delta$ are the cost of implementing the phase oracle $\ket{j} \to e^{2\pi i \lambda_{j}} \ket{j}$ for the matrices $\Delta$ and $V$. 
\end{claim} 
\begin{proof} 
As $\Delta$ and $V$ are diagonal in the standard and Fourier bases respectively, we claim that the cost of implementing a single step of the Trotter decomposition 
$e^{i(V+\Delta) dt } = e^{iV dt/2 }e^{i\Delta dt }e^{i V dt/2 }$ in Claim \ref{trotter} is $(2C_{V} + C_{\Delta} + 2\log^{2} N)$. The quantum circuit implementation of 
$e^{iV dt/2 }$ reduces to the application of the phase oracle $O_{V}$ in the standard basis while the circuit for $e^{i\Delta dt }$ reduces to an application of the 
phase oracle for $\Delta$ in the Fourier basis. The total cost of a single iteration of the Trotter circuit is $(2O_{V} + O_{\Delta} + 2\log^{2} N)$. 

If the step size for the time evolution $dt$ is chosen to be $(\epsilon/t)^{1/2}$, then the total error accrued by the algorithm is $(t/dt)dt^{3}= t dt^{2} = \epsilon$. The total number of iterations for the algorithm is $t/dt=t^{3/2}/\epsilon^{1/2}$, the claimed complexity bound follows by multiplying the cost of a single iteration by the total number of iterations. 
\end{proof} 
As the above claim shows, the Trotter based approach is fairly efficient for this class of Hamiltonians, with the complexity comparable to $O(1/\epsilon^{1/2})$ applications of the quantum Fourier transform.

\subsection{VarQRTE based simulation}
In bracket notation \eqref{eq:TDSE1D} becomes
\begin{align}
    \frac{d\ket{\psi(t)}}{dt} = - i {\hat H}\ket{\psi(t)}.
\end{align}
We shall follow the outline from \cite[$\mathsection$A.1.1]{variational} as well as \cite[Supplementary Material]{general} and \cite[$\mathsection$3.1]{feynmankac} for a description of real time evolution. VarQRTE is an algorithm to approximate QRTE. Rather than evolving quantum states in the complete (exponential) Hilbert state space, the evolution is \textit{approximated} by mapping it to the ansatz parameters via one of the variational principles. Typically one chooses among McLachlan's, Dirac and Frenkel's, or the time dependent variational principle. Although the tradeoff is an approximation, the gain is the possibility of implementing this approximation using shallow quantum circuits, suitable for near-term devices. The parameterized trial state is $\ket{\phi(\boldsymbol{\theta}(t))}$ with $\boldsymbol{\theta}(t) = (\theta_1(t), \theta_2(t), \cdots, \theta_N(t))$. We shall assume that $\theta_j \in \R$ for $j \in \{1,2,\dots,N\}$ where $N$ is an integer depending on the ansatz circuit of our choice. The choice of ansatz is very important in order to guarantee the performance of the algorithm. The dynamical evolution of $\ket{\Psi(t)}$ can be simulated by introducing the above mentioned ansatz $\ket{\phi(\boldsymbol{\theta}(t))} = \mathbf{G} (\boldsymbol{\theta}(t)) \ket{0}^{\otimes n}$ where $\mathbf{G} (\boldsymbol{\theta}(t)) = \prod_{i=1}^N \mathbf{G}_i (\theta_i(t))$ is the product of $N$ parametric unitaries $\mathbf{G}_i$, each composed of one parametric rotation gates $e^{i \theta_k \mathfrak{G}_k}$ with $\mathfrak{G}_k = \mathfrak{G}_k^\dagger$. Following the McLachlan variational principle \cite{McLachlan}, the evolution for the angles $\theta_j$ is given by the system of first order ODEs given by
\begin{align} \label{eq:preeuler}
    \sum_{j=0}^N \real(M_{k,j}) {\dot \theta}_j = \imag(V_k), \quad \textnormal{for each} \quad k=0,1,\cdots,N,
\end{align}
where the elements of the matrix $M$ and the vector $V$ are given by
\begin{align}
    M_{k,j} = \frac{\partial \bra{\phi({\vec \theta}(t))}}{\partial \theta_k}\frac{\partial \ket{\phi({\vec \theta}(t))}}{\partial \theta_j} \quad \textnormal{and} \quad V_k = \frac{\partial \bra{\phi({\vec \theta}(t))}}{\partial \theta_ki} {\hat H} \ket{\phi({\vec \theta}(t))}.
\end{align}
Since $\ket{{\tilde v}(\boldsymbol \theta (t))}$ is implemented with parameterized unitaries, the terms $M_{k,j}$ and $V_k$ can be computed parametrically by using the quantum circuit shown in Figure~\ref{fig:expValueCircuit}. 

\begin{figure}[h!]
\begin{center}
\begin{quantikz}
\lstick{$\tfrac{\ket{0} + e^{i\theta}\ket{1}}{\sqrt{2}}$} & \qw & \qw & \gate{X} & \ctrl{1} & \gate{X} & \qw & \qw & \ctrl{1} & \gate{H} & \meter{}\\
\lstick{$\ket{0}$}  & \gate{\mathbf{G}_N} & \gate{\cdots} & \gate{\mathbf{G}_k} & \gate{\mathfrak{G}_k} & \gate{\mathbf{G}_{k-1}} & \gate{\cdots} & \gate{\Omega} & \gate{\Upsilon} & \qw & \qw
\end{quantikz}
\caption{Circuit to evaluate the matrix elements $M_{k,j}$ and the vector elements $V_k$ depending on $\Omega$ and $\Upsilon$.}
\label{fig:expValueCircuit}
\end{center}
\end{figure}
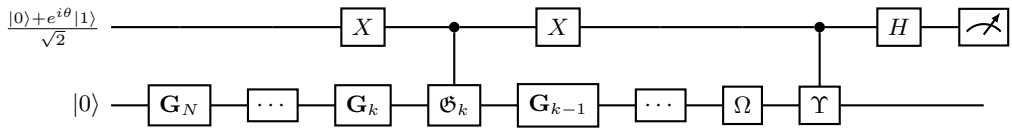

In order to obtain the matrix elements $M_{k,j}$ and the vector elements ${V_k}$ we make the following choices.  
\begin{itemize}
    \item One takes the sequence of gates $\mathbf{G}_{k-1}, \cdots, \Omega$ to be $\mathbf{G}_{k-1}$ to $\mathbf{G}_j$ and $\Upsilon = \mathfrak{G}_j$, if we are interested in $M_{k,j}$,
    \item or one sets the sequence of $\mathbf{G}_{k-1}, \cdots, \Omega$ to be $\mathbf{G}_{k-1}$ to $\mathbf{G}_1$ and then $\Upsilon = U_j$, where $U_n$ are easily-implementable unitary operators such as multi-qubit Pauli operators as represented with $\hat H = \sum_{h=1}^H \lambda_h U_h$, if we are interested in $V_k$. Here $\lambda_h \in \C$ and $U_h$ is an easily implementable unitary operator made of Pauli operators. The details can be found for instance in \cite[$\mathsection$3.1]{feynmankac}. 
\end{itemize}
Further details on the respective implementation are given in  \cite{variational, variationalansatz, zoufal_error_bounds}.

From \eqref{eq:preeuler} and using the above formulas for $M_{k,j}$ and $V_k$, we can now obtain the evolution of the parameters $\dot \theta_k$ using the forward Euler method
\begin{align}
\boldsymbol{\theta}(\tau + \delta \tau) \sim {\boldsymbol \theta}(\tau) + {\dot {\boldsymbol \theta}} \delta \tau = {\boldsymbol \theta}(\tau) + {\bf M}^{-1}(\tau) \cdot {\bf V} \delta \tau. \nonumber
\end{align}
That is, McLachlan's principle defines an ODE which can be numerically solved with an arbitrary ODE solver with $N_T=\tau/\delta \tau$ time steps.

For imaginary time evolutions, the choice of ansatz was \texttt{RealAmplitudes} \cite[$\mathsection$4.2]{feynmankac}. This was a circuit that contained only real gates (CNOTs as well as $R_y$ rotations) embedded with circular entanglement. The reason was because a heat equation, with real solutions, was being solved for its associated heat kernel. In this case, however, we are solving for the Schr\"{o}dinger equation and the solution will be a complex wave function. Therefore, adding $R_x$ and $R_z$ rotations in addition to $R_y$ rotations could be advantageous. We also need to consider that an ansatz which is sufficiently expressive, such that it corresponds to an (approximate) $t$-design \cite{60,61}, is prone to result in barren plateaus \cite{62}. The chosen ansatz is therefore a natural upgrade to the one in \cite[$\mathsection$4.2]{feynmankac} which had a good tradeoff between expressivity and potential propagation issues. 

Typically $n=6$ qubits has worked well. We start by setting

\begin{quantikz}
&&\ghost{H} \\
&& \ghost{H}
\end{quantikz}
$U_n(\boldsymbol{\theta}) = U_n(\boldsymbol{\theta}_1,\boldsymbol{\theta}_2,\cdots,\boldsymbol{\theta}_6)=$
\begin{quantikz}
\qw & \gate{R_x(\theta_{n,1,1})}\gategroup[wires=6, steps=3, style={dashed, rounded corners,fill=yellow!10, inner xsep=2pt}, background]{Initial set of angles} & \gate{R_y(\theta_{n,1,2})} & \gate{R_z(\theta_{n,1,3})} & \qw \\ 
\qw & \gate{R_x(\theta_{n,2,1})} & \gate{R_y(\theta_{n,2,2})} & \gate{R_z(\theta_{n,2,3})} & \qw \\ 
\qw & \gate{R_x(\theta_{n,3,1})} & \gate{R_y(\theta_{n,3,2})} & \gate{R_z(\theta_{n,3,3})} & \qw \\ 
\qw & \gate{R_x(\theta_{n,4,1})} & \gate{R_y(\theta_{n,4,2})} & \gate{R_z(\theta_{n,4,3})} & \qw \\ 
\qw & \gate{R_x(\theta_{n,5,1})} & \gate{R_y(\theta_{n,5,2})} & \gate{R_z(\theta_{n,5,3})} & \qw \\ 
\qw & \gate{R_x(\theta_{n,6,1})} & \gate{R_y(\theta_{n,6,2})} & \gate{R_z(\theta_{n,6,3})} & \qw 
\end{quantikz}

so that the ansatz circuit is $\mathbf{G} (\boldsymbol{\theta}(t)) \ket{0}^{\otimes 6}$

\begin{center}
\begin{quantikz}
\lstick{$\ket{0}$} & \gate{H}\gategroup[wires=6, steps=2, style={dashed, rounded corners,fill=green!10, inner xsep=2pt}, background]{Initialize} & \gate[6,disable auto height]{U_0(\boldsymbol{\theta})} & \ctrl{1}\gategroup[wires=6, steps=7, style={dashed, rounded corners,fill=blue!10, inner xsep=2pt}, background]{Repeat \textit{n} times} & \qw & \qw & \qw &  \qw & \targ{} & \gate[6,disable auto height]{U_{n \ge 1}(\boldsymbol{\theta})} & \meter{}\\ 
\lstick{$\ket{0}$} & \gate{H} & & \targ{} & \ctrl{1} & \qw & \qw & \qw & \qw & ~ & \meter{}\\ 
\lstick{$\ket{0}$} & \gate{H} &  & \qw & \targ{} & \ctrl{1} & \qw  & \qw & \qw & ~ & \meter{}\\ 
\lstick{$\ket{0}$} & \gate{H} &  & \qw & \qw & \targ{} & \ctrl{1}  & \qw & \qw & ~ & \meter{}\\ 
\lstick{$\ket{0}$} & \gate{H} &  & \qw & \qw & \qw & \targ{} & \ctrl{1} & \qw & ~ & \meter{}\\ 
\lstick{$\ket{0}$} & \gate{H} &  & \qw & \qw & \qw & \qw & \targ{} & \ctrl{-5} & ~ & \meter{} 
\end{quantikz}
\end{center}

We discussed two fairly efficient approaches for Hamiltonian simulation for one dimensional Hamiltonians of the form $\hat H= \Delta^{2} + V(x)$, these methods are applicable for arbitrary potentials and can thus be used for all the examples discussed in sections II and III. Further, some Hamiltonians with closed form propagators are known to be fast-forwardable and can be simulated more efficiently than the generic Trotter or VarQITE based methods. For the quantum Harmonic oscillator the gate complexity for the simulation can be reduced to the cost of two quantum Fourier transforms together with controlled phase operations in the standard basis as indicated in equation \eqref{qho}. Simlar disentangling formulae are known for the linear potential $V(x)=ax$ as well. Establishing the fast-forwardability of the other Hamiltonians with closed form propagators considered in Section II remains an open question.  

\section{Discussion and conclusion} \label{sec:conclusion}
It is desirable to have a broader toolkit for the probability distribution loading problem that can operate with multiple probability distributions and also their extension to 
multivariate settings. Further, in order to obtain substantial quantum speedups for the task of distribution loading it is also desirable to use quantum native techniques rather 
than methods that artificially quantize a classical sampling algorithm. As we discussed in the introduction, the recent progress in quantum loading of probability distributions has had mixed success with respect to these goals. Many of these techniques are designed to work for the Gaussian distribution but they are not easily reusable for different probability distributions.  

In this paper we propose some new methods that avoid the `classical to quantum' paradigm and that work natively with quantum Hamiltonians whose Feynman propagators have an analytic closed form. The main advantage of this approach is that one can leverage the large number of tools from quantum mechanics such as the literature related to the solutions of the Schr\"{o}dinger equation for particular physical systems and the theory of the Feynman path integral \cite{handbookPI, susyQM, kleinertPI}. The Hamiltonians also have close connections to the classical theory of special functions and could potentially be components of a broader toolkit that is able to operate with a larger family of probability distributions. In particular, we were able to obtain Laplace and Maxwell-Boltzmann distributions using these methods. Further, efficient Hamiltonian simulation methods for the radial and the P\"{o}schl-Teller potentials would yield efficient algorithms for loading the two-parameter family of $\Gamma$ and $\beta$ distributions \cite{susyQM}.

Further, one may expect that due to the special form of these propagators the Hamiltonian simulation for them may be more efficient than for generic Hamiltonians. One of the important theoretical questions suggested by this paper is the relation between fast-forwardable Hamiltonians \cite{atia2017fast, gu2021fast} and the Hamiltonians with closed form Feynman propagators \cite{handbookPI}. In particular, some of these Hamiltonians are known to be fast forwardable with the  the resources needed to simulate them sub-linear in the time $t$ and fast-forwardability opens up the possibility that these distribution loading algorithms could be useful in the near-term. The ideal complexity for Hamiltonian simulation will have gate complexity comparable to a constant number of quantum Fourier transforms, this can be achieved for example for the Harmonic oscillator and thus for the normal distribution loading problem. The known algorithms for normal distribution loading based on the quantum singular value transform have gate complexity equivalent to the application of $O(\log n)$ Fourier transforms \cite{mcardle2022quantum}.

On the other hand, the disadvantage of employing these physics-inspired techniques is that one has to accept the spectrum of tools already available for such ends. In particular, if the solution to a given potential is not found in the literature, then one has to work it out and this is an analytically difficult task. However, when the mathematical analysis works (such as in the case of the harmonic oscillator), the advantage in the distribution loading is potentially exponential and is native to quantum. The simulation of one dimensional Hamiltonians thus provides a reasonably efficient and generally applicable solution for loading probability distributions. We leave a more systematic study of finding the more efficient Hamiltonian simulation methods for the Hamiltonians having a given probability distribution as ground state for future work. 

On the practical side, we showed empirically that the accuracy of the loading is subject to the precision of the ladder and that will act as the limiting factor for the evolution based algorithms. The limitation for the ground state projection based algorithms is the efficiency of Hamiltonian simulation and the spectral gap for the Hamiltonian under consideration. 

\section{Acknowledgements}
AP and NR would like to acknowledge participation in the IBM Quantum Finance Working Group Kickoff that took place in October 2022 in New York City where the ideas for this manuscript were first discussed. The authors are thankful to Alvin Moon and Teddy Parker for useful insights. In addition, NR wishes to acknowledge discussions with Amol Deshmukh, Noelle Ibrahim, Vaibhaw Kumar, Stefan Woerner, and Christa Zoufal.



\section{Appendix: Additional results on radial path integrals and Bessel functions} \label{sec:appendix1}
The purpose of this Appendix is to elaborate on how to compute Feynman path integrals analytically and the type of mathematical techniques needed to produce closed-formed solutions. The methodologies to deal with Feynman path integrals can be analytic, like the approach that we take here, but they also have algebraic or group theoretical analogues and counterparts, see for instance \cite{inomatajunkerLie, inomataKayed2, reckoning, kleinertPI}. In order to illustrate a non-trivial example, we have stepped aside from the harmonic oscillator in one dimension and decided to present the radial harmonic oscillator in $n$ dimensions. The radial harmonic oscillator in one dimension was worked out in \cite{peakinomata} by Peak and Inomata in 1969. Inomata revisited the calculation in \cite{reckoning} in 1994. Unfortunately, there are several serious misprints, typos, and some inaccuracies in \cite{reckoning}. We now welcome the opportunity to provide an amended version. See also \cite{groscheSteiner1987} for further information on path integrals in curved manifolds. Along the way, we have established generalized results on integrals of Bessel functions that have not appeared in the literature nor are they tabulated in tables of integrals or mathematical software, \cite{higher1, higher2, watson, gradryz}. These results will be of independent interest in the theory of Bessel functions and in the toolkit for computations of Feynman path integrals.

We shall first start with results on Bessel functions and the so-called Weber integrals, which are as follows. Let $\real(\nu)>-1$, $|\arg(p)| < \frac{\pi}{4}$ and $a,b >0$. One has that \cite[$\mathsection$7.7.3]{higher2} and \cite[$\mathsection$13.32]{watson}
\begin{align} \label{eq:weberformula}
   \int_0^\infty e^{-p^2 t^2} J_\nu(at) J_\nu(bt) t dt = \frac{1}{2p^2} \exp\bigg[-\frac{a^2+b^2}{4p^2}\bigg] I_\nu \bigg(\frac{ab}{2p^2}\bigg) ,
\end{align}
where $J_\nu(z)$ is the Bessel function of the first kind and $I_\nu(z)$ is the modified Bessel function of the first kind, both of which will be defined in Section \ref{sec:appendix2}.

\subsection{Some preliminary new results on Bessel integrals} \label{sec:appendix2}
Let $J_\alpha(x)$ and $I_\alpha(x)$ be the regular and modified Bessel functions of the first kind respectively defined by \cite{gradryz, watson}
\begin{align} \label{eq:appendAux01}
    J_\alpha(x) = \sum_{m=0}^\infty \frac{(-1)^m}{m!\Gamma(m+\alpha+1)}\bigg( \frac{x}{2}\bigg)^{2m+\alpha} \quad \textnormal{and} \quad I_\alpha(x) = i^{-\alpha} J_\alpha (ix).
\end{align}
The change of variables $x \to x^h$ in (2.30) of the Glasser-Montaldi paper \cite{montaldiglasser} with $h > 0$ leads to
\begin{align} \label{eq:appendAux02}
{J_v}(a{x^h}){J_v}(b{x^h}) = {\left( {\frac{1}{2}\frac{{ab{x^h}}}{{\sqrt {{a^2} + {b^2}} }}} \right)^v}\sum\limits_{r = 0}^\infty  \frac{1}{{r!\Gamma (v + r + 1)}}{{\left( {\frac{{ab{x^h}}}{{2\sqrt {{a^2} + {b^2}} }}} \right)}^{2r}}{J_{v + 2r}}({x^h}\sqrt {{a^2} + {b^2}} ).
\end{align}
Moreover, expanding $J_z$ we can integrate termwise to find
\begin{align} \label{eq:appendAux03}
\int_0^\infty  {e^{ - p{x^2}}} &{x^{hv + 2hr + 1}}{J_{v + 2r}}({x^h}\sqrt {{a^2} + {b^2}} )dx \nonumber \\
&= \frac{1}{{2p}}{\bigg( {\frac{{\sqrt {{a^2} + {b^2}} }}{2}} \bigg)^{v + 2r}} \sum\limits_{m = 0}^\infty  {\frac{{{{( - 1)}^m}}}{{m!\Gamma (m + v + 2r + 1)}}{{\bigg( {\frac{{\sqrt {{a^2} + {b^2}} }}{2}} \bigg)}^{2m}}\frac{{\Gamma (h(m+v + 2r ) + 1)}}{{{p^{h(v + 2r + m)}}}}} .
\end{align}
If $h$ were equal to $1$ then the gamma functions in the summand of \eqref{eq:appendAux03} would cancel and this would lead to the much simpler formula
\begin{align}
\int_0^\infty  {{e^{ - p{x^2}}}{x^{v + 2r + 1}}{J_{v + 2r}}(x\sqrt {{a^2} + {b^2}} )dx}  = \frac{{{{({a^2} + {b^2})}^{v/2 + r}}}}{{{{(2p)}^{v + 2r + 1}}}}\exp \left( { - \frac{{{a^2} + {b^2}}}{{4p}}} \right).
\end{align}
Setting $\Upsilon = \Upsilon(h)$ to be the integral we are interested in, we see that
\begin{align} 
  \Upsilon (h) :&= \int_0^\infty  {x{e^{ - p{x^2}}}{J_v}(a{x^h}){J_v}(b{x^h})dx}  \nonumber \\
   &= \int_0^\infty  {x{e^{ - p{x^2}}}{{\left( {\frac{1}{2}\frac{{ab{x^h}}}{{\sqrt {{a^2} + {b^2}} }}} \right)}^v}} \sum\limits_{r = 0}^\infty  {\frac{1}{{r!\Gamma (v + r + 1)}}{{\left( {\frac{{ab{x^h}}}{{2\sqrt {{a^2} + {b^2}} }}} \right)}^{2r}}{J_{v + 2r}}({x^h}\sqrt {{a^2} + {b^2}} )} dx \nonumber \\
   &= {\left( {\frac{1}{2}\frac{{ab}}{{\sqrt {{a^2} + {b^2}} }}} \right)^v}\sum\limits_{r = 0}^\infty  {\frac{1}{{r!\Gamma (v + r + 1)}}{{\left( {\frac{{ab}}{{2\sqrt {{a^2} + {b^2}} }}} \right)}^{2r}}} \int_0^\infty  {{e^{ - p{x^2}}}{x^{hv + 2hr + 1}}{J_{v + 2r}}({x^h}\sqrt {{a^2} + {b^2}} )dx}  \nonumber \\
   &= {\left( {\frac{1}{2}\frac{{ab}}{{\sqrt {{a^2} + {b^2}} }}} \right)^v}\sum\limits_{r = 0}^\infty  {\frac{1}{{r!\Gamma (v + r + 1)}}{{\bigg( {\frac{{ab}}{{2\sqrt {{a^2} + {b^2}} }}} \bigg)}^{2r}}} \frac{1}{{2p}}{\bigg( {\frac{{\sqrt {{a^2} + {b^2}} }}{2}} \bigg)^{v + 2r}} \nonumber \\
   &\quad \times \sum\limits_{m = 0}^\infty  {\frac{{{{( - 1)}^m}}}{{m!\Gamma (m + v + 2r + 1)}}{{\bigg( {\frac{{\sqrt {{a^2} + {b^2}} }}{2}} \bigg)}^{2m}}\frac{{\Gamma (hv + 2hr + hm + 1)}}{{{p^{h(v + 2r + m)}}}}}  \nonumber \\
   &= \frac{1}{{2{p^{1 + hv}}}}{\left( {\frac{{ab}}{4}} \right)^v}\sum\limits_{r = 0}^\infty  {\frac{1}{{r!\Gamma (v + r + 1)}}{{\left( {\frac{{ab}}{{4{p^h}}}} \right)}^{2r}}}  \nonumber \\
   &\quad \times \sum\limits_{m = 0}^\infty  {\frac{{{{( - 1)}^m}}}{{m!\Gamma (m + v + 2r + 1)}}{{\bigg( {\frac{{\sqrt {{a^2} + {b^2}} }}{2}} \bigg)}^{2m}}\frac{{\Gamma (hv + 2hr + hm + 1)}}{{{p^{hm}}}}} . \nonumber
\end{align}
Cleaning up and re-arrenging the terms leads to
\begin{align} \label{eq:coupledsum}
 \Upsilon(h) = \frac{1}{{2{p}}}{\left( {\frac{{ab}}{4p^h}} \right)^v}\sumtwo_{m,r \ge 0}  \frac{{\Gamma (h(m + v + 2r) + 1)}}{{\Gamma (m + v + 2r + 1)}}\frac{1}{{r!\Gamma (v + r + 1)}}
 {\left( {\frac{{ab}}{{4{p^h}}}} \right)^{2r}}
 \frac{{{{( - 1)}^m}}}{{m!}}{\bigg( {\frac{{{{a^2} + {b^2}} }}{{4{p^{h}}}}} \bigg)^{m}}. 
\end{align}
Once again the difficulty is in the presence of general $h$, for if $h=1$ then we would have
\begin{align} 
  \Upsilon_v :=\Upsilon (1) &= \int_0^\infty  {x{e^{ - p{x^2}}}{J_v}(ax){J_v}(bx)dx}  \nonumber \\
   &= \frac{1}{{2{p^{1 + v}}}}{\left( {\frac{{ab}}{4}} \right)^v}\sum\limits_{m = 0}^\infty  {} \frac{{{{( - 1)}^m}}}{{m!}}{\bigg( {\frac{{{{a^2} + {b^2}} }}{{4{p}}}} \bigg)^{m}}\sum\limits_{r = 0}^\infty  {} \frac{1}{{r!\Gamma (v + r + 1)}}{\left( {\frac{{ab}}{{4p}}} \right)^{2r}} \nonumber \\
   &= \frac{1}{{2p}}\exp \left( { - \frac{{{a^2} + {b^2}}}{{4p}}} \right){I_v}\left( {\frac{{ab}}{{2p}}} \right),
\end{align}
which is Weber's formula \eqref{eq:weberformula}. Therefore the difficulty of the analysis resides in the term 
\begin{align} \label{def:wcoupling}
w(m,r,v;h) = w(h) := \frac{{\Gamma (h(m + v + 2r) + 1)}}{{\Gamma (m + v + 2r + 1)}}
\end{align}
which keeps the $m$- and $r$-sums in $\Upsilon(h)$ from decoupling, except when $h=1$ for which $w(1)=1$. Employing $\Gamma(z+1)=z \Gamma(z)$ we can simplify $w$ to read
\begin{align}
 w(h) = \frac{h(m + v + 2r)\Gamma (h(m + v + 2r))}{(m + v + 2r)\Gamma (m + v + 2r)} = h\frac{\Gamma (h(m + v + 2r))}{\Gamma (m + v + 2r)}.
\end{align}
Assuming that $h$ is an integer (by analytic continuation, this could be extended beyond that range) we can use the duplication formula
\begin{align}
    \Gamma(h z) = (2\pi)^{\frac{1-h}{2}} h^{hz-\frac{1}{2}} \prod_{k=0}^{h-1} \Gamma\bigg(z+\frac{k}{h}\bigg).
\end{align}
Therefore dividing by $\Gamma$ we can cancel out the $\Gamma$ in the denominator so that
\begin{align}
    h\frac{\Gamma(hz)}{\Gamma(z)} = (2\pi)^{\frac{1-h}{2}} h^{hz+\frac{1}{2}} \prod_{k=1}^{h-1} \Gamma\bigg(z+\frac{k}{h}\bigg) .
\end{align}
Setting $z=m+v+2r$ and employing the above allows us to write
\begin{align} \label{eq:auxcorollary}
     w(h) = {(2\pi )^{\frac{{1 - h}}{2}}}{h^{h(m + v + 2r) + \frac{1}{2}}}\prod\limits_{k = 1}^{h - 1} {\Gamma \left( {m+v+r+2r + \frac{k}{h}} \right)},
\end{align}
effectively trading the difficulty of dealing with a gamma function in the denominator by a product of gamma functions. Algebraic manipulations with the product \eqref{eq:auxcorollary} allow us to prove the following result on a $(h-1)$-folded integral representation of $\Upsilon(h)$.
\begin{proposition}
    One has that
    \begin{align}
\Upsilon (h) = \frac{{{{(2\pi )}^{\tfrac{{1 - h}}{2}}}{h^{\tfrac{1}{2}}}}}{{2p}}\underbrace {\int_0^\infty  { \cdots \int_0^\infty  {} } }_{h - 1} &{I_v}\bigg[ {{\left( {\frac{h}{p}} \right)}^h}{\frac{{ab}}{2}\prod\limits_{k = 1}^{h - 1} {{x_k}} } \bigg]  \exp \bigg[ { - {{\left( {\frac{h}{p}} \right)}^h}\frac{{{a^2} + {b^2}}}{4}\prod\limits_{k = 1}^{h - 1} {{x_k}} } \bigg]\prod\limits_{k = 1}^{h - 1} {x_k^{\tfrac{k}{h} - 1}{e^{ - {x_k}}}d{x_k}} .\nonumber
\end{align}
for $a,b,p$ and $v$ as in \eqref{eq:weberformula} and $h \in \N$.
\end{proposition}
For instance if $h=1$ then we clearly recover \eqref{eq:weberformula}.
\begin{corollary} \label{cor:Upsilon2}
    One has the identity
    \begin{align} 
  \Upsilon (2) &= \int_0^\infty  {x{e^{ - p{x^2}}}{J_v}(a{x^2}){J_v}(b{x^2})dx} \nonumber \\
   &= \frac{1}{2 \sqrt{\pi}}\frac{1}{{\sqrt {{a^2} + {b^2} + {p^2}}}} \frac{{\Gamma (\tfrac{1}{2} + v)}}{{\Gamma (1 + v)}}{\left( {\frac{{ab}}{{{a^2} + {b^2} + {p^2}}}} \right)^v} {_2}{F_1} \bigg( {\frac{{1 + 2v}}{4},\frac{{3 + 2v}}{4},1 + v,{{\bigg( {\frac{{2ab}}{{{a^2} + {b^2} + {p^2}}}} \bigg)}^2}} \bigg), \nonumber
\end{align}
for $a,b,p$ and $v$ as in \eqref{eq:weberformula}.
\end{corollary}
This seems to be a new result in the literature and one that advanced mathematical software cannot compute. 
\begin{proof}[Proof of Corollary \textnormal{\ref{cor:Upsilon2}}]
    To find out the value at $h=2$ we now make use of the formula \cite[$\mathsection$11.3]{oberhettinger}
\begin{align} \label{eq:someequation}
    \int_0^\infty I_v(A x) \exp(-B x) x^{-1/2} e^{-x} dx &= \frac{2^{-v}}{\sqrt{1+B}} \bigg(\frac{A^2}{(1+B)^2}\bigg)^{v/2} \frac{\Gamma(\frac{1}{2}+v)}{\Gamma(1+v)}  {_2}{F_1}\bigg(\frac{1+2v}{4},\frac{3+2v}{4},1+v,\frac{A^2}{(1+B)^2}\bigg), 
\end{align}
for $\real(A-B) \le 1$ and $\real(v) > -\frac{1}{2}$. Next we write
\begin{align} \label{eq:Upsilon2}
  \Upsilon (2) &= \int_0^\infty  {x{e^{ - p{x^2}}}{J_v}(a{x^2}){J_v}(b{x^2})dx} = \frac{1}{{2\sqrt \pi  p}}\int_0^\infty  {} {I_v}\left( {\frac{{2ab}}{{{p^2}}}x} \right)\exp \left( { - \frac{{{a^2} + {b^2}}}{{{p^2}}}x} \right){x^{ - \tfrac{1}{2}}}{e^{ - x}}dx .
\end{align}
The result now follows by using \eqref{eq:someequation} in \eqref{eq:Upsilon2} with $A = \frac{2ab}{p^2}$ and $B=\frac{a^2+b^2}{p^2}$.
\end{proof}
As an aside, we can also find
\begin{align}
  \Upsilon (3) &= \frac{{\sqrt 3 }}{{12\pi }}{\left( {\frac{{ab}}{{{a^2} + {b^2}}}} \right)^v}{\left( {\frac{4}{{{a^2} + {b^2}}}} \right)^{1/3}}\frac{{\Gamma (\tfrac{1}{3} + v)}}{{\Gamma (\tfrac{1}{2}(\tfrac{1}{3} + v))\Gamma (\tfrac{1}{2}(\tfrac{1}{3} + 1 + v))}} \nonumber \\
   &\quad \times \sum\limits_{k = 0}^\infty  {\frac{{\Gamma (\tfrac{1}{2}(\tfrac{1}{3} + v) + k)\Gamma (\tfrac{1}{2}(\tfrac{1}{3} + 1 + v) + k)}}{{\Gamma (1 + v + k)}}\frac{{\Gamma (\tfrac{2}{3} + v + 2k)}}{{\Gamma (1 + k)}}{{\left( {\frac{{2ab}}{{{a^2} + {b^2}}}} \right)}^{2k}}} U\left( {\frac{1}{3} + v + 2k,\frac{2}{3},\frac{4}{{27}}\frac{{{p^3}}}{{{a^2} + {b^2}}}} \right), \nonumber  
\end{align}
where $U(a,b,z)$ is the confluent hypergeometric function \cite{gradryz, watson} by first noticing that 
\begin{align}
  \Upsilon (3) &= \frac{{{3^{\tfrac{1}{2}}}}}{{4\pi p}}\int_0^\infty  {\int_0^\infty  {} } {I_v}\left( {\frac{{27ab}}{{2{p^3}}}xy} \right)\exp \left( { - \frac{{27({a^2} + {b^2})}}{{4{p^3}}}xy} \right){x^{ - \tfrac{2}{3}}}{y^{ - \tfrac{1}{3}}}{e^{ - x}}{e^{ - y}}dxdy \nonumber \\
   &= {\left( {\frac{A}{2}} \right)^v}\frac{{\sqrt 3 }}{{4\pi p}}\frac{{\Gamma (\tfrac{1}{3} + v)}}{{\Gamma (1 + v)}}\int_0^\infty  {{e^{ - y}}{y^{v - \tfrac{1}{3}}}{{(1 + By)}^{ - v - \tfrac{1}{3}}}{\,_2}{F_1}\left( {\frac{1}{6}(1 + 3v),\frac{1}{6}(4 + 3v),1 + v;\frac{{{A^2}{y^2}}}{{{{(1 + By)}^2}}}} \right)dy}  \nonumber  
\end{align}
where $A = \frac{{27ab}}{{2{p^3}}}$ and $B = \frac{{27}}{{4{p^3}}}({a^2} + {b^2})$ and then employing \cite{gradryz}
\begin{align}
_2{F_1}(a_1,a_2,a_3;x) = \sum\limits_{k = 0}^\infty  {\frac{{\Gamma (a_1 + k)\Gamma (a_2 + k)\Gamma (a_3)}}{{\Gamma (a_1)\Gamma (a_2)\Gamma (a_3 + k)}}\frac{{{x^k}}}{{\Gamma (k + 1)}}}.
\end{align}
This leaves us with
\begin{align}\Upsilon (3) = {\left( {\frac{A}{2}} \right)^v}\frac{{\sqrt 3 }}{{4\pi p}}\frac{{\Gamma (\tfrac{1}{3} + v)}}{{\Gamma (1 + v)}}\sum\limits_{k = 0}^\infty  {\frac{{{A^{2k}}\Gamma (a_1 + k)\Gamma (a_2 + k)\Gamma (a_3)}}{{\Gamma (a_1)\Gamma (a_2)\Gamma (a_3 + k)\Gamma (k + 1)}}} \int_0^\infty  {{e^{ - y}}{y^{v - \tfrac{1}{3} + 2k}}{{(1 + By)}^{ - v - \tfrac{1}{3} - 2k}}\,dy} , \nonumber
\end{align}
where $a_1 = \tfrac{1}{6}(1 + 3v), a_2 = \tfrac{1}{6}(1 + 3v)$ and $a_3=1+v$. The integral can be computed as \cite{gradryz, watson}
\begin{align}\int_0^\infty  {{e^{ - y}}{y^{v - \tfrac{1}{3} + 2k}}{{(1 + By)}^{ - v - \tfrac{1}{3} - 2k}}\,dy}  = {B^{ - \tfrac{1}{3} - 2k - v}}\Gamma \left( {\frac{2}{3} + 2k + v} \right)U\left( {\frac{1}{3} + 2k + v,\frac{2}{3},\frac{1}{B}} \right).
\end{align}

We are now going to demonstrate how to find the values of the derivatives of $\Upsilon$. For the purpose of illustration we shall compute $\Upsilon'(1)$. Going back \eqref{eq:coupledsum} and using \eqref{def:wcoupling} we define
\begin{align}
    \vartheta_{r,m}(h) := \frac{1}{{2{p}}}{\left( {\frac{{ab}}{4p^h}} \right)^v} w(m,r,v;h) \frac{1}{{r!\Gamma (v + r + 1)}}{\left( {\frac{{ab}}{{4{p^h}}}} \right)^{2r}}\frac{{{{( - 1)}^m}}}{{m!}}{\bigg( {\frac{{{{a^2} + {b^2}} }}{{4{p^{h}}}}} \bigg)^{m}}
\end{align}
so that $\Upsilon(h)$ can now be written as
\begin{align}
\Upsilon(h) = \sumtwo_{0 \le r,m \le \infty} \vartheta_{r,m}(h).
\end{align}
We can easily find that
\begin{align}
    -\vartheta_{r,m}'(1) &= \frac{(-1)^{m+2} m \sqrt{a^2+b^2} p^{-v-\frac{3}{2}} (\log p) (a b)^v 2^{-2 m-4 r-2 v-1}
   }{m! r! \Gamma(r+v+1)} \bigg(\frac{\sqrt{a^2+b^2}}{\sqrt{p}}\bigg)^{2 m-1} \left(\frac{a b}{p}\right)^{2 r}\nonumber \\
   &\quad + \frac{a b (-1)^{m+2} r p^{-v-2} (\log p) (a b)^v 2^{-2 m-4 r-2 v}
   }{m! r! \Gamma
   (r+v+1)} \bigg(\frac{\sqrt{a^2+b^2}}{\sqrt{p}}\bigg)^{2 m} \left(\frac{a b}{p}\right)^{2 r-1} \nonumber \\
   &\quad +\frac{(-1)^{m+2} v p^{-v-1} (\log p) (a b)^v 2^{-2 m-4 r-2 v-1}
   }{m! r! \Gamma
   (r+v+1)} \bigg(\frac{\sqrt{a^2+b^2}}{\sqrt{p}}\bigg)^{2 m} \left(\frac{a b}{p}\right)^{2 r}\nonumber \\
   &\quad +\frac{(-1)^{m+1} p^{-v-1} (a b)^v 2^{-2 m-4 r-2 v-1} 
   }{m! r! \Gamma (r+v+1)} \bigg(\frac{\sqrt{a^2+b^2}}{\sqrt{p}}\bigg)^{2 m} \left(\frac{a b}{p}\right)^{2 r} (m+2 r+v)\psi(m+2r+v+1) \nonumber \\
   &= \theta_1(r,m) + \theta_2(r,m) + \theta_3(r,m) + \theta_4(r,m), \nonumber
\end{align}
where $\psi(s) = \frac{\Gamma'}{\Gamma}(s)$ is the digamma function, \cite{gradryz}. The first three terms $\theta_1(r,m), \theta_2(r,m)$ and $\theta_3(r,m)$ can be summed over both $m$ and $r$ and consequently we can write
\begin{align}
    -\Upsilon'(1) &= \sum_{r=0}^\infty \sum_{m=0}^\infty (-\vartheta_{r,m}'(1)) = \sum_{r=0}^\infty \sum_{m=0}^\infty \sum_{i=1}^4\theta_i(r,m) \nonumber \\
    &= \frac{\log p}{8p^2}  \exp\bigg(-\frac{a^2+b^2}{4 p}\bigg) \left[\left(4 p v-(a^2+b^2)\right) I_v\left(\frac{a b}{2 p}\right)+2 a b I_{v+1}\left(\frac{a b}{2p}\right)\right] + \sumtwo_{m,r \ge 0} \theta_4(r,m) \nonumber \\
   &= \frac{{\log p}}{{4p}}\left[ {(4pv - ({a^2} + {b^2})){\Upsilon _v} + 2ab{\Upsilon _{v + 1}}} \right]  + \sumtwo_{m,r \ge 0} \theta_4(r,m),
\end{align}
where $\Upsilon _v$ was defined in \eqref{eq:weberformula}. The difficulty of the term involving $\theta_4(r,m)$ resides in the fact that the variables $m$ and $r$ are coupled in $(m+2 r+v)\psi(m+2r+v+1)$. We shall need to employ Mellin transforms to decouple these variables. Specifically, we shall need (\cite{oberhettinger} for the former and \cite{axioms} for the latter)
\begin{align}
    \Gamma'(s) = \int_0^\infty x^{s-1}e^{-x} \log x dx \quad \textnormal{and} \quad \frac{1}{\Gamma(s)} = \frac{1}{2\pi}\int_{-\infty}^\infty (c+iy)^{-s} e^{c+iy}dy,
\end{align}
for $c>0$ and $\real(s)>0$. We then write the digamma function as a product of two integrals 
\begin{align}
    \psi(s) = \frac{\Gamma'}{\Gamma}(s) = \bigg(\int_0^\infty x^{s-1}e^{-x} \log x dx\bigg) \bigg(\frac{1}{2\pi} \int_{-\infty}^\infty (c+iy)^{-s} e^{c+iy}dy \bigg),
\end{align}
with the key observation that $s$ is, at the cost of bringing up these integrals, nicely decoupled into $x^{s-1}$ and $(c+iy)^{-s}$. Therefore, defining $\mathfrak{S}_4$ to be the sums of $\theta_4(r,m)$ we get that
\begin{align}
    \mathfrak{S}_4 :&=\sumtwo_{m,r \ge 0} \theta_4(r,m) \nonumber \\
    &= - \frac{1}{2\pi}\frac{1}{{8p^2}}\int_0^\infty  {\int_{ - \infty }^\infty  {} } {e^{ - x}}(\log x){e^{c + iy}}\frac{{{x^2}}}{{{{(c + iy)}^2}}}\exp \left( { - \frac{{{a^2} + {b^2}}}{{4p}}\frac{x}{{c + iy}}} \right) \nonumber \\
   &\quad \times \left[ {\left( {4pv\frac{{c + iy}}{x} - ({a^2} + {b^2})} \right){I_v}\left( {\frac{{ab}}{{2p}}\frac{x}{{c + iy}}} \right) + 2ab{I_{v + 1}}\left( {\frac{{ab}}{{2p}}\frac{x}{{c + iy}}} \right)} \right]dydx ,
\end{align}
for any $c>0$. Essentially this amounts to taking the result for $\sum_{r=0}^\infty \sum_{m=0}^\infty \sum_{i=1}^3\theta_i(r,m)$ and making the identification $p \to p \frac{c+iy}{x}$ minding the new term $\log p$ as well as the new negative sign. Therefore we arrive at
\begin{align}
    -\Upsilon'(1)  &= \frac{\log p}{8p^2}  \exp\bigg(-\frac{a^2+b^2}{4 p}\bigg) \left[\left(4 p v-(a^2+b^2)\right) I_v\left(\frac{a b}{2 p}\right)+2 a b I_{v+1}\left(\frac{a b}{2p}\right)\right] \nonumber \\
    &\quad - \frac{1}{2\pi}\frac{1}{{8p^2}}\int_0^\infty  {\int_{ - \infty }^\infty  {} } {e^{ - x}}(\log x){e^{c + iy}}\frac{{{x^2}}}{{{{(c + iy)}^2}}}\exp \left( { - \frac{{{a^2} + {b^2}}}{{4p}}\frac{x}{{c + iy}}} \right) \nonumber \\
   &\quad \quad \times \left[ {\left( {4pv\frac{{c + iy}}{x} - ({a^2} + {b^2})} \right){I_v}\left( {\frac{{ab}}{{2p}}\frac{x}{{c + iy}}} \right) + 2ab{I_{v + 1}}\left( {\frac{{ab}}{{2p}}\frac{x}{{c + iy}}} \right)} \right]dydx .
\end{align}

Another type of generalization of Weber's formula can be accomplished as follows. Suppose we wish to compute the following generic integral
\begin{align}
\Omega (a,A,B,\nu;s) = \int_0^\infty  \exp \left[ { - A{x^2} - \frac{B}{{4x}}} \right]{J_\nu }\left( {ax} \right)x^{s-1}dx,
\end{align}
with $A, B >0$ and $a >0$. To that end we shall need the following Mellin transforms
\begin{align}
\int_0^\infty  \exp \left[ { - A{x^2} - \frac{B}{{4x}}} \right]{x^{s - 1}}dx  &=  - \frac{1}{8}{A^{(1 - s)/2}}B\Gamma \left( {\frac{{s - 1}}{2}} \right){\,_0}{F_2}\left( {;\frac{3}{2};\frac{{3 - s}}{2}; - \frac{{A{B^2}}}{{64}}} \right) \nonumber \\
 &\quad + \frac{1}{2}{A^{ - s/2}}\Gamma \left( {\frac{s}{2}} \right){\,_0}{F_2}\left( {;\frac{1}{2};1 - \frac{s}{2}; - \frac{{A{B^2}}}{{64}}} \right) \nonumber \\
 &\quad + {\left( {\frac{B}{4}} \right)^s}\Gamma ( - s){\,_0}{F_2}\left( {;\frac{s}{2} + \frac{1}{2};\frac{s}{2} + 1; - \frac{{A{B^2}}}{{64}}} \right),
\end{align}
which was found by the use of mathematical software, as well as \cite[$\mathsection$11]{oberhettinger}
\begin{align}
\int_0^\infty  {J_\nu }\left( ax \right){x^{s - 1}}dx  = \frac{2^{s-1}a^{-s}\Gamma(\frac{s}{2}+\frac{\nu}{2})}{\Gamma(-\frac{s}{2}+\frac{\nu}{2}+1)} .
\end{align}
The technique is the same as in Section \ref{sec:list04}, where we used Mellin transforms and the Cauchy residue theorem. The connecting formula we need to use is
\begin{align}
\int_0^\infty  {f(x)g(x){x^{s - 1}}dx}  = \frac{1}{{2\pi i}}\int_{(k)} \mathfrak{F}(w)\mathfrak{G}(s - w)dw ,
\end{align}
where $\mathfrak{F}$ is the Mellin transform of $f$ and $\mathfrak{G}$ that of $g$. Therefore, after some calculations and the application of the Cauchy residue theorem we arrive at
\begin{align}
\Omega (a,A,B,\nu;s) = -\lim_{N \to \infty} \sum_{n=0}^N \frac{1}{2\pi i} \oint_{\gamma(n)} \Omega(a,A,B,\nu;s)ds = -\sum_{n=0}^\infty h(n),
\end{align}
with the minus due to the orientation of the contour and $\gamma(n)$ being small circles around the poles of $\Omega$, and where
\begin{align}
    h(n) &= \frac{(-1)^{n+1} a^{2 n+\nu} 2^{-6 n-2 s-3 \nu-3} A^{\frac{1}{2} (-2 n-s-\nu)}}{n! \Gamma (n+\nu+1)} \nonumber \\
    & \quad \times \bigg[ 4^{2 n+s +\nu} \bigg\{4 \Gamma \left(\frac{1}{2} (2 n+s+\nu)\right) \, _0F_2\left(;\frac{1}{2},-n-\frac{s}{2}-\frac{\nu}{2}+1;-\frac{A
   B^2}{64}\right) \nonumber \\
   & \quad \quad -\sqrt{A} B \Gamma \left(\frac{1}{2} (2 n+s+\nu-1)\right) \,
   _0F_2\left(;\frac{3}{2},-n-\frac{s}{2}-\frac{\nu}{2}+\frac{3}{2};-\frac{A B^2}{64}\right)\bigg\} \nonumber \\
   &\quad \quad \quad + 8 (A^{\frac{1}{2}} B)^{2 n+s+\nu} \Gamma (-2 n-s-\nu) \,
   _0F_2\left(;n+\frac{s}{2}+\frac{\nu}{2}+\frac{1}{2},n+\frac{s}{2}+\frac{\nu}{2}+1;-\frac{A B^2}{64}\right) \bigg].
\end{align}
This works when $\nu \notin \N$, but requires some additional care when $\nu \in \N$ or when $B \to 0$. In fact, the sum over $n$ can be performed when $B = 0$ and one finds
\begin{align}
    \Omega(a,A,0,\nu;s) = \frac{2^{-\nu-1} a^\nu A^{-\frac{s}{2}-\frac{\nu}{2}} \Gamma \left(\frac{s}{2}+\frac{\nu}{2}\right) \,
   _1F_1\left(\frac{s}{2}+\frac{\nu}{2};\nu+1;-\frac{a^2}{4 A}\right)}{\Gamma (\nu+1)}.
\end{align}
Now we have to bring the other Bessel function $J_\nu(bx)$ into the game. The key is to use again \cite[Eq. (2.30)]{montaldiglasser}
\begin{align}
{J_\nu }(ax){J_\nu }(bx) = {\left( {\frac{1}{2}\frac{{abx}}{{\sqrt {{a^2} + {b^2}} }}} \right)^\nu }\sum\limits_{r = 0}^\infty  {{{\left( {\frac{{abx}}{{2\sqrt {{a^2} + {b^2}} }}} \right)}^{2r}}\frac{1}{{r!\Gamma (\nu  + r + 1)}}{J_{\nu  + 2r}}\left( {x\sqrt {{a^2} + {b^2}} } \right)} .
\end{align}
Indeed, we see that
\begin{align}
  \Xi (a,b,A,B,\nu) &:= \int_0^\infty  {\exp \left( { - A{x^2} - \frac{B}{{4x}}} \right){J_\nu }(ax){J_\nu }(bx)xdx}  \nonumber \\
   &= \int_0^\infty  \exp \left( { - A{x^2} - \frac{B}{{4x}}} \right){{\left( {\frac{1}{2}\frac{{abx}}{{\sqrt {{a^2} + {b^2}} }}} \right)}^\nu } \nonumber \\
   &\quad \times \sum\limits_{r = 0}^\infty  {{{\left( {\frac{{abx}}{{2\sqrt {{a^2} + {b^2}} }}} \right)}^{2r}}\frac{1}{{r!\Gamma (\nu  + r + 1)}}{J_{\nu  + 2r}}\left( {x\sqrt {{a^2} + {b^2}} } \right)} xdx  \nonumber \\
   &= \left( {\frac{1}{2}\frac{{ab}}{{\sqrt {{a^2} + {b^2}} }}} \right)^\nu \sum\limits_{r = 0}^\infty  {\frac{1}{{r!\Gamma (\nu  + r + 1)}}{{\left( {\frac{{ab}}{{2\sqrt {{a^2} + {b^2}} }}} \right)}^{2r}}} \nonumber \\
   &\quad \times \int_0^\infty  \exp \left( { - A{x^2} - \frac{B}{{4x}}} \right){x^{\nu  + 2r + 1}}{J_{\nu  + 2r}}\left( {x\sqrt {{a^2} + {b^2}} } \right)dx  \nonumber \\
   &= {\left( {\frac{1}{2}\frac{{ab}}{{\sqrt {{a^2} + {b^2}} }}} \right)^\nu }\sum\limits_{r = 0}^\infty  \frac{\Omega \big( {\sqrt {{a^2} + {b^2}} ,A,B,\nu  + 2r,\nu  + 2r + 2} \big)}{{r!\Gamma (\nu  + r + 1)}}{{\left( {\frac{{ab}}{{2\sqrt {{a^2} + {b^2}} }}} \right)}^{2r}}.
\end{align}
This represents our second generalization of Weber's formula \eqref{eq:weberformula}, which follows by letting $B \to 0$. It is instructive to expand around $B=0$. To do so we first write
\begin{align}
h(a,A,B,\nu ,s;n) = :{B^{2n + s + \nu }}{h_0}(a,A,B,\nu ,s;n) + {h_1}(a,A,B,\nu ,s;n)
\end{align}
where
\begin{align}
{h_0}(a,A,B,\nu ,s;n) &= \frac{{{{( - 1)}^{1 + n}}{2^{ - 6n - 2s - 3\nu }}{a^{2n + \nu }}\Gamma ( - 2n - s - \nu )}}{{n!\Gamma (n + \nu  + 1)}} \nonumber \\ 
&\quad + \frac{{{{( - 1)}^{1 + n}}{2^{ - 2(2 + 3n + s) - 3\nu }}{a^{2n + \nu }}\Gamma ( - 2 - 2n - s - \nu )}}{{\Gamma (n + 1)\Gamma (n + \nu  + 1)}}{B^2} + O({B^4}),
\end{align}
as well as
\begin{align}
  {h_1}(a,A,B,\nu ,s;n) &= \frac{{{{( - 1)}^{1 + n}}{2^{ - 1 - 2n - \nu }}{a^{2n + \nu }}{A^{( - 2n - s - \nu )/2}}\Gamma (\tfrac{{2n + s + \nu }}{2})}}{{n!\Gamma (n + \nu  + 1)}} \nonumber \\
  &\quad+ \frac{{{{( - 1)}^n}{2^{ - 3 - 2n - \nu }}{a^{2n + \nu }}{A^{(1 - 2n - s - \nu )/2}}\Gamma (\tfrac{{2n + s + \nu  - 1}}{2})}}{{n!\Gamma (n + \nu  + 1)}}B \nonumber \\
   &\quad + \frac{{{{( - 1)}^{1 + n}}{2^{ - 6 - 2n - \nu }}{a^{2n + \nu }}{A^{1 - n - s/2 - \nu /2}}\Gamma (\tfrac{{2n + s + \nu  - 2}}{2})}}{{\Gamma (n + 1)\Gamma (n + \nu  + 1)}}{B^2} \nonumber \\
   &\quad + \frac{{{{( - 1)}^n}{2^{ - 8 - 2n - \nu }}{a^{2n + \nu }}{A^{(3 - 2n - s - \nu )/2}}\Gamma (\tfrac{{2n + s + \nu  - 3}}{2})}}{{3\Gamma (n + 1)\Gamma (n + \nu  + 1)}}{B^3} + O({B^4}) .
\end{align}
This allows us to divide $\Omega$ as
\begin{align}
  \Omega \left( {a,A,B,\nu ,s} \right) =  - \sum\limits_{n = 0}^\infty  {h(a,A,B,\nu ,s;n)}  &=  - \sum\limits_{n = 0}^\infty  {{B^{2n + s + \nu }}{h_0}(a,A,B,\nu ,s;n) + {h_1}(a,A,B,\nu ,s;n)}  \nonumber \\
   &= :{\Omega _0}\left( {a,A,B,\nu ,s} \right) + {\Omega _1}\left( {a,A,B,\nu ,s} \right)  
\end{align}
For $\Omega_1$ we then see that
\begin{align}
  \Omega_1 \left( {a,A,B,\nu ,s} \right) &=  - \frac{1}{{{2^0}}}\Omega \left( {a,A,0,\nu ,s} \right) + \frac{1}{{{2^2}}}\Omega \left( {a,A,0,\nu ,s - 1} \right)B  - \frac{1}{{{2^5}}}\Omega \left( {a,A,0,\nu ,s - 2} \right){B^2} \nonumber \\
  &\quad + \frac{1}{3}\frac{1}{{{2^7}}}\Omega \left( {a,A,0,\nu ,s - 3} \right){B^3} - \frac{1}{3}\frac{1}{{{2^{11}}}}\Omega \left( {a,A,0,\nu ,s - 4} \right){B^4} + \frac{1}{{15}}\frac{1}{{{2^{13}}}}\Omega \left( {a,A,0,\nu ,s - 5} \right){B^5} \nonumber \\
   &\quad - \frac{1}{{45}}\frac{1}{{{2^{16}}}}\Omega \left( {a,A,0,\nu ,s - 6} \right){B^6} + \frac{1}{{315}}\frac{1}{{{2^{18}}}}\Omega \left( {a,A,0,\nu ,s - 7} \right){B^7}  - \frac{1}{{315}}\frac{1}{{{2^{23}}}}\Omega \left( {a,A,0,\nu ,s - 8} \right){B^8} \nonumber \\
   &\quad + \frac{1}{{2835}}\frac{1}{{{2^{25}}}}\Omega \left( {a,A,0,\nu ,s - 9} \right){B^9} - \frac{1}{{14175}}\frac{1}{{{2^{28}}}}\Omega \left( {a,A,0,\nu ,s - 10} \right){B^{10}} + O({B^{11}}) .
\end{align}
Apart from the sign alternation, the sequence of powers of $2$ is 
$
\{0,2,5,7,11,13,16,18,23,25,28, \cdots \},
$
which are the denominators in expansion of $(1-x)^{-1/4}$, and the denominators form the sequence 
\[
\{1,1,1,3,3,15,45,315,315,2835,14175, \cdots\}
\]  
which correspond to the largest odd divisor of $n!$, according to the Online Encyclopedia of Integer Sequences \cite{oeis}. This is interesting because the case $B=0$ is, essentially, Weber's formula. 

We now proceed to show an application of Weber's formula in the context of the path integral of the radial harmonic oscillator.

\subsection{Revisiting the radial harmonic oscillator in $n$ dimensions} \label{sec:appendix3}

Let $n \in \N$. Our starting point is the path integral for an arbitrary propagator in $n$ dimensions
\begin{align} \label{eq:propagatorn}
    K(n,{\bf r}'',{\bf r}';t'',t') = \lim_{N \to \infty} \int \prod_{j=1}^N K(n,{\bf r}_j, {\bf r}_{j-1}, t_j - t_{j-1}) \prod_{j=1}^{N-1} d^n {\bf r}_j,
\end{align}
where, for $\tau_j = t_j - t_{j-1}$, the short time propagator is given by
\begin{align} \label{eq:shortimepropagator}
    K(n,{\bf r}_j, {\bf r}_{j-1}, t_j - t_{j-1}) = \bigg(\frac{M}{2\pi i \hbar \tau_j}\bigg)^{n/2} \exp \bigg(\frac{iM}{2\hbar \tau_j} (\Delta {\bf r}_j)^2 - \frac{i}{\hbar}V(r_j)\tau_j\bigg).
\end{align}
In the discretization procedure we have set $t'=t_0, t''=t_N$ as well as ${\bf r}_j={\bf r}(t_j)$ and, most importantly, $r_j = |{\bf r}_j|$. This means that the time interval $\tau = t_N - t_0$ is divided into $N$ subintervals, and as $N \to \infty$ we have that $\tau_j \to 0$. The choice $\tau_j = T/N$ is called the isometric and it is the simplest choice, but it is not unique.

The propagator $K$ that appears in \eqref{eq:propagatorn} satisfies the composition rule
\begin{align}
    \int K(n,{\bf r}'',{\bf r};t'',t)K(n,{\bf r},{\bf r}';t,t')d^n r = K(n,{\bf r}'',{\bf r}';t'',t')
\end{align}
as well as the initial condition
$
    \lim_{t'' \to t'} K(n,{\bf r}'',{\bf r}';t'',t') = \delta^{(n)} ({\bf r}''-{\bf r}')
$.
It is well-known from the literature that the procedure to switch from Cartesian to polar coordinates requires the Ozaki-Edwards-Gulyaev switch \cite{EdwardsGulyaev} that we will discuss shortly. This switch rewrites the short time space interval $(\Delta {\bf r}_j)^2$ appearing in the action of \eqref{eq:shortimepropagator} as
\begin{align} \label{eq:deltar2Switch}
    (\Delta {\bf r}_j)^2 = r_j^2 r_{j-1}^2 - 2 r_j r_{j-1} \cos \Theta_j
\end{align}
where the angle $\Theta_j$ is
$
    \Theta_j = \arccos (\frac{{\bf r}_j \cdot {\bf r}_{j-1}}{r_j r_{j-1}})
$.
The next step is to use the Jacobi-Anger (or Gegenbauer) expansion formula \cite{gradryz, watson}
\begin{align}
    \exp(z \cos \Theta) = \bigg(\frac{2}{z}\bigg)^\nu \Gamma(\nu) \sum_{\ell=0}^\infty (\ell + \nu) C_\ell^{(\nu)} (\cos \Theta) I_{\ell + \nu}(z),
\end{align}
where $I_\nu(z)$ is the modified Bessel function and $C_\ell^{(\nu)}(z)$ is the Gegenbauer function \cite{gradryz, watson}. However, if we make the choices $\nu = \frac{n-2}{2}, \Theta = \Theta_j$ and $z = \frac{M}{i\hbar \theta_j} r_j r_{j-1}$,
then it leads us to
\begin{align}
\exp \left( \frac{M}{i\hbar \tau _j}{r_j}{r_{j - 1}}\cos \Theta _j \right) &= \left( \frac{2i\hbar \tau _j}{M r_j r_{j - 1}} \right)^{(n - 2)/2}\Gamma \left( \frac{n - 2}{2} \right) \nonumber \\
& \quad \times \sum_{\ell  = 0}^\infty  \left( \ell  + \frac{n - 2}2 \right)C_\ell ^{((n-2)/2)}\left( \cos \Theta _j \right)I_{\frac{n - 2}{2}}\left( \frac{M}{i\hbar {\tau _j}}r_jr_{j - 1} \right).
\end{align}
The normalized spherical function is defined by \cite{gradryz, watson}
\begin{align}
    C_n^{(\lambda )}\left( x \right) = \frac{\Gamma \left( \lambda  + \frac{1}{2} \right)}{\Gamma \left( 2\lambda  \right)}\frac{\Gamma \left( n + 2\lambda \right)}{\Gamma \left( n + \lambda  + \frac{1}{2} \right)}P_n^{(\lambda  - 1/2,\lambda  - 1/2)}\left( x \right),
\end{align}
where $P_n^{(a,b)}(x)$ are the Jacobi polynomials \cite{gradryz}. Some care needs to be applied when the gamma functions become singular. This means that
\begin{align}
    C_\ell ^{((n - 2)/2)}\left( \cos \Theta _j \right) = \frac{\Gamma \left( \frac{n - 1}{2} \right)}{\Gamma \left( n - 2 \right)}\frac{\Gamma \left( \ell  + n - 2 \right)}{\Gamma \left( {\ell  + \frac{n - 1}{2}} \right)}P_\ell ^{((n - 3)/2,(n - 3)/2)}\left( \cos \Theta _j \right)
\end{align}
This implies, using \eqref{eq:deltar2Switch}, that the short time propagator in polar coordinates is given by
\begin{align} \label{eq:shorttimepropagatorpolarcoordinates}
  K(n,{\mathbf{r}}_j,{{\mathbf{r}}_{j - 1}};{t_j} - {t_{j - 1}}) &= {\left( {\frac{M}{{2\pi i\hbar {\tau _j}}}} \right)^{n/2}}\exp \left[ {\frac{{iM}}{{2\hbar {\tau _j}}}\left( {r_j^2 + r_{j - 1}^2} \right) - \frac{i}{\hbar }V({r_j}){\tau _j}} \right]\exp \left( {\frac{M}{{i\hbar {\tau _j}}}{r_j}{r_{j - 1}}\cos {\Theta _j}} \right) \nonumber \\
   &= \sum\limits_{\ell  = 0}^\infty  {{K_\ell }(n,{r_j},{r_{j - 1}};{\tau _j})P_\ell ^{((n - 3)/2,(n - 3)/2)}\left( {\cos {\Theta _j}} \right)} 
\end{align}
where the quantity $K_\ell$ is the radial short time propagator
\begin{align} \label{eq:radialshorttimepropagator}
  K_\ell (n,{r_j},{r_{j - 1}};{\tau _j}) &= 2^{2 - n}\pi ^{(1 - n)/2}\frac{M}{{i\hbar {\tau _j}}}{\left( {{r_j}{r_{j - 1}}} \right)^{ - (n - 2)/2}}\exp \left[ {\frac{{iM}}{{2\hbar {\tau _j}}}\left( {r_j^2 + r_{j - 1}^2} \right) - \frac{i}{\hbar }V({r_j}){\tau _j}} \right] \nonumber \\
   &\quad \times \left( {\ell  + \frac{{n - 2}}{2}} \right)\frac{{\Gamma \left( {\ell  + n - 2} \right)}}{{\Gamma \left( {\ell  + \tfrac{{n - 1}}{2}} \right)}}{I_{\ell  + \frac{{n - 2}}{2}}}\left( {\frac{{M{r_j}{r_{j - 1}}}}{{i\hbar {\tau _j}}}} \right) .
\end{align}
As a sanity check, we note that when $n \to 3$, this reduces to
\begin{align} 
{K_\ell }(3,{r_j},{r_{j - 1}};{\tau _j}) &= \frac{1}{{4\pi }}\frac{M}{{i\hbar {\tau _j}}}{\left( {{r_j}{r_{j - 1}}} \right)^{ - 1/2}}\exp \left[ {\frac{{iM}}{{2\hbar {\tau _j}}}\left( {r_j^2 + r_{j - 1}^2} \right) - \frac{i}{\hbar }V({r_j}){\tau _j}} \right] \nonumber \\
&\quad \times \sum_{\ell  = 0}^\infty  {\left( {2\ell  + 1} \right){P_\ell }\left( {\cos {\Theta _j}} \right)I_{\ell  + \frac{1}{2}}\left( {\frac{{M{r_j}{r_{j - 1}}}}{{i\hbar {\tau _j}}}} \right)} 
\end{align}
by the use of $P_\ell ^{(0,0)}(x) = {P_\ell }(x) = C_\ell ^{(1/2)}(x)$, which coincides with \cite[$\mathsection$2.4.1]{coherentPI}. Here $P_\ell(x)$ is the Legendre polynomial, \cite{gradryz}. Before proceeding to compute \eqref{eq:radialshorttimepropagator}, it worth getting rid of the angular coordinates by performing the integration with respect to $\Theta_j$ in \eqref{eq:propagatorn}. The change of variables that we have performed from Cartesian to polar changes the volume element as
$
    d^n {\bf r}_j = r_j^{n-1} dr_j d\Theta_j
$,
as we can show by bringing in the appropriate Jacobian matrix. For example, if $n=3$, then this becomes $d^3{\bf r}_j = r_j^2 dr_j \sin \theta_j d\theta_j d\phi_j$. Another important tool we shall use is the integral relation of the normalized spherical function \cite{gradryz, watson}
\begin{align} \label{firstangleintegration}
    \int_0^\pi  P_\ell ^{((n - 3)/2,(n - 3)/2)}(\cos \Theta )P_{\ell '}^{((n - 3)/2,(n - 3)/2)}(\cos \Theta )(\sin \Theta )^{n - 2}d\Theta = \frac{{{2^{n - 2}}}}{{2\ell  + n - 2}}\frac{{\Gamma {{(\ell  + \tfrac{{n - 3}}{2} + 1)}^2}}}{{\Gamma (\ell  + n - 2)\ell !}}{\delta _{\ell \ell '}},
\end{align}
as well as
\begin{align} 
\int_0^{2\pi } {\frac{{{\pi ^{(1 - n)/2}}}}{2}\frac{{\Gamma (\ell  + \tfrac{{n - 1}}{2})}}{{\ell !}}d\phi }  = {\pi ^{1 + (1 - n)/2}}\frac{{\Gamma (\ell  + \tfrac{{n - 1}}{2})}}{{\ell !}} = {\pi ^{(3 - n)/2}}\frac{{\Gamma (\ell  + \tfrac{{n - 1}}{2})}}{{\Gamma (\ell  + 1)}}.
\end{align}
This will have the advantage of cleanly separating the angular part from the radial part as
\begin{align} \label{eq:radialshorttimepropagatorANGLESOUT}
{K_\ell }(n,{r_j},{r_{j - 1}};{\tau _j}) &= {\pi ^{(3 - n)/2}}\frac{{\Gamma (\ell  + \frac{{n - 1}}{2})}}{{\Gamma (\ell  + 1)}}\frac{M}{{i\hbar {\tau _j}}}{\left( {{r_j}{r_{j - 1}}} \right)^{ - (n - 2)/2}} \nonumber \\
&\quad \times \exp \left[ {\frac{{iM}}{{2\hbar {\tau _j}}}\left( {r_j^2 + r_{j - 1}^2} \right) - \frac{i}{\hbar }V({r_j}){\tau _j}} \right] {I_{\ell  + \frac{{n - 2}}{2}}}\left( {\frac{{M{r_j}{r_{j - 1}}}}{{i\hbar {\tau _j}}}} \right).
\end{align}
Therefore going back to \eqref{eq:shorttimepropagatorpolarcoordinates}
\begin{align} 
K(n,\mathbf{r}'',\mathbf{r}';t'',t') = \sum_{\ell  = 0}^\infty K_\ell (n,r'',r',t'',t')\mathcal{P}_\ell^{(n)} (\hat{\bf r}'' \cdot \hat{\bf r}'),
\end{align}
where the $\ell$-th partial wave propagator is given by a radial path integral
\begin{align}
    K_\ell (n,r'',r',t'',t') = \lim_{N \to \infty} \int \prod_{j=1}^N K_\ell(n,r_j, r_{j-1};t_j - t_{j-1}) \prod_{j=1}^{N-1} r_j^{n-1}dr_j,
\end{align}
which is now our object of study, and for which we need to specify a potential $V$ and where $K_\ell(r_j, r_{j-1};t_j - t_{j-1})=K_\ell(r_j, r_{j-1};\tau_j)$ is given by \eqref{eq:radialshorttimepropagatorANGLESOUT}. If we take
\begin{align} \label{eq:appendixVpotential}
    V(r) = \frac{1}{2}M\omega^2 r^2 + \frac{b\hbar^2}{2Mr^2},
\end{align}
with $b>-(n/2-1)^2$, then the approximation to the potential for a short time interval $\tau_j$ is given by
\begin{align}
    V_j = V(r_j) = \frac{1}{4}M\omega^2 (r_j^2+r_{j-1}^2) + \frac{b\hbar^2}{2Mr_j r_{j-1}}.
\end{align}
The fundamentally important idea is to now \textit{recombine} the exponential in \eqref{eq:radialshorttimepropagatorANGLESOUT} with the Bessel function $I_{\ell+\frac{n-2}{2}}$. To do this, one uses the so-called \textit{recombination techniques} from \cite[$\mathsection$2.5]{coherentPI}. The result we need is that if $\real(a z)>0$ and $|z| \sim O(\varepsilon^{-1})$ for $\varepsilon \ll 1$ and if $\nu \ge 0$ and $|\arg(\nu^2-2ac)| \le \frac{\pi}{2}$, then
\begin{align}
    I_\nu(az) e^{c/z} \rightleftharpoons \frac{1}{\sqrt{2 \pi az}} \exp\bigg[az-\frac{1}{2az}\bigg(\nu^2 - 2ac - \frac{1}{4}\bigg)\bigg] \rightleftharpoons I_{\lambda}(az),
\end{align}
where $a$ and $c$ are constants and $\lambda = \sqrt{\nu^2-2ac}$. Here $\rightleftharpoons$ signifies equivalence in path integration \cite[$\mathsection$2.1]{coherentPI}. Effectively this means that \textit{while} computing a path integral, we may replace the left-hand side of $\rightleftharpoons$ by the right-hand side. The nature of this recombination stems from the Edwards-Gulyaev asymptotic formula for the Bessel function \cite{EdwardsGulyaev}
\begin{align}
    I_\nu (z) \sim \frac{1}{\sqrt{2\pi z}} \exp\bigg[z- \frac{\nu^2-\frac{1}{4}}{2z}\bigg]
\end{align}
for $\real(z)$ positive and large. This allows us to write
\begin{align}
    \exp\bigg[-\frac{i b \hbar \tau_j}{2M r_j r_{j-1}}\bigg] I_{\ell+\frac{n-2}{2}} \bigg( \frac{M}{i \hbar \tau_j} r_j r_{j-1}\bigg) \rightleftharpoons I_{\lambda(n,\ell)} \bigg( \frac{M}{i \hbar \tau_j} r_j r_{j-1}\bigg)
\end{align}
where the recombination shift term $\lambda(n,\ell)$ is now given by
\begin{align}
    \lambda(n,\ell) = \bigg[\bigg(\ell+\frac{n-2}{2}\bigg)^2+b \bigg]^{1/2}.
\end{align}
The next step is to introduce the angular parameter $\varphi_j$ defined by $\sin \varphi_j = \omega \tau_j \ge 0$. We now use the MacLaurin expansion $\cos \varphi_j = \cos [\arcsin (\omega \tau_j)] = 1 - \frac{1}{2} \omega^2 \tau_j^2 + O(\tau_j^4)$ so that we can rewrite the propagator as
\begin{align} 
{K_\ell }(n,{r_j},{r_{j - 1}};{\tau _j}) &= {\pi ^{(3 - n)/2}}\frac{{\Gamma (\ell  + \frac{{n - 1}}{2})}}{{\Gamma (\ell  + 1)}}\frac{M \omega}{{i\hbar }}{\left( {{r_j}{r_{j - 1}}} \right)^{ - (n - 2)/2}} \csc(\varphi_j) \exp \left[ \frac{iM \omega}{{2\hbar }}\left( {r_j^2 + r_{j - 1}^2} \right) \cot(\varphi_j) \right] \nonumber \\
&\quad \times  I_{\lambda(n,\ell)}\left( {\frac{{M \omega}}{{i\hbar }}} {r_j}{r_{j - 1}} \csc \varphi_j \right) =: {\pi ^{(3 - n)/2}}\frac{{\Gamma (\ell  + \frac{{n - 1}}{2})}}{{\Gamma (\ell  + 1)}} {\hat K}_\ell (n,{r_j},{r_{j - 1}};{\tau _j})
\end{align}
In doing so, we have applied the rule that in path integration all contributions from the terms of order $O(\tau_j^{1+\varepsilon})$ with $\varepsilon>0$ can be ignored in the short time propagator. We now have to employ the techniques from Bessel integral and Weber formulas from Section \ref{sec:appendix2}. For $0 < \eta, \eta'< \infty$ and $0 < \varphi < \pi$ and $\real(\nu) > -1$ the $v$-function is defined as
\begin{align}
    v_\nu(\eta, \eta'; \varphi) = - i \csc \varphi \exp[i(\eta+\eta')\cot \varphi]I_\nu (-2i \sqrt{\eta \eta'} \csc \varphi).
\end{align}
As a function of $\eta$ we see that $v_\nu$ diverges at $\eta = 0$ unless $\real(\nu) \ge 0$. However, $\sqrt{\eta \eta'} v_\nu(\eta, \eta';\varphi)$ is regular for the entire range of $\eta$, i.e. $0 < \eta < \infty$ insofar as $\real(\nu) > -1$. One of the important properties that it satisfies is that
$
    \lim_{\varphi \to 0} v_\nu (\eta, \eta' ; \varphi) = \frac{1}{2} (\eta \eta')^{-1/4} \delta (\sqrt{\eta}-\sqrt{\eta'})
$.
If we use the Weber formula \eqref{eq:weberformula} in the form
\begin{align}
    \int_0^\infty \exp(i\alpha r^2) I_{\nu}(-i a r)I_{\nu}(-i b r) r dr = \frac{i}{2\alpha} \exp\bigg[-\frac{i}{4\alpha}(a^2+b^2)\bigg] I_\nu \bigg(-\frac{iab}{2\alpha}\bigg),
\end{align}
which is valid for $\real(\alpha)>0$ and $\real(\nu)>-1$, then we can very easily verify the composition formula
\begin{align}
    \int_0^\infty v_\nu (\eta'', \eta, \varphi'') v_\nu(\eta, \eta', \varphi') d \eta = v_\nu(\eta'',\eta',\varphi''+\varphi').
\end{align}
With an induction argument we can extend this to
\begin{align} \label{eq:multiextensionradialpropagator}
    v_\nu(\eta_N, \eta_0; \varphi) = \int_0^\infty \prod_{j=1}^N v_\nu (\eta_j, \eta_{j-1}; \varphi_j) \prod_{j=1}^N d\eta_j \quad \textnormal{where} \quad \varphi = \sum_{j=1}^N \varphi_j.
\end{align}
Let us make the identifications $\eta_j = \frac{M\omega}{2\hbar} r_j^2$ and $\varphi_j = \arcsin(\omega \tau_j)$
and write ${\hat K}_\ell(n, r_j, r_{j-1};\tau_j)$ as 
\begin{align}
    {\hat K}_\ell(n, r_j, r_{j-1};\tau_j) = \frac{M\omega}{\hbar}(r_j r_{j-1})^{-(n-2)/2} v_{\lambda(n,\ell)} (\eta_j, \eta_{j-1};\varphi_j).
\end{align}
This allows us to write
\begin{align} \label{eq:radialPropagatorProperty2}
    \int_0^\infty {\hat K}_\ell(n, r_{j+1}, r_{j};\tau_{j+1}) {\hat K}_\ell(n, r_j, r_{j-1};\tau_j) r_j^{n-1} dr_j &= \frac{M\omega}{\hbar} (r_{j+1}r_{j-1})^{(n-2)/2} v_{\lambda(n,\ell)} (\eta_{j+1}, \eta_{j-1};\varphi_{j+1}+\varphi_{j}) \nonumber \\
    &=  {\hat K}_\ell(n, r_{j+1}, r_{j-1};\tau_{j+1}+\tau_{j}).
\end{align}
Extending this by the use of \eqref{eq:multiextensionradialpropagator} we see that
\begin{align}
    {\hat K}_\ell(n,r_N, r_0; \tau) &= \frac{M\omega}{\hbar} (r_0 r_N)^{-(n-2)/2} v_{\lambda(n,\ell)}(\eta_N, \eta_0; \varphi) \nonumber \\
    &= \frac{M\omega}{i\hbar} (r_0 r_N)^{-(n-2)/2} (\csc \varphi) \exp[i(\eta_0+\eta_N)\cot \phi] I_{\lambda(n,\ell)} (-2i \sqrt{\eta_0 \eta_N} \csc \varphi),
\end{align}
where $\tau$ and $\varphi$ are given by the sums
\begin{align}
    \tau = \sum_{j=1}^N \tau_j \quad \textnormal{and} \quad \varphi = \sum_{j=1}^N \arcsin(\omega \tau_j).
\end{align}
The last detail is that in the limit as $N \to \infty$, which means $\tau_j \to 0$, we obtain that $\arcsin(\omega \tau_j) \to \omega \tau_j$, and therefore $\varphi$ reaches the unique limit $\varphi = \omega \tau$. The final result is
\begin{align}
    {\hat K}_\ell (n,r'',r',\tau) &= \frac{M\omega}{i\hbar} (r''r')^{-(n-2)/2}\csc(\omega \tau)  \exp\bigg(\frac{iM\omega}{2\hbar}(r''^2+r'^2)\cot(\omega \tau)\bigg) I_{\lambda(n,\ell)} \bigg(\frac{M\omega}{i\hbar}r''r'\csc(\omega\tau)\bigg).
\end{align}
If $n \to 1$ or $n \to 3$, then 
\begin{align} 
    {\hat K}_\ell (1,r'',r',\tau) &= \frac{M\omega \csc(\omega \tau)\sqrt{r''r'}}{i\hbar }  \exp\bigg(\frac{iM\omega}{2\hbar}(r''^2+r'^2)\cot(\omega \tau)\bigg) I_{[(\ell-\frac{1}{2})^2+b]^{1/2}} \bigg(\frac{M\omega}{i\hbar}r''r'\csc(\omega\tau)\bigg), \nonumber \\
    {\hat K}_\ell (3,r'',r',\tau) &= \frac{M\omega \csc(\omega \tau)}{i\hbar \sqrt{r''r'}}  \exp\bigg(\frac{iM\omega}{2\hbar}(r''^2+r'^2)\cot(\omega \tau)\bigg)  I_{[(\ell+\frac{1}{2})^2+b]^{1/2}} \bigg(\frac{M\omega}{i\hbar}r''r'\csc(\omega\tau)\bigg),
\end{align}
which are of special interest.

\subsection{Future directions of path integrals for quantum probability loading} \label{sec:appendix4}

The purpose of having introduced the above path integral involving polar coordinates is to use this technique as a stepping stone for potentials of the form
\begin{align}
    V(a,r)= \frac{1}{2}M\omega^2 r^2 + \frac{b\hbar \cos(ar^2)}{2Mr^2}
\end{align}
which in turn could lend themselves to potentials such as \eqref{eq:lognormalpotential}. We note that $V(0,r)=V(r)$ from \eqref{eq:appendixVpotential}. Following the procedures in \cite{inomataKayed2, inomataKayed1, coherentPI, goovaerts1975radial}, the short time interval potential in polar coordinates should be of the form
\begin{align}
    V(a,r_j)= \frac{1}{2}M\omega^2 (r_j^2+r_{j-1}^2) + \frac{b\hbar \cos(ar_j)\cos(ar_{j-1})}{2Mr_jr_{j-1}}.
\end{align}
We also note that the strategy followed by Duru in \cite{duru1985radial} could also be useful to deal with other potentials of the form $V(r) = \frac{f(r)}{r^2}$ for a suitable function $f$. If the techniques were adaptable, then we could attemp to find, or approximate, potentials of the from $\frac{\log r}{r^2}$ and $\frac{\log^2 r}{r^2}$ that correspond to the lognormal distribution as described in Section \ref{sec:inverselist}.


\begin{thebibliography}{1}

\bibitem{Note1}
In our context, we shall work with the Copenhagen interpretation of a spinless
  particle and the Schrodinger picture.

\bibitem{Note2}
This is sometimes referred to as the change of variables $t \to -it$, but this
  notion is misleading as it hides away the subtleties of the process of
  analytic continuation. The setup in $t \to -it$ is known as `Euclidean time'.

\bibitem{Note3}
Note that the Lagrangian is a classical scalar \protect \textit {function}, not
  an operator, and its form is $T-V$, where $T$ represents kinetic energy,
  rather than $T+V$ which corresponds to a classical Hamiltonian.

\end{thebibliography}


\begin{thebibliography}{10}
{
        \bibitem{feynmankac}
		H.~Alghassi, A.~Deshmukh, N.~Ibrahim, N.~Robles, S.~Woerner, and C.~Zoufal.
		\newblock{\em A variational quantum algorithm for the Feynman-Kac formula}.
		\newblock{Quantum}, (6): 730, 2022.

        \bibitem{accelerated}
		D.~An, N.~Linden, J.~Liu, A.~Montanaro, C.~Shao, and J.~Wang.
		\newblock{\em Quantum-accelerated multilevel Monte Carlo methods for stochastic differential equations in mathematical finance}.
		\newblock{Quantum}, (5): 481, 2021.

        \bibitem{61}
		A.~Ambainis and J.~Emerson.
		\newblock{\em Quantum t-designs: t-wise Independence in the Quantum World}.
		\newblock{Twenty-Second Annual IEEE Conference on Computational Complexity (CCC’07)}, 2007.

       \bibitem{atia2017fast}
       Y.~Atia and D.~Aharonov. 
       \newblock{\em Fast-forwarding of Hamiltonians and exponentially precise measurements}. 
       \newblock{Nature communications, 8:1, 1572, 2017}.

        \bibitem{higher1}
         H.~Bateman.
        \newblock {\em Higher Transcendental Functions, Volume \textnormal{I}}.
        \newblock {McGraw-Hill}, 1953
    
        \bibitem{higher2}
        H.~Bateman.
        \newblock {\em Higher Transcendental Functions, Volume \textnormal{II}}.
        \newblock {McGraw-Hill}, 1953.

       \bibitem{BHMT00} 
        G.~Brassard, P.~Hoyer, M.~Mosca, and A.~Tapp.
        \newblock{\em Quantum amplitude amplification and estimation}.
        \newblock{Contemporary Mathematics, 305, 53--74, 2002}. 

        \bibitem{cprv23}
		S.~Certo, A.~Pham, N.~Robles, and A.~Vlasic.
		\newblock {\em Conditional generative models for learning stochastic processes}.
		\newblock Quantum Machine Intelligence, volume 5, Article number: 42 (2023).

        \bibitem{chakrabarti2021threshold}
        S.~Chakrabarti, R.~Krishnakumar, G.~Mazzola, N.~Stamatopoulos, S.~Woerner, and W.~Zeng.
        \newblock{\em A threshold for quantum advantage in derivative pricing}. 
        \newblock{Quantum, 5:463, 2021}. 

       \bibitem{susyQM}
		F.~Cooper, A.~Khare, and U.~Sukhatme.
		\newblock{\em Supersymmetry in quantum mechanics}.
		\newblock{World Scientific}, 2001.

        \bibitem{curtright2021lie} 
        T.~L.~Curtright, Z.~Cao, A.~Peca, D.~Sarker, and B.~D.~Shrestha.
        \newblock{\em Lie Groups and Propagators Exemplified}.
        \newblock{arXiv:2112.14401, 2021}.

        \bibitem{dasguptaPaine}
		K.~Dasgupta and B.~Paine.
		\newblock {\em Loading Probability Distributions in a Quantum circuit}.
		\newblock arXiv:2208.13372, 2022.

        \bibitem{duru1985radial}
		I.~H.~Duru.
		\newblock {\em On the path integral for the potential $V=ar^{-2}+br^2$}.
		\newblock Physics Letters, volume 112A, number 9, 18 November 1985.  

        \bibitem{EdwardsGulyaev}
		S.~F.~Edwards and Y.~V.~Gulyaev.
		\newblock {\em Path Integrals in Polar Co-ordinates}.
		\newblock Proceedings of the Royal Society of London. Series A, Mathematical and Physical Sciences, Vol. 279, No. 1377 (May 26, 1964), pp. 229-235.

        \bibitem{general}
		S.~Endo, J.~Sun, Y.~Li, S.~C.~Benjamin, and X.~Yuan.
		\newblock{\em Variational Quantum Simulation of General Processes}.
		\newblock{Phys. Rev. Lett.}, (125): 010501, 2020.

        \bibitem{montaldiglasser}
		M.~L.~Glasser and E.~Montaldi.
		\newblock {\em Some Integrals Involving Bessel Functions}.
		\newblock {Journal of Mathematical Analysis and Applications}, (183): 577-590, 1994.

        \bibitem{goovaerts1975radial}
		M.~J.~Goovaerts.
		\newblock {\em Path-integral evaluation of a nonstationary Calogero model}.
		\newblock Journal of Mathematical Physics, Vol. 16, No. 3, March 1975.  

        \bibitem{gradryz}
        I.~S.~Gradshteyn and I.~M.~Ryzhik
        \newblock {\em Table of integrals, series, and products, Seventh Edition}.
        \newblock {Elsevier}, 2007.

        \bibitem{griffiths}
		D.~Griffiths.
		\newblock{\em Introduction to quantum mechanics}.
		\newblock{Prentice Hall}, 1995.

        \bibitem{grover2002creating} 
        L.~Grover and T.~Rudolph.
        \newblock{\em Creating superpositions that correspond to efficiently integrable probability distributions}. 
        \newblock{Preprint quant-ph/0208112, 2002}
        
        \bibitem{handbookPI}
		C.~Grosche and F.~Steiner.
		\newblock{\em Handbook of Feynman Path Integrals}.
		\newblock{Springer Tracts in Modern Mathematics}, 145, 1998.

        \bibitem{groscheSteiner1987}
		C.~Grosche and F.~Steiner.
		\newblock {\em Path integrals on curved manifolds}.
		\newblock {Z. Phys. C - Particles and Fields}, 36, 699 714 (1987).

        \bibitem{gu2021fast} 
        S.~Gu, R.~D.~Somma, and B.~\c{S}ahino{\u{g}}lu.
        \newblock{\em Fast-forwarding quantum evolution}. 
        \newblock{Quantum, 5:577, 2021}. 

        \bibitem{60}
		A.~Hayashi, T.~Hashimoto, and M.~Horibe.
		\newblock{\em Reexamination of optimal quantum state estimation of pure states}.
		\newblock{Physical Review A}, 72(3): 2005.

        \bibitem{reckoning}
		A. Inomata.
		\newblock {\em Reckoning of the Besselian path integral}.
		\newblock {On Klauder's Path: A Field Trip, pp. 99-106 (1994)}, World Scientific, 1994.

        \bibitem{inomataKayed2}
		A.~Inomata and M.~A.~Kayed.
		\newblock {\em Path integral on $S^2$: The Rosen-Morse oscillator}.
		\newblock J. Phys. A: Math. Gen., 18 (1985).   

        \bibitem{inomataKayed1}
		A.~Inomata and M.~A.~ Kayed.
		\newblock {\em Path integral quantization of the symmetric P\"{o}schl-Teller potential}.
		\newblock Physics Letters, volume 108A, number 1, 11 March 1985.   

        \bibitem{inomatajunkerLie}
		A.~Inomata and G.~Junker.
		\newblock {\em Path integrals and Lie Groups}.
		\newblock Noncompact Lie Groups and Some of Their Applications, NATO ASI Series, Kluwer Academic Publishers, Dordrecht, 1994, p. 199.

        \bibitem{coherentPI}
		A.~Inomata, H.~Kuratsuji, and C.~C.~Gerry.
		\newblock {\em Path integrals and coherent states of $SU(2)$ and $SU(1,1)$}.
		\newblock World Scientific, 1992.
        
        \bibitem{Kitaev95} 
        A.~Kitaev.
        \newblock{\em Quantum measurements and the Abelian stabilizer problem}.
        \newblock{Peprint quant-ph/9511026, 1995}. 

        \bibitem{kleinertPI}
		H.~Kleinert.
		\newblock{\em Path Integrals in Quantum Mechanics, Statistics, Polymer Physics, and Financial Markets}. 5th Edition.
		\newblock{World Scientific}, 2009.

        \bibitem{pricingmultiasset}
		K.~Kubo, K.~Miyamoto, K.~Mitarai, and K.~Fujii.
		\newblock{\em Pricing multi-asset derivatives by variational quantum algorithm}.
        \newblock{Preprint arXiv:2207.01277, 2022.}
          
        \bibitem{protocols}
		O.~Kyriienko, A.~E.~Paine, and V.~E.~Elfving.
		\newblock{\em Protocols for Trainable and Differentiable Quantum Generative Modelling}.
        \newblock{Preprint arXiv:2202.08253, 202.2}

        \bibitem{wellsfargo2}
		V.~Markov, C.~Stefanski, A.~Rao, and C.~Gonciulea
		\newblock {\em A Generalized Quantum Inner Product and Applications to Financial Engineering}.
        \newblock{Preprint arXiv:2201.09845, 2022.}

        \bibitem{mcardle2022quantum}
        S.~McArdle, A.~Gily{\'e}n, G.~Mazzola, and B.~Mario.
        \newblock{\em Quantum state preparation without coherent arithmetic}, 
        \newblock{Preprint arXiv:2210.14892, 2022.}

        \bibitem{variationalansatz}
		S.~McArdle, T.~Jones, S.~Endo, Y.~Li, S.~C.~Benjamin, and X.~Yuan.
		\newblock{\em Variational ansatz based quantum simulation of imaginary time evolution}.
		\newblock{npj Quantum Information}, (5): 75, 2019.

        \bibitem{62}
		J.~McClean, S.~Boixo, V.~N.~Smelyanskiy, R.~Babbush, and H.~Neven.
		\newblock{\em Barren plateaus in quantum neural network training landscapes}.
		\newblock{Nature Communications}, 9, 2018.

        \bibitem{McLachlan}
		A.~D.~McLachlan.
		\newblock{\em A variational solution of the time-dependent schrodinger equation}.
		\newblock{Molecular Physics}, 8(1): 39–44, 1964.

        \bibitem{M15}
        A.~Montanaro.
        \newblock{\em Quantum speedup of Monte Carlo methods}.
        \newblock{Proceedings of the Royal Society A: Mathematical, Physical and Engineering Sciences, Volume:471-2181, 2015}.

        \bibitem{nielsenchuang}
		M.~A.~Nielsen and I.~L.~Chuang.
		\newblock{\em Quantum Computation and Quantum Information}.
		\newblock{Cambridge University Press}, 10th Anniversary Edition, 2016.     

        \bibitem{oberhettinger}
		F.~Oberhettinger.
		\newblock{\em Tables of Mellin Transforms}.
		\newblock{Springer-Verlag}, 1974.

        \bibitem{peakinomata}
		D.~Peak and A.~Inomata.
		\newblock {\em Summation over Feynman Histories in Polar Coordinates}.
		\newblock Journal of Mathematical Physics, volume 10, number 8, August 1969.

        \bibitem{plesch} 
        M.~Plesch and C.~Brukner. 
        \newblock{\em Quantum-state preparation with universal gate decompositions}. 
        \newblock{Physical Review A}, 83(3): 2011.  
  

        \bibitem{rattew2021efficient} 
        A.~G.~Rattew, Y.~Sun, P.~Minssen, and M.~Pistoia.
        \newblock{\em The efficient preparation of normal distributions in quantum registers}.
        \newblock{Quantum, 5:609, 2021}. 

        \bibitem{xanaduOption} 
        P.~Rebentrost, B.~Gupt, and T.~R.~Bromley.
        \newblock{\em Quantum computational finance: Monte Carlo pricing of financial derivatives}.
        \newblock{Phys. Rev. A 98, 022321 (2018)}. 

        \bibitem{lognormal1}
		P.~E.~Sartwell.
		\newblock {\em The distribution of incubation periods of infectious disease}.
		\newblock American journal of hygiene 51 (1950): 310-318.

        \bibitem{lognormal2}
		G.~Scheler.
		\newblock {\em Diversity and stability in neuronal output rates}.
		\newblock 36th Society for Neuroscience Meeting, Atlanta, 2006.

        \bibitem{axioms}
		M.~D.~Schmidt.
		\newblock {\em A Short Note on Integral Transformations and Conversion Formulas for Sequence Generating Functions}.
		\newblock Axioms 2019, 8(2), 62.

        \bibitem{capgemini}
		M.~Sharma K.N., C.~de~Valk, A.~Raina, and J.~van~Velzen.
		\newblock {\em Quantum state preparation for bell-shaped probability distributions using deconvolution methods}.
        \newblock{Preprint arXiv:2310.05044, 2023.}

        \bibitem{oeis}
		N.~J.~A.~Sloane.
		\newblock {\em The On-Line Encyclopedia of Integer Sequences}.
		\newblock \texttt{https://oeis.org/}.  

        \bibitem{wellsfargo1}
		C.~Stefanski, V.~Markov, and C.~Gonciulea
		\newblock {\em Quantum Amplitude Interpolation}.
        \newblock{Preprint arXiv:2203.08758, 2022.}

        \bibitem{su2021fast}
        Y.~Su. 
        \newblock{\em Fast-Forwardable Quantum Evolution and Where to Find Them}.
        \newblock{Quantum Views, 5:62, 2021}. 

        \bibitem{titchmarshFourier}
		E.~C.~Titchmarsh
		\newblock{\em Introduction to the theory of Fourier integrals, Second Edition}.
		\newblock{Oxford University Press}, 1948.

        \bibitem{watson}
        G.~N.~Watson.
        \newblock {\em A Treatise of the Theory of the Bessel Functions}.
        \newblock {Cambridge University Press}, 1958.

        \bibitem{variational}
		X.~Yuan, S.~Endo, Q.~Zhao, Y.~Li, and S.~C.~Benjamin.
		\newblock{\em Theory of variational quantum simulation}.
		\newblock{Quantum}, (3): 191, 2019.

         \bibitem{quantumPDEs}
		T.~Zhao, C.~Sun, A.~Cohen, J.~Stokes, and S.~Veerapaneni.
		\newblock{\em Quantum-inspired variational algorithms for partial differential equations: application to financial derivative pricing}.
		\newblock{Quantitative Finance}, September 2023.

        \bibitem{zoufal_error_bounds}
		C.~Zoufal, D.~Sutter, and S.~Woerner.
		\newblock{\em Error Bounds for Variational Quantum Time Evolution}.
       \newblock{Preprint arXiv:2108.00022, 2021.}

      

      \bibitem{iaconis2023quantum}
      J.~Iaconis and S.~Johri and E.~Zhu.
      \newblock{\em Quantum State Preparation of Normal Distributions using Matrix Product States}. 
      \newblock{Preprint arXiv:2303.01562, 2023.}

 
}

\end{thebibliography}
\end{document}